\PassOptionsToPackage{nosumlimits,nonamelimits}{amsmath}
\documentclass[acmsmall]{acmart}
\numberwithin{equation}{section}

\usepackage{ifthen}
\newboolean{arxiv}
\setboolean{arxiv}{true} 
\newcommand{\ifarx}[2]{\ifthenelse{\boolean{arxiv}}{#1}{#2}}
\ifarx{
\settopmatter{printacmref=false}{} 
\renewcommand\footnotetextcopyrightpermission[1]{}} 

\RequirePackage{quiver}

\usepackage{stel-common}

\usepackage{booktabs}   
\usepackage{subcaption} 

\usepackage{stel-common}

\usepackage{bm} 

\newcommand{\mybar}[3]{%
	\mathrlap{\hspace{#2}\overline{\scalebox{#1}[1]{\phantom{\ensuremath{#3}}}}}\ensuremath{#3}
}

\newcommand{\barB}{\mybar{0.6}{2.5pt}{B}}

\providecommand{\catname}{\mathbf}
\providecommand{\clsname}{\mathcal}
\providecommand{\oname}[1]{{\mathop{\mathsf{#1}}\xspace}}

\def\defcatname#1{\expandafter\def\csname B#1\endcsname{\catname{#1}}}
\def\defcatnames#1{\ifx#1\defcatnames\else\defcatname#1\expandafter\defcatnames\fi}
\defcatnames ABCDEFGHIJKLMNOPQRSTUVWXYZ\defcatnames

\def\defclsname#1{\expandafter\def\csname C#1\endcsname{\clsname{#1}}}
\def\defclsnames#1{\ifx#1\defclsnames\else\defclsname#1\expandafter\defclsnames\fi}
\defclsnames ABCDEFGHIJKLMNOPQRSTUVWXYZ\defclsnames

\def\defbbname#1{\expandafter\def\csname BB#1\endcsname{{\mathbb{#1}}}}
\def\defbbnames#1{\ifx#1\defbbnames\else\defbbname#1\expandafter\defbbnames\fi}
\defbbnames ABCDEFGHIJKLMNOPQRSTUVWXYZ\defbbnames

\def\Set{\catname{Set}}

\DeclareOldFontCommand{\bf}{\normalfont\bfseries}{\mathbf}

\providecommand{\PSet}{{\mathcal P}}			                 
\providecommand{\Id}{\operatorname{Id}}

\providecommand{\id}{\mathsf{id}}
\providecommand{\op}{\mathsf{op}}
\providecommand{\comp}{\mathbin{\circ}}

\providecommand{\xto}[1]{\,\xrightarrow{#1}\,}

\providecommand{\dar}{\kern-1.2pt\operatorname{\downarrow}}
\providecommand{\uar}{\kern-1.2pt\operatorname{\uparrow}}

\providecommand{\brks}[1]{\langle #1\rangle}

\providecommand{\inl}{\oname{inl}}
\providecommand{\inr}{\oname{inr}}
\providecommand{\inj}{\oname{in}}

\DeclareSymbolFont{Symbols}{OMS}{cmsy}{m}{n}
\DeclareMathSymbol{\iobj}{\mathord}{Symbols}{"3B}

\providecommand{\ev}{\oname{ev}}

\usepackage{stmaryrd}

\providecommand{\lsem}{\llbracket}
\providecommand{\rsem}{\rrbracket}
\providecommand{\sem}[1]{\lsem #1 \rsem}

\providecommand{\pacman}[1]{}					                     

\newcommand{\undefine}[1]{\let #1\relax}					                       

\providecommand{\mone}{{\text{\kern.5pt\rmfamily-}\mathsf{\kern-.5pt1}}}

\makeatletter
\def\mfix#1{\oname{#1}\@ifnextchar\bgroup\@mfix{}}	       
\def\@mfix#1{#1\@ifnextchar\bgroup\mfix{}}			           
\makeatother

\providecommand{\case}[3]{\mfix{case}{\mathbin{}#1}{of}{#2}{\kern-1pt;}{\mathbin{}#3}}

\usepackage{dsfont}

\DeclareMathSymbol{\mathinvertedexclamationmark}{\mathord}{operators}{'074}
\DeclareMathSymbol{\mathexclamationmark}{\mathord}{operators}{'041}
\makeatletter
\newcommand{\raisedmathinvertedexclamationmark}{%
	\mathord{\mathpalette\raised@mathinvertedexclamationmark\relax}%
}
\newcommand{\raised@mathinvertedexclamationmark}[2]{%
	\raisebox{\depth}{$\m@th#1\mathinvertedexclamationmark$}%
}
\makeatother

\newcommand{\product}{\times}

\newcommand{\SKI}{\mathrm{SKI}}

\newcommand{\Powf}{\mathcal{P}_{\omega}}

\newcommand{\Sigmas}{\Sigma^{\star}}

\newcommand{\ar}{\mathsf{ar}}

\newcommand{\seq}{\subseteq}

\newcommand{\beh}{{\mathsf{beh}}}

\providecommand{\C}{}

\newcommand{\B}{{\mathbb{B}}}
\renewcommand{\C}{{\mathbb{C}}}

\renewcommand{\id}{{\mathsf{id}}}
\newcommand{\Nat}{\mathds{N}}

\newcommand{\f}{\oname{f}}

\newcommand{\takeout}[1]{\empty}

\newcommand{\ini}{\iota}

\newcommand{\wh}{\widehat}

\renewcommand{\rho}{\varrho}
\newcommand{\N}{\mathds{N}}
\newcommand{\opp}{\mathsf{op}}

\makeatletter
\newsavebox{\@brx}
\providecommand{\llangle}[1][]{\savebox{\@brx}{\(\m@th{#1\langle}\)}%
	\mathopen{\copy\@brx\kern-0.5\wd\@brx\usebox{\@brx}}}
\providecommand{\rrangle}[1][]{\savebox{\@brx}{\(\m@th{#1\rangle}\)}%
	\mathclose{\copy\@brx\kern-0.5\wd\@brx\usebox{\@brx}}}
\makeatother

\renewcommand{\comp}{\cdot}
\renewcommand{\c}{\colon}

\newcommand{\xra}[1]{\mathrel{\raisebox{-1.15pt}{$\xrightarrow{\;\smash{\raisebox{2.5pt}{\makebox(3,0)[t]{\scriptsize $#1$}}\;}}$}}}
\renewcommand{\xto}{\xra}

\renewcommand{\Nat}{\mathbb{N}}

\newcommand{\fset}{{\mathbb{F}}}
\newcommand{\vcat}{{\Set}^{\fset}}

\newcommand{\mS}{{\mu\Sigma}}

\renewcommand{\epsilon}{\varepsilon}

\newcommand*\xbar[1]{%
	\kern.2em\hbox{%
		\vbox{%
			\hrule height 0.5pt 
			\kern0.5ex
			\hbox{%
				\kern-0.2em
				\ensuremath{#1}%
				\kern-0.4em
			}%
		}%
	}\kern.4em %
}

\makeatletter
\newcommand{\monto}{\@ifstar{\@mtolifted}{\@mto}}
\newcommand{\@mto}{\multimapdot}
\newcommand{\@mtolifted}{\mathbin{\xbar{\multimapdot}}}
\makeatother

\newcommand{\app}[2]{\,}

\makeatletter
	\newcommand{\pushright}[1]{\ifmeasuring@#1\else\omit\hfill$\displaystyle#1$\fi\ignorespaces}
	\newcommand{\pushleft}[1]{\ifmeasuring@#1\else\omit$\displaystyle#1$\hfill\fi\ignorespaces}
\makeatother

\usepackage[capitalize,noabbrev,nameinlink]{cleveref} 

\usepackage[inline]{enumitem}
\setlist[enumerate,1]{label=(\arabic*),font=\normalfont,align=left,leftmargin=0pt,labelindent=0pt,listparindent=\parindent,labelwidth=0pt,itemindent=!,topsep=2pt,parsep=0pt,itemsep=2pt,start=1}
\setlist[enumerate,2]{label=(\alph*),font=\normalfont,labelindent=*,leftmargin=*,start=1}
\setlist[itemize]{labelindent=*,leftmargin=*}
\setlist[description]{labelindent=*,leftmargin=*,itemindent=-1 em}

\usepackage{microtype}

\usepackage{xspace}

\usepackage{ifdraft}

\tikzstyle{shiftarr}=[
        rounded corners,%
        to path={--([#1]\tikztostart.center)
                     -- ([#1]\tikztotarget.center) \tikztonodes
                     -- (\tikztotarget)},
]

\usepackage{tikz-cd}
\usetikzlibrary{arrows.meta}

\tikzcdset{scale cd/.style={every label/.append style={scale=#1},
  cells={nodes={scale=#1}}},
  arrow style=tikz,
  diagrams={>=stealth},
  row sep=large,
  column sep=huge
}

\usepackage[footnote,marginclue,nomargin]{fixme}
\FXRegisterAuthor{sg}{asg}{SG} 
\FXRegisterAuthor{st}{ast}{ST} 
\FXRegisterAuthor{hu}{ahu}{HU} 
\FXRegisterAuthor{sv}{asv}{SV} 

\RequirePackage[utf8]{inputenc}

\setcopyright{cc}
\setcctype{by}
\acmDOI{10.1145/3828693}
\acmYear{2026}
\acmJournal{PACMPL}
\acmVolume{10}
\acmNumber{ICFP}
\acmArticle{295}
\acmMonth{8}
\acmSubmissionID{icfp26main-p75-p}
\received{2026-03-02}
\received[accepted]{2026-05-13}

\begin{document}\allowdisplaybreaks

\title{Towards a Higher-Order Bialgebraic Denotational Semantics}

\author{Sergey Goncharov}
\email{s.goncharov@bham.ac.uk}
\orcid{0000-0001-6924-8766}
\affiliation{%
	\institution{University of Birmingham}
	\city{Birmingham}
	\country{United Kingdom}
}

\author{Marco Peressotti}
\email{peressotti@imada.sdu.dk}
\orcid{0000-0002-0243-0480}
\affiliation{%
	\institution{University of Southern Denmark}
	\city{Odense}
	\country{Denmark}
}

\author{Stelios Tsampas}
\email{stelios@imada.sdu.dk}
\orcid{0000-0001-8981-2328}
\affiliation{%
	\institution{University of Southern Denmark}
	\city{Odense}
	\country{Denmark}
}

\author{Henning Urbat}
\email{henning.urbat@fau.de}
\orcid{0000-0002-3265-7168}
\affiliation{%
	\institution{Friedrich-Alexander-Universität Erlangen-Nürnberg}
	\city{Erlangen}
	\country{Germany}
}

\author{Stefano Volpe}
\email{stefano@imada.sdu.dk}
\orcid{0009-0003-5277-1958}
\affiliation{%
	\institution{University of Southern Denmark}
	\city{Odense}
	\country{Denmark}
}


\begin{abstract}
	The bialgebraic \emph{abstract GSOS} framework by Turi and Plotkin provides an elegant
	categorical approach to modelling the operational and denotational semantics
	of programming and process languages. In abstract GSOS, bisimilarity is always a
	congruence, and it coincides with denotational equivalence. This saves the language designer
	from intricate, ad-hoc reasoning to establish these properties. The bialgebraic perspective on operational semantics in the style of abstract GSOS
	has recently been extended to
	higher-order languages, preserving compositionality of bisimilarity. However, a
	categorical understanding of  bialgebraic \emph{denotational semantics}
	according to Turi and Plotkin's original vision has so far been missing in the higher-order setting. In the
	present paper, we develop a theory of adequate denotational semantics in higher-order abstract GSOS. The denotational models are parametric in
	an appropriately chosen semantic domain in the form of a \emph{locally final coalgebra} for a behaviour bifunctor, whose construction is fully decoupled from the syntax of
	the language. Our approach captures existing accounts of denotational semantics such as semantic domains
	built via general step-indexing, previously introduced on a per-language basis, and is shown to be applicable to a wide range of different higher-order languages, e.g.~simply typed and untyped languages, or languages with computational effects such as probabilistic or non-deterministic branching.
\end{abstract}




\maketitle

\newcommand{\cal}{\mathcal}
\newcommand{\CBUlt}{\catname{CBUlt}}
\newcommand{\CBUltinh}{\CBUlt_\catname{\text{inh}}}
\newcommand{\Setinh}{\Set_\catname{\text{inh}}}
\def\HOB{\rho\catname{-HOB}}
\newcommand{\paral}{\parallel} \newcommand{%
	\paralv}{\mathrel{\mathchoice{\PARV}{\PARV}{\scriptsize\PARV}{\tiny\PARV} }}
\def\PARV{{\setbox0\hbox{$
					\parallel$}\rlap{\hbox to \wd0{\hss\rule[-.5ex]{.75em}{.12ex}\hss}}\box0 }}%

\newcommand{\Tr}{\mathsf{Tr}} 
\newcommand{\Trl}{\Lambda}
\newcommand{\Trv}{\Tr_{\mathsf{v}}}
\newcommand{\Trn}{\Tr_{\bar{\mathsf{v}}}}
\newcommand{\xcl}{\textbf{xCL}\xspace}
\newcommand{\xclfg}{$\textbf{xCL}_{\textbf{fg}}$\xspace}
\newcommand{\synt}{\Sigma}
\newcommand{\stscr}{\ensuremath{\boldsymbol\mu}\textbf{TCL}\xspace}
\newcommand{\fpc}{\textbf{FPC}\xspace}
\newcommand{\stsc}{\textbf{xTCL}\xspace}
\newcommand{\stlc}{\textbf{STLC}\xspace}
\newcommand{\pcf}{\textbf{PCF}\xspace}
\newcommand{\utype}{\mathsf{unit}}
\newcommand{\booltype}{\mathsf{bool}}
\newcommand{\voidtype}{\mathsf{void}}
\newcommand{\unittype}{\mathsf{unit}}
\newcommand{\nattype}{\mathsf{nat}}
\newcommand{\Ty}{\mathsf{Ty}}
\newcommand{\Tyl}{\mathsf{Ty}}
\newcommand{\Ops}{\mathsf{Ops}}
\newcommand{\arty}[2]{#1 \rightarrowtriangle #2} 
\newcommand{\cty}[2]{#1 \rightsquigarrow #2} 

\newcommand{\den}[1]{\llbracket #1 \rrbracket}

\newcommand{\ndc}{\mathsf{nc}}

\newcommand{\later}{\triangleright}
\newcommand{\next}{\mathtt{next}}
\newcommand{\topos}{\Set^{\omega^\opp}}
\newcommand{\vartopos}{(\vcat)^{\omega^{\opp}}}

\newcommand{\xnccl}{\textbf{xNCCL}\xspace}
\newcommand{\xPTCL}{\textbf{xPTCL}\xspace}

\newcommand{\prob}{\mathsf{p}}
\newcommand{\untyped}{\mathsf{u}}

\newcommand{\xCL}{\textbf{xCL}\xspace}
\newcommand{\ndet}{\oplus} 
\newcommand{\infer}[2]{\frac{#2}{#1}}

\newcommand{\val}{\mathsf{v}}
\newcommand{\com}{\mathsf{c}}

\newcommand{\tcomp}{\textsf{app}}
\newcommand{\Dar}{\Downarrow}

\newcommand{\Df}{\mathcal{D}_\omega}




\theoremstyle{definition}
\newtheorem{notation}[theorem]{Notation}
\newtheorem{assumption}[theorem]{Assumption}
\newtheorem{construction}[theorem]{Construction}

\section{Introduction}

One of the core areas of the theory of programming languages is the study of
\emph{formal semantics}~\cite{winskel1993formal}, which aims to assign a precise
mathematical ``meaning'' to every well-formed program of a given language. Such a foundation is a prerequisite for proving properties of individual programs (such as functional correctness or termination) or of the language itself (such as type or memory safety), and for formalizing and automating such reasoning in proof assistants. Two widely used approaches to formal semantics are \emph{operational semantics}, which specifies how a program $p$ behaves when it is executed, and \emph{denotational semantics}, which associates to a program~$p$ some mathematical object $\den{p}$ (e.g.~an element of an order-theoretic domain or a morphism of a category) representing its abstract content. In many settings, both types of semantics are available, and studying their connections provides deep insights into the nature of the programming language at hand. Two key desiderata for a denotational semantics to be useful and well-behaved are:
\begin{enumerate}
	\item \emph{Compositionality.} The denotation $\den{p}$ of a composite program $p=\f(p_1,\ldots,p_n)$ only depends on the program constructor $\f$ and the denotations $\den{p_i}$ of the components $p_i$.
	\item \emph{Adequacy.} If $p$ and $q$ are \emph{denotationally equivalent},
	      i.e. $\den{p}=\den{q}$, then they are \emph{behaviourally equivalent}: their observable behaviour according to the \emph{operational} semantics is indistinguishable.
\end{enumerate}
Compositionality allows for modular reasoning, while adequacy guarantees that
the denotational semantics is compatible with the (often simpler and more
intuitive) operational semantics. These properties are typically
proven on an ad hoc, per-language basis, and tend to be cumbersome. For this reason, there has been longstanding interest in general \emph{categorical} accounts of operational and denotational semantics, which mitigate these complexities via the power of abstraction.

\paragraph*{Abstract GSOS} In their seminal paper, \citet{Turi97} proposed to
approach operational and denotational semantics via a \emph{bialgebraic}
framework. It rests on the observation that the operational rules of many
programming or process languages, given in the style of \emph{structural
	operational semantics}~\cite{Plotkin1981SOS}, can be presented as a \emph{GSOS
	law}, a natural transformation of type
\[ \rho_X\colon \Sigma(X\times BX)\to B\Sigmas X \]
parametric in a \emph{syntax functor} $\Sigma\colon\C\to \C$ (with free monad
$\Sigmas$) modelling the constructors of the language, and a \emph{behaviour
	functor} $B\colon \C\to\C$ modelling the type of transitions that programs can
make during their execution. For example, to model a simple process algebra, one
takes $\Sigma$ to be a polynomial set functor representing the process
constructors, like $\mathsf{nil}$ or $-\parallel-$, and $B$ to be the functor $BX=(\Powf X)^L$ representing the behaviour of image-finite labelled transition systems.

Every GSOS law canonically induces an \emph{operational model}, which emerges by equipping the initial algebra $\Sigma(\mS)\xto{\cong}\mS$ (the object of program terms) with the structure of a \emph{coalgebra} $\gamma\colon \mS\to B(\mS)$. Intuitively, $\gamma$ is the transition system that runs programs according to the operational rules encoded by the given law. Dually, there is a canonical \emph{denotational model} obtained by equipping the final coalgebra $\nu B\xto{\cong} B(\nu B)$ (the domain of abstract program behaviours) with the structure of an \emph{algebra} $\alpha\colon \Sigma(\nu B)\to \nu B$. The two models form the \emph{initial} and \emph{final} bialgebra for the given GSOS law, and the  unique bialgebra morphism $\den{-}\colon \mS\to \nu B$ gives a denotational semantics of programs.
\begin{equation}\label{eq:den}
	\begin{tikzcd}
		\Sigma(\mS) \ar{r}{\cong} \ar{d}[swap]{\Sigma \den{-}} & \mS \ar[dashed]{d}{\den{-}} \ar{r}{\gamma} & B(\mS) \ar{d}{B\den{-}} \\
		\Sigma(\nu B) \ar{r}{\alpha} & \nu B \ar{r}{\cong} & B(\nu B)
	\end{tikzcd}
\end{equation}
The main feature of this elegant bialgebraic approach is that the two desired well-behavedness properties of the denotational semantics come for free: adequacy holds trivially because denotational and behavioural equivalence simply coincide,  and compositionality holds by construction because the denotation morphism $\den{-}$ is a morphism of $\Sigma$-algebras. Thus, for languages modelled in abstract GSOS, intricate compositionality proofs are no longer required.

\paragraph*{Higher-Order Abstract GSOS} A key limitation of Turi and Plotkin's framework is that it does not apply to \emph{higher-order} languages: GSOS laws are not expressive enough to capture the operational semantics of the $\lambda$-calculus or related systems. In recent work, \citet{Goncharov23,gmstu26_jfp} have removed this limitation by developing the \emph{higher-order abstract GSOS} framework. The underlying idea is simple: move from behaviour \emph{endo}functors to \emph{bi}functors $B\colon \C^\op\times \C \to \C$ of mixed variance, accounting for the fact that in higher-order languages, programs can occur both as inputs (contravariantly) and as outputs (covariantly) of programs. For example, an untyped $\lambda$-calculus can be modelled by a behaviour bifunctor like $B(X,Y)=Y+Y^X$, which reflects that a $\lambda$-term either $\beta$-reduces to a $\lambda$-term, or acts as a function mapping $\lambda$-terms to $\lambda$-terms. The operational rules of a higher-order language are then presented as a \emph{higher-order GSOS law}, a (di)natural transformation of type
\[ \rho_{X,Y}\colon \Sigma(X\times B(X,Y))\to B(X,\Sigmas(X+Y)). \]
Much like in the first-order case, every higher-order GSOS law induces an \emph{operational model}, which extends the initial algebra $\mS$ of programs (e.g.~$\lambda$-terms) to the initial \emph{higher-order bialgebra}
\[
	\begin{tikzcd}
		\Sigma(\mS) \ar{r}{\cong}  & \mS \ar{r}{\gamma} & B(\mS,\mS).
	\end{tikzcd}
\]
Thanks to its flexible and highly parametric nature, this abstract framework applies to a wide range of higher-order languages; for instance, it has been shown to capture the operational semantics of both typed~\cite{Goncharov24,gmtu24lics} and untyped~\cite{Goncharov23} languages, computational effects such as (higher-order) state~\cite{Goncharov25b}, concurrency~\cite{Goncharov25}, or probabilities~\cite{Urbat26}, and various modes of program evaluation (call-by-name, call-by-value, call-by-push-value)~\cite{10.1145/3704871}. Additionally, fundamental operational reasoning techniques such as Howe's method~\cite{UrbatTsampasEtAl23,Urbat26} and (step-indexed) logical relations~\cite{Goncharov24,gmtu24lics} have been developed at the generality of higher-order abstract GSOS. Thus, overall, the theory of higher-order operational semantics is rather well-developed in the abstract bialgebraic framework. In contrast, a bialgebraic understanding of higher-order \emph{denotational} semantics along the lines of Turi and Plotkin's original vision has been missing so far.

\subsection*{Contributions} In the present paper, we develop a novel theory of
\emph{bialgebraic denotational semantics} within the higher-order abstract GSOS
framework. The main challenge is to identify the correct abstract notion of
denotational model. Following the guidance of Turi and Plotkin's work, the most
obvious approach is to let \emph{final} higher-order coalgebras (and bialgebras
derived from them) serve as denotational domains. Unfortunately, this idea does
not work: while in the first-order setting many important behaviour endofunctors
have a final coalgebra, it turns out that the mixed-variance bifunctors
modelling higher-order languages usually fail to have a final higher-order
coalgebra, even for the simplest kinds of deterministic languages
(\Cref{ex:final-higher-order-coalgebra-not-exist}). Unlike in the first-order
setting, the non-existence of final higher-order coalgebras is not merely due to
set-theoretic size issues, which could be circumvented by suitably bounding the functors or moving to superlarge categories, but inherent to the nature of higher-order coalgebras and their morphisms.

The core insight of our contribution is to work with \emph{locally final coalgebras} instead of {final} ones. In a nutshell, a higher-order coalgebra $Z\xto{\cong} B(Z,Z)$ for a bifunctor $B$ is {locally final} if it is final as a coalgebra for the \emph{endo}functor $B(Z,-)\colon \C\to \C$. We think of a locally final coalgebra as a domain of abstract higher-order behaviours. Unlike final coalgebras, locally final coalgebras exist for many higher-order languages, and are often unique up to isomorphism. Thus, they form a natural foundation for bialgebraic denotational models in higher-order abstract GSOS. These claims are substantiated by our main results, which we summarize as follows:

\begin{enumerate}
	\item We show that every locally final coalgebra $Z$ canonically extends to a bialgebra for a given higher-order GSOS law, provided that the latter satisfies a syntactic restriction called \emph{relative flatness} (\Cref{th:denotational-model}). We thus obtain, as a higher-order counterpart of \eqref{eq:den}, a unique denotation morphism
	      \begin{equation}
		      \begin{tikzcd}[row sep=1, column sep=12]
			      \Sigma(\mS) \ar{r}{\cong} \ar{dd}[swap]{\Sigma \den{-}} & \mS \ar[dashed]{dd}{\den{-}} \ar{r}{\gamma} & B(\mS,\mS) \ar{dr}{B(\id,\den{-})} & \\
			      &&& B(\mS,Z) \\
			      \Sigma Z \ar{r}{a} & Z \ar{r}{\cong} & B(Z,Z) \ar{ur}[swap]{B(\den{-},\id)} &
		      \end{tikzcd}
	      \end{equation}
	      from the initial higher-order bialgebra $\mS$ to the bialgebra $Z$. Just like in first-order abstract GSOS, this denotational semantics is compositional by construction. Moreover, it is shown to be \emph{adequate} (\Cref{th:adequacy}) in the sense that denotational equivalence implies behavioural equivalence with respect to the final coalgebra for the endofunctor $B(\mS,-)$. In concrete examples, behavioural equivalence corresponds to notions of (strong) \emph{applicative bisimilarity}~\cite{Abramsky90} of programs.
	\item We develop a general theory of existence, construction, and uniqueness of locally final coalgebras. Our approach is based on the metric-enriched framework of \emph{M-categories} due to~\citet{Birkedal10}, which provides a natural setting for the construction of fixed points of functors. We give a sufficient criterion (\Cref{thm:lfc-existence}) for a behaviour bifunctor $B\colon \C^\op\times \C\to \C$ on an M-category to admit a (unique) locally final coalgebra, and provide an iterative construction method.
\end{enumerate}
To highlight the scope of our theoretical results, we instantiate them to various higher-order languages, including  $\lambda$-calculi and combinatory logics (both simply typed and untyped), as well as extensions of those languages with concurrent or probabilistic features. In the untyped cases, our example languages are modelled in guarded form within the \emph{topos of trees}, a standard model of synthetic guarded domain theory~\cite{nakano2000,Birkedal12,bbm14,mp16},
and their generic denotational models based on locally final coalgebras consist of a form of \emph{guarded interaction trees}~\cite{Frumin24}. In this way, a well-studied class of denotational models for higher-order languages naturally emerges in our abstract setting.

\section{Higher-Order Abstract GSOS}
We start by reviewing the higher-order abstract GSOS
framework~\cite{Goncharov23} and explain how it is used to model the
operational semantics of higher-order languages such as $\lambda$-calculi or
combinatory logics. The framework is based on three key categorical
abstractions:
\begin{enumerate}
	\item \emph{Algebras} for a functor capture program constructors and sets of program terms.
	\item \emph{Higher-order coalgebras} for a bifunctor capture transition systems on program terms.
	\item \emph{Higher-order GSOS laws} capture the operational semantics specifying such transition systems.
\end{enumerate}
In the following, we introduce these ingredients step by step.

\subsection{Algebras and Coalgebras}\label{sec:categories}
Let us first recall some required terminology from category
theory~\cite{Awodey10}. Familiarity with the basics, such as functors, natural
transformations, (co)limits, and monads, is assumed.

\paragraph{Notation}
Given objects
$X_1, X_2$ in a category~$\C$, we write $X_1\times X_2$ for the
product and $\langle f_1, f_2\rangle\c X\to X_1\times X_2$ for the
pairing of $f_i\c X\to X_i$, $i=1,2$. We let
$X_1+X_2$ denote the coproduct, $\inl\c X_1\to X_1+X_2$ and
$\inr\c X_2\to X_1+X_2$ the injections, $[g_1,g_2]\c X_1+X_2\to X$ the
copairing of $g_i\colon X_i\to X$, $i=1,2$, and
$\nabla=[\id_X,\id_X]\colon X+X\to X$ the codiagonal. We write $\Set$ for the category of sets and functions.


\paragraph{Algebras.}
Given an endofunctor $\Sigma$ on a category $\C$, a \emph{$\Sigma $-algebra} is
a pair $(A,a)$ consisting of an object~$A$ (the \emph{carrier} of the algebra) and a morphism
$a\colon \Sigma A\to A$ (its \emph{structure}). A \emph{morphism} from
$(A,a)$ to a $\Sigma$-algebra $(B,b)$ is a morphism $h\colon A\to B$
of~$\C$ such that $h\comp a = b\comp \Sigma h$. Algebras for~$\Sigma $ and their
morphisms form a category, and an \emph{initial} $\Sigma $-algebra
is simply an initial object in that category.  We denote the initial
$\Sigma $-algebra by $\mu \Sigma $ if it exists, and its structure by
$\ini\colon \Sigma (\mu \Sigma ) \to \mu \Sigma $. By Lambek's lemma~\cite[Lem.~2.4.3]{Jacobs16}, the structure $\ini$ is an isomorphism.

A \emph{free $\Sigma $-algebra} on an object $X$ of $\C$ is a
$\Sigma $-algebra $(\Sigma ^{\star}X,\iota_X)$ with a morphism
$\eta_X\c X\to \Sigma ^{\star}X$ of~$\C$ such that for every algebra $(A,a)$
and every $h\colon X\to A$ in $\C$, there exists a unique $\Sigma $-algebra
morphism $h^\#\colon (\Sigma ^{\star}X,\iota_X)\to (A,a)$ with
$h=h^\#\comp \eta_X$. If free algebras exist on every object, their formation
induces a monad $\Sigma ^{\star}\colon \C\to \C$, the (\emph{algebraically}) \emph{free
	monad}~\cite{Kelly80} generated
by~$\Sigma $. Every $\Sigma $-algebra $(A,a)$ yields an Eilenberg-Moore algebra
$\wh{a} \colon \Sigma ^{\star} A \to A$, the free extension of $\id_A\c A\to
	A$.


\begin{example}
	The most familiar example of functor algebras is that of algebras for a
	signature. Given a set $S$ of \emph{sorts}, an \emph{$S$-sorted algebraic signature} consists of a set~$\Sigma$
	of operation symbols together with a map $\ar\colon \Sigma\to S^{\star}\times S$
	associating to every $\f\in \Sigma$ its \emph{arity}. We write $\f\colon s_1\times\cdots\times s_n\to s$ if $\ar(\f)=(s_1,\ldots,s_n,s)$, and $\f\colon s$ if $n=0$ (in which case $\f$ is called a \emph{constant}). Every
	signature~$\Sigma$ induces a polynomial functor on the category $\Set^S$ of $S$-sorted sets, denoted by the
	same letter~$\Sigma$, given by $(\Sigma X)_s = \coprod_{\f\colon s_1\cdots s_n\to s} \prod_{i=1}^n X_{s_i}$ for $X\in \Set^S$ and $s\in S$. An algebra for the functor $\Sigma$ is
	precisely an algebra for the signature $\Sigma$, viz.~an $S$-sorted set $A=(A_s)_{s\in S}$ in $\Set^S$ equipped with an operation $\f^A\colon \prod_{i=1}^n A_{s_i}\to A_s$ for every $\f\colon s_1\cdots s_n\to s$ in $\Sigma$. Morphisms of $\Sigma$-algebras are $S$-sorted
	maps respecting the algebraic structure. Given an $S$-sorted set $X$ of
	variables, the free algebra $\Sigmas X$ is the $\Sigma$-algebra of
	$\Sigma$-terms with variables from~$X$; more precisely, $(\Sigmas X)_s$ is inductively defined by $X_s\seq (\Sigmas X)_s$ and $\f(t_1,\ldots,t_n)\in (\Sigmas X)_s$ for all $\f\colon s_1\cdots s_n\to s$ and $t_i\in (\Sigmas X)_{s_i}$. In particular, the free
	algebra on the empty set is the initial algebra $\mu \Sigma$; it is
	formed by all \emph{closed terms} of the signature. For every
	$\Sigma$-algebra $(A,a)$, the induced Eilenberg-Moore algebra
	$\wh{a}\colon \Sigmas A \to A$ is given by the map that evaluates terms
	over~$A$ in the algebra~$A$.
\end{example}

\paragraph{Coalgebras.}
Dual to the notion of algebra, a \emph{coalgebra} for an
endofunctor $B$ on $\C$ is a pair $(C,c)$ of an object $C$ (the
\emph{state space}) and a morphism $c\colon C\to BC$ (the
\emph{structure}). A \emph{morphism} $h\colon (C,c)\to (D,d)$ of coalgebras is a morphism $h\colon C\to D$ of $\C$ with $Bh\cdot c = d\cdot h$. A  \emph{final coalgebra} is a final object in the category of coalgebras, denoted by $\nu B$ if it exists. Coalgebras are thought of as categorical abstractions of transition systems whose transition type is specified by the functor $B$, and the final coalgebra $\nu B$ is the domain of abstract observable behaviours of such systems.

\begin{example}\label{ex:final-coalg}
	Consider the  endofunctor $BX = X + O$ on $\Set$ for a fixed set $O$ of outputs. A $B$-coalgebra $c\colon C\to C+O$ corresponds to a deterministic transition system where each state $x\in C$ either transitions to another state $y=\gamma(x)\in C$ (notation: $x\to y$), or terminates with an output $o=\gamma(x)\in O$ (notation: $x{\downarrow} o$.) The final coalgebra is given by
	\[ \tau\colon \Nat\times O + \{\infty\} \to \Nat\times O + \{\infty\} + O,
		\qquad
		\begin{cases}
			(0,o)   & \mapsto\quad o       \\
			(n+1,o) & \mapsto\quad (n,o)   \\
			\infty  & \mapsto\quad \infty.
		\end{cases}
	\]
	The unique morphism $h\colon C\to \Nat\times O + \{\infty\}$ from a coalgebra $(C,c)$ to the final coalgebra sends a state $x\in C$ to its observable behaviour:
	\[ h(x) = \begin{cases}
			(n,o)  & \text{if there exists a finite path of the form $x=x_0\to x_1\to \cdots \to x_n{\downarrow} o$} \\
			\infty & \text{if there exists an infinite path $x=x_0\to x_1\to x_2\to \cdots$}.
		\end{cases} \]
	In other words, $h$ observes whether the system eventually terminates from the state $x$, and in that case records the number of steps until termination and the final output. For more information on coalgebras as models of transition systems, we refer to \citet{Rutten2000} or \citet{Jacobs16}.
\end{example}

\paragraph{Higher-Order Coalgebras.} In higher-order abstract GSOS, one considers a generalized notion of coalgebra whose transition type is given by a \emph{bi}functor $B\colon \C^\op\times \C\to \C$ of mixed
variance rather than an endofunctor. A \emph{higher-order coalgebra} for $B$ is a pair $(C,c)$ of an
object $C$ and a morphism $c\colon C\to B(C,C)$. Thus $(C,c)$ is a coalgebra for the endofunctor $B(C,-)$. A \emph{morphism} from $(C,c)$
to a higher-order coalgebra $(D,d)$ is a morphism $h\colon C\to D$ of $\C$ such
that the diagram
\begin{equation*}
	\begin{tikzcd}[row sep=0]
		C \ar{r}{c} \ar{dd}[swap]{h} & {B(C,C)} \ar{dr}{B(\id,h)} &\\
		&& {B(C,D)} \\
		D \ar{r}{d} & {B(D,D)} \ar{ur}[swap]{B(h,\id)} &
	\end{tikzcd}
\end{equation*}
commutes. Higher-order coalgebras and their morphisms form a category, and a \emph{final} higher-order coalgebra is a final object of that category. Intuitively, a higher-order coalgebra represents the transition behaviour of higher-order programs. For instance, for $B(X,Y)=Y+Y^X$ on $\Set$ and a set~$\Lambda$ of programs (say $\lambda$-terms), a higher-order coalgebra $c\colon \Lambda\to \Lambda+\Lambda^\Lambda$ specifies that a program~$p$ either $\beta$-reduces to a program $c(p)\in \Lambda$, or acts as a higher-order function $c(p)\in \Lambda^\Lambda$ that maps programs to programs. A concrete example of a higher-order coalgebra is given in the next subsection.

\subsection{Higher-Order GSOS Laws}\label{sec:ho-gsos}
In higher-order abstract GSOS~\cite{Goncharov23}, the operational
semantics of a higher-order language is presented via a
\emph{higher-order GSOS law}, a categorical structure parametric in
\begin{enumerate}
	\item a category $\C$ with binary products and binary coproducts;
	\item an endofunctor $\Sigma \c \C \to \C$ that has  an initial algebra $\mS$ and a
	      free $\Sigma$-algebra $\Sigmas X$  on every object $X$ (so that $\Sigma$ generates an algebraically free
	      monad $\Sigmas$);
	\item a mixed-variance bifunctor $B\colon \C^\opp\times \C\to \C$.
\end{enumerate}
The functor $\Sigma$ represents the \emph{syntax} (constructors) of the language under consideration, with the initial algebra $\mS$ thought of as the object of program terms. The bifunctor $B$ represents the \emph{behaviour} of programs; it describes the kind of transitions programs can make. The motivation behind $B$ having two arguments is
that transitions may have labels, which behave contravariantly, and
poststates, which behave covariantly. In term models, the
objects of labels and states coincide.
\begin{definition}[Higher-order GSOS law]\label{def:ho-gsos-law}
	A \emph{higher-order GSOS law} of $\Sigma$ over $B$
	is a family \eqref{eq:ho-gsos-law} of morphisms that is dinatural in $X \in \C$ and natural in $Y\in \C$.
	\begin{align}\label{eq:ho-gsos-law}
		\rho_{X,Y} \c \Sigma (X \times B(X,Y))\to B(X, \Sigma^\star (X+Y)).
	\end{align}
\end{definition}

The idea is that a higher-order GSOS law $\rho$ encodes the (small-step) operational rules of a given higher-order language into a family of maps. It assigns to a constructor of
the language with formal arguments from~$X$ having specified
next-step behaviours in~$B(X,Y)$ (with labels in~$X$ and
poststates in~$Y$) a next-step behaviour in
$B(X, \Sigma^\star (X+Y))$ with the same labels, and
poststates being program terms mentioning variables from both~$X$
and~$Y$. (Di)naturality of~$\rho$ amounts to the encoded rules being parametrically polymorphic, that is, they do not inspect the structure of their metavariables.

Every higher-order GSOS law~\eqref{eq:ho-gsos-law} induces an \emph{operational
	($\rho$-)model}, a higher-order coalgebra
\begin{equation}\label{eq:operational-model}
	\gamma \c \mS \to B(\mS,\mS)
\end{equation}
carried by the object $\mS$ of program terms. Its structure is defined via \emph{primitive
	recursion}~\cite[Prop. 2.4.7]{Jacobs16} as the unique morphism~$\gamma$ making the following diagram commute:
\begin{equation}\label{eq:operational-model-def}
	\begin{tikzcd}[column sep=4ex, row sep=normal]
		\Sigma(\mS)
		\dar["\Sigma \brks{\id,\,\gamma}"']
		\ar[r,"\iota"]
		&  \mS \ar[dashed]{r}{\gamma} & B(\mS,\mS)
		\\
		\Sigma (\mS\times {B(\mS,\mS)})
		\ar{rr}{\rho_{\mS,\mS}}
		&&
		B(\mS,\Sigma^\star(\mS+\mS))
		\ar{u}[swap]{B(\id,\,\hat\iota\comp\Sigmas\nabla)}
	\end{tikzcd}
\end{equation}
Intuitively, the operational $\rho$-model is the transition system that runs programs according to the operational rules represented by the higher-order GSOS law $\rho$.

\paragraph{Simply Typed SKI Calculus.}
To illustrate higher-order abstract GSOS and our bialgebraic approach to
denotational semantics developed subsequently, we use as our running example an
extended version of the simply typed SKI calculus~\cite{Hindley08}, a typed
combinatory logic called \stsc~\cite{Goncharov24}. It is expressively
equivalent to the simply typed $\lambda$-calculus but does not use variables;
hence it avoids the technicalities of variable binding and
substitution (see \Cref{subsec:lambda} for the treatment of languages with
variables). The set~$\Ty$ of \emph{types} of \stsc is inductively defined as\footnote{Extending $\Ty$ with additional base types (e.g.\ booleans, natural numbers), sum types, or product types, poses no difficulty.}
\begin{equation}\label{eq:type-grammar}
	\Ty \Coloneqq \utype \mid \arty{\Ty}{\Ty}.
\end{equation}

The constructor $\arty{}{}$ for function types is right-associative, i.e.\
$\arty{\tau_{1}}{\arty{\tau_{2}}{\tau_{3}}}$ is parsed as
$\arty{\tau_{1}}{({\arty{\tau_{2}}{\tau_{3}}})}$. The terms of $\stsc$ are formed over the $\Ty$-sorted signature $\Sigma$ whose operation symbols are listed below, with
$\tau,\tau_1,\tau_2,\tau_3$ ranging over all types in $\Ty$:
\begin{align*}
	 & \mathsf{e} \c \utype                                                                                                                                 &  & \mathsf{app}_{\tau_1,\tau_2}\colon (\arty{\tau_{1}}{\tau_{2}})\times \tau_1\to \tau_2 \\
	 & S_{\tau_{1},\tau_{2},\tau_{3}} \c \arty{(\arty{\tau_{1}}{\arty{\tau_{2}}{\tau_{3}}})}{\arty{(\arty{\tau_{1}}{\tau_{2}})}{\arty{\tau_{1}}{\tau_{3}}}} &  & K_{\tau_{1},\tau_{2}} \c \arty{\tau_{1}}{\arty{\tau_{2}}{\tau_{1}}}                   \\
	 & S'_{\tau_1,\tau_2,\tau_3}\c (\arty{\tau_{1}}{\arty{\tau_{2}}{\tau_{3}}})\to (\arty{{(\arty{\tau_{1}}{\tau_{2})}}}{\arty{\tau_{1}}{\tau_{3}}})        &  & K'_{\tau_1,\tau_2}\colon \tau_1\to (\arty{\tau_{2}}{\tau_{1}})                        \\
	 & S''_{\tau_1,\tau_2,\tau_3}\c (\arty{\tau_{1}}{\arty{\tau_{2}}{\tau_{3}}})\times (\arty{\tau_{1}}{\tau_{2}})\to (\arty{\tau_{1}}{\tau_{3}})           &  & I_{\tau} \c \arty{\tau}{\tau}
\end{align*}
We let $\Tr=\mS$ denote the $\Ty$-sorted set of closed $\Sigma$-terms. The intention is that $\mathsf{app}$ is function application (accordingly, we write $s\, t$ for $\mathsf{app}(s,t)$),
and that the constants $I_\tau$, $K_{\tau_1,\tau_2}$, $S_{\tau_1,\tau_2,\tau_3}$ represent the $\lambda$-terms (combinators)
\[I_\tau=\lambda t.\,t,\qquad  K_{\tau_1,\tau_2}= \lambda t.\,\lambda s.\, t,\qquad S_{\tau_1,\tau_2,\tau_3}=\lambda t.\,\lambda s.\,\lambda u.\, (t\, u)\, (s\, u).\]
%
%
\begin{figure*}[t]
	\begin{gather*}
		\inference{}{S_{\tau_{1},\tau_{2},\tau_{3}}\xto{t}S'_{\tau_{1},\tau_{2},\tau_{3}}(t)}
		\qquad
		\inference{}{S'_{\tau_{1},\tau_{2},\tau_{3}}(p)\xto{t}S''_{\tau_{1},\tau_{2},\tau_{3}}(p,t)}  \qquad
		\inference{}{S''_{\tau_{1},\tau_{2},\tau_{3}}(p,q)\xto{t}(p\app{\tau_{1}}{\arty{\tau_{2}}{\tau_{3}}} t)\app{\tau_{2}}{\tau_{3}} (q\app{\tau_{1}}{\tau_{2}} t)}
		\\[1ex]
		\inference{}{K_{\tau_{1},\tau_{2}}\xto{t}K'_{\tau_{1},\tau_{2}}(t)}
		\qquad
		\inference{}{K'_{\tau_{1},\tau_{2}}(p)\xto{t}p}
		\qquad
		\inference{}{I_{\tau}\xto{t}t} \qquad
		\inference{}{\mathsf{e} \xto{\checkmark}}
		\\[1ex]
		\inference{p\to p'}{p \app{\tau_{1}}{\tau_{2}} q\to p' \app{\tau_{1}}{\tau_{2}} q}
		\qquad
		\inference{p\xto{q} p'}{p \app{\tau_{1}}{\tau_{2}} q\to p'}
	\end{gather*}
	\caption{(Call-by-name) operational semantics of \stsc.}
	\label{fig:skirules}
\end{figure*}
The operational semantics of \stsc involves three kinds of transitions
($\xto{\checkmark}$, $\xto{t}$, $\xto{}$), which are specified by the inductive rules of \Cref{fig:skirules}. Here, $p,p',q,t$ range over terms in $\Tr$ of appropriate
type. Intuitively, $s\xto{\checkmark}$ identifies~$s$ as an explicitly
irreducible term; $s\xto{t} r$ states that $s$ acts as a function mapping $t$
to~$r$; and $s\to t$ indicates that~$s$ reduces to $t$. The use of the auxiliary operators  $K'$, $S'$
and $S''$ follows the presentation of combinatory logics by~\citet{GianantonioRPO}. These operators do not impact the behaviour of programs, except for possibly adding
more unlabelled transitions, but allow for a simpler coalgebraic presentation of the language. For example, the standard rule $S\, t{}\, s{}\, e\to
	(t\,e)\, (s\, e)$ for the $S$-combinator is rendered in \stsc as the chain of transitions \[S\, t{}\, s{}\, e\to S'(t)\, s\, e\to S''(t,s)\, e\to (t\,e)\,(s\,e).\]
%
%
%

The transition system for \stsc is deterministic: for every term~$s$, either $s
	\xto{\checkmark}$, or there is a unique $t$ such that $s\to t$, or for each
appropriately typed $t$ there is a unique $s_t$ such that $s\xto{t} s_t$. These
transitions are captured by the bifunctor
\begin{equation}
	\label{eq:beh}
	B\c (\Set^{\Ty})^{\opp} \times \Set^{\Ty} \to \Set^{\Ty}, \qquad B_\utype(X,Y)=Y_\utype +1,\qquad B_{\arty{\tau_1}{\tau_2}}(X,Y) = Y_{\arty{\tau_1}{\tau_2}} + Y_{\tau_2}^{X_{\tau_1}},
\end{equation}
where $1=\{*\}$. The operational rules in \Cref{fig:skirules} determine a higher-order $B$-coalgebra given by

\noindent \begin{minipage}{.25\textwidth}
	\[
		\gamma\colon \Tr\to B(\Tr,\Tr),
	\]
\end{minipage}
\begin{minipage}{.75\textwidth}
	\begin{flalign}
		                                                                                     &        & \gamma_\utype(s)=                       & \;\inr(*)                &  & \text{if $s \xto{\checkmark}$ where $s\c \utype$,} \notag  \\
		                                                                                     &        & \gamma_\tau(s)=                         & \;\inl(t)                &  & \text{if $s \xto{} t$ where $s,t\c\tau$,}\label{exa:gamma} \\
		                                                                                     &        & \gamma_{\arty{\tau_{1}}{\tau_{2}}}(s) = & \; \inr(\lambda t.\,s_t) &  &
		\text{if $s \xto{t} s_t$ for $s \c \arty{\tau_{1}}{\tau_{2}}$ and all $t\c \tau_1$.} & \notag
	\end{flalign}
\end{minipage}

\medskip\noindent The operational rules of \stsc{} in \Cref{fig:skirules} can be encoded into a higher-order GSOS law $\rho$ of the syntax functor
$\Sigma$ over the behaviour bifunctor $B$, i.e.\ a family of maps~\eqref{eq:ho-gsos-law}
dinatural in $X\in \Set^\Ty$ and natural in $Y\in \Set^\Ty$. The maps
$\rho_{X,Y}$ are cotuples defined by distinguishing cases on the constructors. To avoid cumbersome notation, when giving terms in $\Sigmas(X+Y)$ we usually leave coproduct injections $\inl\colon X\to X+Y$ and $\inr\colon Y\to X+Y$ and the monad unit $\eta\colon \Id\to \Sigmas$ implicit. For example, in the clause for $\mathsf{app}$ below, $f(q)\in Y_{\tau_2}$ is identified with $\eta_{\tau_2}(\inr_{\tau_2}(f(q))\in \Sigmas_{\tau_2}(X+Y)$, which in turn is viewed as an element of the left summand of $B_{\tau_2}(X,\Sigmas(X+Y)) = \Sigmas_{\tau_2}(X+Y)+\cdots$.
\begin{align*}
	 & \rho_{X,Y}(tr) =
	\texttt{case}~tr~\texttt{of}\!\!\!\!\!\!\!\!\!\!\!                                                                                                                                    \\
	 & \mathsf{e}                                       & \mapsto\;\; & *                    &             & \mathsf{app}_{\tau_1,\tau_2}((p,f),(q,g)) & \mapsto\;\;
	 &
	\begin{cases}
		f(q),                              & \!\!f \in Y_{\tau_{2}}^{X_{\tau_{1}}}   \\
		\mathsf{app}_{\tau_1,\tau_2}(f,q), & \!\!f \in Y_{\arty{\tau_{1}}{\tau_{2}}}
	\end{cases}                                                                         \\
	 & S_{\tau_{1},\tau_{2},\tau_{3}}                   & \mapsto\;\; & \lambda t\,.S'(t)    &             & K_{\tau_{1},\tau_{2}}                     & \mapsto\;\; & \lambda t.\, K'(t) \\
	 & S'_{\tau_{1},\tau_{2},\tau_{3}}(p,f)             & \mapsto\;\; & \lambda t.\,S''(p,t) &             & K'_{\tau_{1},\tau_{2}}(p,f)               & \mapsto\;\; & \lambda t.\, p     \\*
	 & S''_{\tau_{1},\tau_{2},\tau_{3}}((p,f),(q,g))
	 & \mapsto\;\;
	 & \lambda t.\,(p \app{}{} t)\app{}{}(q \app{}{} t) &             & I_{\tau}             & \mapsto\;\; & \lambda t.\, t.
\end{align*}

The operational $\rho$-model is the coalgebra \eqref{exa:gamma} induced by the rules of \stsc.

\subsection{Models of Higher-Order GSOS Laws}


Every higher-order GSOS law induces a {category of models} where both the
operational model \eqref{eq:operational-model} and the denotational models
constructed further below live. A model consists of an algebra and a
higher-order coalgebra that are compatible with the given higher-order GSOS law.

\begin{definition}[$\rho$-bialgebra]\label{def:model} Let $\rho$ \eqref{eq:ho-gsos-law} be a higher-order GSOS law.
	A $\rho$-\emph{bialgebra}, or \emph{$\rho$-model}, $(X,a,c)$ is given by an object  $X\in \C$ with a $\Sigma$-algebra and a higher-order $B$-coalgebra structure
	\begin{displaymath}
		\Sigma X \xto{a} X \xto{c} B(X, X)
	\end{displaymath}
	such that diagram \eqref{eq:rho-model} below on the left commutes.
	A \emph{morphism} $h\colon (X,a_X,c_X) \to (Y,a_Y,c_Y)$ of $\rho$-bialgebras is a morphism $h: X \to Y$ of $\C$ that is both a morphism of $\Sigma$-algebras and a morphism of higher-order $B$-coalgebras, that is, the diagram below on the right
	commutes.

	\noindent \begin{minipage}{.55\textwidth}
		\begin{equation}\label{eq:rho-model}
			\begin{tikzcd}[column sep=0ex, row sep=normal]
				\Sigma X
				\dar["\Sigma \brks{\id,\, c}"']
				\ar[r,"a"]
				&  X \ar{r}{c} & B(X,X)
				\\
				\Sigma (X\times {B(X,X)})
				\ar{rr}{\rho_{X,X}}
				&&
				B(X,\Sigma^\star(X+X))
				\ar{u}[swap]{B(\id,\,\hat{a}\cdot\Sigmas\nabla)}
			\end{tikzcd}
		\end{equation}
	\end{minipage}
	\begin{minipage}{.03\textwidth}~\end{minipage}
	\begin{minipage}{.4\textwidth}
		\begin{equation*}\label{eq:bialg-morphism}
			\begin{tikzcd}[row sep=1, column sep=12]
				{\Sigma X} & X & {B(X,X)} & \\
				&&& {B(X,Y)} \\
				{\Sigma Y} & Y & {B(Y,Y)}&
				\arrow["{a_X}", from=1-1, to=1-2]
				\arrow["{\Sigma h}"', from=1-1, to=3-1]
				\arrow["{c_X}", from=1-2, to=1-3]
				\arrow["h", from=1-2, to=3-2]
				\arrow["{B(\id,h)}", from=1-3, to=2-4]
				\arrow["{a_Y}", from=3-1, to=3-2]
				\arrow["{c_Y}", from=3-2, to=3-3]
				\arrow["{B(h,\id)}"', from=3-3, to=2-4]
			\end{tikzcd}
		\end{equation*}
	\end{minipage}
\end{definition}

The collection of $\rho$-bialgebras and their morphisms form a
category~\cite[Prop.~4.19]{Goncharov23}. An \emph{initial $\rho$-bialgebra} is
an initial object of that category. By diagram
\eqref{eq:operational-model-def}, the operational $\rho$-model is a
$\rho$-bialgebra $(\mS,\ini,\gamma)$, and we have the following result~\cite[Prop.~4.20]{Goncharov23}:
\begin{theorem}\label{prop:initial-bialgebra}
	The operational model $(\mS,\ini,\gamma)$ of a higher-order GSOS law $\rho$ is the initial $\rho$-bialgebra.
\end{theorem}


What about a \emph{denotational} model of a higher-order GSOS law $\rho$? As discussed in the introduction, one may first think of constructing such a model by taking the final higher-order $B$-coalgebra and extending it to a
final $\rho$-bialgebra, analogous to first-order abstract GSOS~\cite{Turi97}. However, this straightforward approach fails: final higher-order coalgebras usually do not exist for non-trivial behaviour bifunctors, even in such simple settings
as \stsc.
\begin{example}\label{ex:final-higher-order-coalgebra-not-exist}
	The behaviour bifunctor $B$ \eqref{eq:beh} of \stsc has no final higher-order coalgebra. Indeed, suppose that such a coalgebra $t\colon T\to B(T,T)$ exists. For any higher-order coalgebra $c\colon C\to B(C,C)$, the unique morphism $h\colon (C,c)\to (T,t)$ of higher-order coalgebras makes the two diagrams below commute, where we put $\mathsf{u}=\utype$ for short:
	\[
		\begin{tikzcd}
			C_{\mathsf{u}} \ar{r}{c_{\mathsf{u}}} \ar{d}[swap]{h_{\mathsf{u}}} & C_{\mathsf{u}} + 1 \ar{d}{h_{\mathsf{u}} + \id } \\
			T_{\mathsf{u}} \ar{r}{t_{\mathsf{u}}} & T_{\mathsf{u}} + 1
		\end{tikzcd}
		\quad\quad\quad
		\begin{tikzcd}[row sep=0]
			C_{\arty{{\mathsf{u}}}{{\mathsf{u}}}} \ar{r}{c_{\arty{{\mathsf{u}}}{{\mathsf{u}}}}} \ar{dd}[swap]{h_{\arty{{\mathsf{u}}}{{\mathsf{u}}}}} & C_{\arty{{\mathsf{u}}}{{\mathsf{u}}}} + C_{{\mathsf{u}}}^{C_{{\mathsf{u}}}}  \ar{dr}{h_{\arty{{\mathsf{u}}}{{\mathsf{u}}}} +  h_\mathsf{u}\cdot (-)} & \\
			&& T_{\arty{{\mathsf{u}}}{{\mathsf{u}}}} + T_{{\mathsf{u}}}^{C_{{\mathsf{u}}}} \\
			T_{\arty{{\mathsf{u}}}{{\mathsf{u}}}} \ar{r}{t_{\arty{{\mathsf{u}}}{{\mathsf{u}}}}} & T_{\arty{{\mathsf{u}}}{{\mathsf{u}}}} + T_{{\mathsf{u}}}^{T_{{\mathsf{u}}}} \ar{ur}[swap]{h_{\arty{{\mathsf{u}}}{{\mathsf{u}}}} + (-)\cdot h_\mathsf{u}} &
		\end{tikzcd}
	\]
	The first diagram tells us that $(T_\mathsf{u},t_\mathsf{u})$ is the final coalgebra for the set endofunctor $(-)+1$. Indeed, this follows from finality of $(T,t)$ together with the observation that every coalgebra morphism $h_u\colon (C_\mathsf{u},c_\mathsf{u})\to (T_\mathsf{u},t_\mathsf{u})$ extends to a higher-order coalgebra morphism $h\colon (C,c)\to (T,t)$ by choosing $C_\tau=\emptyset$ and $h_\tau\colon C_\tau\to T_\tau$ the unique map for all $\tau\neq \mathsf{u}$.

	Now choose $(C,c)$ to be the following higher-order coalgebra:
	\begin{align*}
		C_{\mathsf{u}}                    & =\{x,y,z\}, &  & c_\mathsf{u}(x)=x,\, c_\mathsf{u}(y)=y,\, c_\mathsf{u}(z) = *,                                                       \\
		C_{\arty{\mathsf{u}}{\mathsf{u}}} & = \{a\},    &  & c_{\arty{\mathsf{u}}{\mathsf{u}}}(a) = (x\mapsto x,\, y\mapsto z,\, z\mapsto z) \in C_{\mathsf{u}}^{C_{\mathsf{u}}},
	\end{align*}
	and $C_\tau$ and $c_\tau$ are arbitrary on all other types $\tau$. By \Cref{ex:final-coalg}, the final coalgebra $(T_\mathsf{u},t_\mathsf{u})$ for $(-)+1$ is carried by $T_\mathsf{u}=\Nat+ \{\infty\}$ and the map $h_\mathsf{u}$ is given by
	\begin{equation}\label{eq:hu-def} h_\mathsf{u}(x)=h_\mathsf{u}(y)=\infty\qquad \text{and}\qquad h_\mathsf{u}(z)=0. \end{equation}
	We put $f:=c_{\arty{\mathsf{u}}{\mathsf{u}}}(a)$. Commutativity of the second diagram above tells us that $g:=t_{\arty{\mathsf{u}}{\mathsf{u}}}(h_{\arty{\mathsf{u}}{\mathsf{u}}}(a))$ is an element of $T_{{\mathsf{u}}}^{T_{{\mathsf{u}}}}$, and moreover that
	\begin{equation}\label{eq:hu-prop} h_\mathsf{u}\cdot f = g\cdot h_{\mathsf{u}}.  \end{equation}
	Therefore, we get
	\[\infty \overset{\eqref{eq:hu-def}}{=} h_\mathsf{u}(x)
		\overset{\text{def.\ $f$}}{=} h_u(f(x))
		\overset{\eqref{eq:hu-prop}}{=} g(h_\mathsf{u}(x))
		\overset{\eqref{eq:hu-def}}{=} g(h_\mathsf{u}(y))
		\overset{\eqref{eq:hu-prop}}{=} h_\mathsf{u}(f(y))
		\overset{\text{def.\ $f$}}{=} h_\mathsf{u}(z) \overset{\eqref{eq:hu-def}}{=} 0,\]
	which is impossible.

\end{example}

Our approach to denotational
semantics in higher-order abstract GSOS therefore does not use final
higher-order coalgebras (and derived bialgebras). Instead, we base our denotational models on
\emph{locally} final higher-order coalgebras, a novel concept that we introduce in the next section.

\section{Bialgebraic Denotational Semantics}
\label{sec:denotational-models}
We now develop a general method to associate to a language specified in higher-order abstract GSOS a denotational model that is adequate with respect to the operational semantics.

\begin{assumption}\label{ass:ho-gsos}
	In the remainder of this section, we fix a higher-order GSOS law $\rho$ of $\Sigma\colon \C\to \C$ over $B\colon \C^\opp\times \C \to \C$, with its operational $\rho$-model $(\mS,\ini,\gamma)$ given by \eqref{eq:operational-model-def}. Moreover, we assume that for each object $X\in \C$ the endofunctor $B(X,-)$ has a final coalgebra $\nu B(X,-)$.
\end{assumption}


\subsection{Locally Final Coalgebras}
The foundation for our bialgebraic approach to denotational semantics is the notion
of a \emph{locally final coalgebra} for the behaviour bifunctor~$B$. It takes
the role that the final coalgebra for an endofunctor has in the setting of first-order abstract GSOS.

\begin{definition}[Locally final coalgebra] A higher-order
	$B$-coalgebra $z : Z \to B(Z,Z)$ is \emph{locally final} if it is a final coalgebra for the endofunctor $B(Z,-)\colon \C\to \C$.
\end{definition}
A locally final coalgebra serves as a domain of abstract higher-order behaviours. Note that the structure of a locally final coalgebra is an isomorphism by the dual of Lambek's lemma (\Cref{sec:categories}). However, unlike the final coalgebra for an endofunctor, a locally final coalgebra need not be unique up to isomorphism of higher-order coalgebras. General sufficient conditions for the existence, construction, and uniqueness  of locally final coalgebras are discussed in \Cref{sec:locally-final-coalgebras}. For now, we consider the setting of \stsc to give the reader some intuition about the nature of such coalgebras:

\begin{example}[Locally final coalgebra for \stsc]\label{ex:stsc}
	The bifunctor $B$ \eqref{eq:beh} on $\Set^\Ty$ which models \stsc has a locally final coalgebra $z\colon Z\xto{\cong} B(Z,Z)$ given by the $\Ty$-sorted set of trees obtained by unravelling all possible $B$-behaviours. More precisely, the components $z_\tau\colon Z_\tau\xto{\cong} B_\tau(Z,Z)$ ($\tau\in \Ty$) of the locally final coalgebra are defined by structural recursion on types as follows:
	\begin{itemize}
		\item $Z_\utype$ is the set containing all finite paths and the infinite path. The map
		      \[ z_\utype\colon Z_\utype \xto{\cong} Z_\utype + 1 = B_\utype(Z,Z) \]
		      sends the path consisting only of the root node to $*\in 1$, and every other
		      path to the subpath starting at the unique child of the root. In particular,
		      the infinite path is sent to itself.
		\item $Z_{\arty{\tau_1}{\tau_2}}$ consists of the infinite path and all trees that are given by a finite path composed with a tree whose root has children indexed by $s\in Z_{\tau_1}$, and where the subtree rooted at any such child is a tree in $Z_{\tau_2}$. The map \[ z_{\arty{\tau_1}{\tau_2}}\colon Z_{\arty{\tau_1}{\tau_2}} \xto{\cong} Z_{\arty{\tau_1}{\tau_2}} + Z_{\tau_2}^{Z_{\tau_1}} = B_{\arty{\tau_1}{\tau_2}}(Z,Z) \]
		      is defined on $t\in Z_{\arty{\tau_1}{\tau_2}}$ as follows: If $t$ is the
		      infinite path, then $z(t)=t$. Otherwise, $t$ is given by an initial finite path
		      followed by a tree as described above. If the initial path has length $>0$,
		      $z(t)\in Z_{\arty{\tau_1}{\tau_2}}$ is the subtree starting at the unique child
		      of the root of $t$. Otherwise, the root of $t$ is already the branching point,
		      and $z(t)\in Z_{\tau_2}^{Z_{\tau_1}}$ is the map where $z(t)(s)$ is the subtree
		      starting at the unique $s$-indexed child of the root of $t$.
	\end{itemize}
	To see that $(Z,z)$ is a locally final coalgebra, we observe that we can describe each individual component $(Z_\tau,z_\tau)$ as a final coalgebra. For $O\in \Set$ consider the endofunctor
	\[ \bar{B}_O\colon \Set\to \Set,\qquad  \barB_O X =X+O. \]
	Its final coalgebra $\nu\bar{B}_O$ is carried by $\Nat\times O + \{\infty\}$ (\Cref{ex:final-coalg}). Identifying $n\in\Nat$ with the path of length $n$ and $\infty$ with the
	infinite path, we see that
	\[ (Z_\utype,z_\utype) = \nu\bar{B}_1 \qquad\text{and}\qquad (Z_{\arty{\tau_1}{\tau_2}}, z_{\arty{\tau_1}{\tau_2}}) = \nu\bar{B}_{Z_{\tau_2}^{Z_{\tau_1}}}.\]
	From this, local finality of the higher-order coalgebra $(Z,z)$ easily follows. Given a $B(Z,-)$-coalgebra
	$c\colon X\to B(Z,X)$, a $B(Z,-)$-coalgebra morphism from $(X,c)$ to $(Z,z)$
	corresponds to a $\Set^\Ty$-morphism $h\colon X\to Z$ making the following
	diagrams commute:
	\[
		\begin{tikzcd}
			X_\utype \ar{r}{c_\utype} \ar{d}[swap]{h_\utype} & X_\utype + 1 \ar{d}{h_\utype + \id } \\
			Z_\utype \ar{r}{z_\utype} & Z_\utype + 1
		\end{tikzcd}
		\quad\quad
		\begin{tikzcd}
			X_{\arty{\tau_1}{\tau_2}} \ar{r}{c_{\arty{\tau_1}{\tau_2}}} \ar{d}[swap]{h_{\arty{\tau_1}{\tau_2}}} & X_{\arty{\tau_1}{\tau_2}} + Z_{\tau_2}^{Z_{\tau_1}}  \ar{d}{h_{\arty{\tau_1}{\tau_2}} + \id} \\
			Z_{\arty{\tau_1}{\tau_2}} \ar{r}{z_{\arty{\tau_1}{\tau_2}}} & Z_{\arty{\tau_1}{\tau_2}} + Z_{\tau_2}^{Z_{\tau_1}}
		\end{tikzcd}
	\]
	The first diagram states that $h_\utype$ is a $\bar{B}_1$-coalgebra
	morphism into $\nu\bar{B}_1$, and the second diagram
	that $h_{\arty{\tau_1}{\tau_2}}$ is a
	$\bar{B}_{Z_{\tau_2}^{Z_{\tau_1}}}$-coalgebra morphism into
	$\nu\bar{B}_{Z_{\tau_2}^{Z_{\tau_1}}}$. By finality of these coalgebras,
	it thus follows that there exists a unique $B(Z,-)$-coalgebra morphism $h\colon
		(X,c) \to (Z,z)$.

	Explicitly, $h$ sends an element of $X$ to the tree obtained by unravelling the coalgebra $c$ starting from $x$. For instance, if $x\in X_{\arty{\tau_1}{\tau_2}}$ and $c_{\arty{\tau_1}{\tau_2}}(x)=x_1$ and $c_{\arty{\tau_1}{\tau_2}}(x_1)=f$ for some $x_1\in X_{\arty{\tau_1}{\tau_2}}$ and $f\in X_{\tau_2}^{Z_{\tau_1}}$, then $h_{\arty{\tau_1}{\tau_2}}(x)$ is the tree given by an initial path of length $1$ composed with a tree whose child with index $z\in Z_{\tau_1}$ is $h_{\tau_2}(f(z))$.

	This concludes the proof that $(Z,z)$ is a locally final coalgebra. We shall see later from a general result (\Cref{thm:xtcl-lfc}) that $(Z,z)$ is, up to isomorphism, the \emph{unique} locally final coalgebra for $B$.

\end{example}

\subsection{Constructing a Denotational Model from a Locally Final Coalgebra}\label{sec:relative-flatness}
A locally final coalgebra represents abstract unravelled higher-order
behaviours and thus provides a natural candidate for a denotational domain for
the given higher-order GSOS law $\rho$. However, to make it properly fit into the
bialgebraic framework, we need to extend the coalgebra to a
$\rho$-bialgebra, that is, equip it with a compatible $\Sigma$-algebra
structure. Formally:

\begin{definition}[Denotational $\rho$-model]
	Let $\rho$ be a higher-order GSOS law. A \emph{denotational $\rho$-model} is a
	$\rho$-bialgebra whose underlying higher-order coalgebra is locally final.
\end{definition}

It turns out that this extension step requires a syntactic restriction on the
higher-order GSOS law, called \emph{relative flatness}. This notion was
originally introduced by \citet{Goncharov24} in the context of constructing
well-behaved logical predicates for higher-order GSOS laws, and remarkably, it
is also precisely what makes the construction of denotational models
work.

\begin{definition}[Relative flatness]
	Let $(J, <)$ be a well-founded relation, and suppose that $\Sigma$ decomposes into a $J$-indexed coproduct $\Sigma = \coprod_{j\in J} \Sigma_j$ of endofunctors $\Sigma_j$, where
	$\Sigma_j$ and $\Sigma_{< j}=\coprod_{i<j} \Sigma_i$ have algebraically free monads
	$\Sigma^{\star}_j$ and $\Sigma^{\star}_{< j}$ for each $j\in J$.
	A \emph{relatively flat higher-order GSOS law} is a $J$-indexed family \eqref{eq:relative-flat-law} of morphisms dinatural in $X\in \C$ and natural in $Y\in \C$.
	\begin{equation}\label{eq:relative-flat-law}
		\rho^j_{X,Y} : \Sigma_j (X\times B(X,Y)) \to B(X, \Sigma^{\star}_{< j} (X+Y) +
		\Sigma_j \Sigma^{\star}_{< j} (X+Y)).
	\end{equation}
\end{definition}


\begin{construction}
	Every relatively flat higher-order GSOS law \eqref{eq:relative-flat-law} induces an (actual) higher-order GSOS law $\rho$ of $\Sigma$ over $B$ whose component at $X,Y\in \C$ is given by the composite
	\begin{align*}
		\rho_{X,Y} \;\coloneq\; \big(\, \Sigma (X\times B(X,Y)) \xrightarrow{[\inj_j \cdot
					\rho^j_X]_{j \in J}}
		\coprod_{j \in J} B(X, \Sigma^{\star}_{< j} X + \Sigma_j \Sigma^{\star}_{<
			j} X)
		\xrightarrow{[B(\id, \, e_{j,X})]_{j \in J}} B(X, \Sigma^{\star} X)\,\big).
	\end{align*}
	Here the morphism $e_{j,X}$ is given by
	\begin{align*}
		e_{j, X}\colon \Sigma^{\star}_{< j} X + \Sigma_j \Sigma^{\star}_{< j} X \to
		\Sigma^{\star} X,\qquad
		e_{j, X} \coloneq [\inj^{\star}_{< j},\, \ini_X \cdot \inj_j \cdot \Sigma_j
			\inj^{\star}_{< j}],
	\end{align*}
	where $\inj_j\colon \Sigma_j\to \Sigma$ and $\inj_{<j}=[\inj_i]_{i<j}\colon \Sigma_{<j}\to \Sigma$ are the coproduct injections (with respective free extensions $\inj_j^\star\colon \Sigmas_j \to \Sigmas$ and $\inj_{<j}^\star\colon \Sigmas_{<j}\to \Sigmas$), and $\ini_X\colon \Sigma\Sigmas X\to \Sigmas X$ denotes the structure of the free $\Sigma$-algebra on $X$.
\end{construction}



Following the interpretation of higher-order GSOS laws as systems of
operational rules, relative flatness means that every constructor from the
signature $\Sigma$ can be assigned a rank $j\in J$ in such a way that every
term $\f(p_1, \dots, p_n)$, with $\f$ having rank $j$, transitions into terms using
only constructors of strictly lower ranks, except for their head symbol, which may
have rank $j$ as well. This is achieved by decomposing $\rho$ into a stratified
family of components.

\begin{example}[\stsc]\label{ex:stsc-flat}
	The higher-order GSOS law corresponding to \stsc is relatively flat: put $J = \{0< 1\}$, let $\Sigma_0$ contain application only, and let $\Sigma_1$ contain all other operation symbols.
\end{example}

\begin{remark}
	Unbounded recursion, e.g.\ in the form of recursive types $\mu X.\tau$ or a $Y$-combinator
	$\mathsf{Y}\,e \to e\, (\mathsf{Y}\,e)$, typically breaks relative flatness.
\end{remark}





For relatively flat laws, we can base bialgebraic denotational semantics on locally final coalgebras:

\begin{theorem}[Extension theorem]\label{th:denotational-model}
	Let $\rho$ be a relatively flat higher-order GSOS law. Then every locally final coalgebra $z\colon Z\xto{\cong} B(Z,Z)$ extends to a denotational $\rho$-model \[\Sigma Z \xto{a_\rho} Z \xto{z} B(Z,Z).\]
\end{theorem}

\begin{proof}[Proof Sketch]
	To construct the algebra structure $a_\rho$, we first define a family of morphisms
	\begin{displaymath}
		(a^j_{\rho} : Z + \Sigma_j Z \to Z)_{j \in J}
	\end{displaymath}
	by well-founded recursion on $j$. Let $j \in J$, and assume that
	$a^i_{\rho}$ is already defined for $i < j$. Put
	\[
		a^{< j}_{\rho}\colon \Sigma_{< j} Z \to Z,\qquad
		a^{< j}_{\rho} \coloneq [a^i_{\rho} \cdot \inr ]_{i < j}.
	\]
	The object $Z+\Sigma_j Z$ carries a $B(Z,-)$-coalgebra structure given by the path from $Z+\Sigma_j Z$ to $B(Z,Z+\Sigma_j Z)$ below. We take $a^j_\rho$ to be the unique morphism to  the final $B(Z,-)$-coalgebra $(Z,z)$:
	\[
		\begin{tikzcd}[scale cd = .83, column sep=4em]
			{Z + \Sigma_j Z} & Z & {B(Z,Z)} \\
			{B(Z,Z)+\Sigma_j (Z \times B(Z,Z))} && {B(Z, Z + \Sigma_j Z)} \\
			{B(Z,Z)+B(Z, \Sigma^{\star}_{<j}Z + \Sigma_j \Sigma^{\star}_{<j}Z)} & {B(Z,Z)+B(Z, Z + \Sigma_j Z)} & {B(Z,Z + \Sigma_j Z)+B(Z, Z + \Sigma_j Z)}
			\arrow["{a^j_{\rho}}", dashed, from=1-1, to=1-2]
			\arrow["{z+\Sigma_j\langle \id,\, {z} \rangle}"', from=1-1, to=2-1]
			\arrow["{z \quad \cong}", from=1-2, to=1-3]
			\arrow["{\id + \rho^j_{Z}}"', from=2-1, to=3-1]
			\arrow["{B(\id,\, a^j_{\rho})}"', from=2-3, to=1-3]
			\arrow["{\id + B(\id, \, \widehat{a^{<j}_{\rho}} + \Sigma_j \widehat{a^{<j}_{\rho}})}", bend left=10, from=3-1, to=3-2]
			\arrow["{B(\id,\,\inl) + \id}", from=3-2, to=3-3, bend left=10]
			\arrow["{\nabla}"', from=3-3, to=2-3]
		\end{tikzcd}
	\]
	The desired
	$\Sigma$-algebra structure $a_\rho$ on $Z$ is then given by
	\[
		a_{\rho}\colon \Sigma Z \to Z,\qquad
		a_{\rho} \coloneq [a^j_{\rho} \cdot \inr]_{j \in J}.
	\]
	A lengthy calculation shows that  $(Z,a_\rho,z)$ forms a $\rho$-bialgebra (see
	the appendix\ifarx{}{ of the extended arXiv
		version~\cite{Goncharov2026arXiv}}).
\end{proof}

\subsection{Adequacy}
We next show that, given a higher-order GSOS law $\rho$, any denotational $\rho$-model is adequate with respect to the operational $\rho$-model. This rests on the following general concept:


\begin{definition}[Adequacy]\label{def:adequacy}
	A  $\Sigma$-algebra $(A,a)$ is (\emph{computationally}) \emph{adequate} if the unique $B(\mS,-)$-coalgebra morphism $\beh\colon (\mS,\gamma)\to \nu B(\mS, -)$ factorizes through the unique $\Sigma$-algebra morphism $\llbracket - \rrbracket\colon \mS\to (A,a)$, i.e.\ there is a morphism $A\to \nu B(\mS, -)$ in $\C$  making the triangle below commute.
	\[
		\begin{tikzcd}
			& \mS \ar{dr}{\beh} \ar{dl}[swap]{\llbracket - \rrbracket} & \\
			A \ar[dashed]{rr}{\exists} && \nu B(\mS,-)
		\end{tikzcd}
	\]
\end{definition}

We think of $(A,a)$ as the algebraic part of a denotational $\rho$-model and of $\llbracket - \rrbracket$ as the map sending a program to its denotation. For $\C = \Set$ (or  presheaf categories $\C=\Set^\B$), the fact that  $\llbracket - \rrbracket$ is a $\Sigma$-algebra morphism means that the ensuing denotational semantics is \emph{compositional}, that is, the denotation of a program $\f(p_1,\ldots,p_n)$ is uniquely determined by the denotations of its subprograms $p_i$. Moreover, adequacy in the above abstract sense entails the usual meaning~\cite{Milner75}, namely that denotational equivalence implies behavioural equivalence of programs:
\[ \forall p,q\in \mS.\, \llbracket p \rrbracket = \llbracket q \rrbracket \implies \beh(p)=\beh(q). \]
In such cases, we also say that the model is \emph{adequate for behavioural equivalence}.




Our bialgebraic setup guarantees the following general adequacy result:
\begin{theorem}[Adequacy]
	\label{th:adequacy}
	For every $\rho$-bialgebra, the underlying $\Sigma$-algebra is adequate.
\end{theorem}
\begin{proof}
	Let $(X, a,c)$ be a $\rho$-bialgebra, and let
	\begin{itemize}
		\item $\llbracket-\rrbracket$ be the unique $\Sigma$-algebra morphism from $\mS$ to $(X,a)$;
		\item $\inj$ be the unique $B(\mS,-)$-coalgebra morphism from $(X,B(\llbracket-\rrbracket,\id)\cdot c)$ to $\nu B(\mS,-)$.
	\end{itemize}
	Denoting the final $B(\mS,-)$-coalgebra by $(Z,\zeta)$, we thus have the following diagram:
	\[\begin{tikzcd}[column sep=3.5em]
			\mS & X & Z \\
			& {B(X,X)} \\
			{B(\mS,\mS)} & {B(\mS,X)} & {B(\mS,Z)}
			\arrow["{\llbracket-\rrbracket}", from=1-1, to=1-2]
			\arrow["{\gamma}"', from=1-1, to=3-1]
			\arrow["\inj", from=1-2, to=1-3]
			\arrow["{c}", from=1-2, to=2-2]
			\arrow["\zeta", from=1-3, to=3-3]
			\arrow["{B(\llbracket-\rrbracket,\, \id)}", from=2-2, to=3-2]
			\arrow["{B(\id,\,\llbracket-\rrbracket)}"', from=3-1, to=3-2]
			\arrow["{B(\id,\,\inj)}"', from=3-2, to=3-3]
		\end{tikzcd}\]
	The left cell commutes because the $\Sigma$-algebra morphism $\llbracket -
		\rrbracket$ is also a morphism $\llbracket - \rrbracket\colon (\mS,\ini,\gamma)\to (A,a,c)$ of $\rho$-bialgebras
	by \Cref{prop:initial-bialgebra}. The right cell commutes because $\inj$ is a $B(\mS,-)$-coalgebra morphism by definition. Thus the outside commutes, which shows  that $\inj \cdot
		\llbracket - \rrbracket$ is a $B(\mS,-)$-coalgebra morphism from $(\mS,\gamma)$ to $(Z,\zeta)$. Since $\beh$ is, by definition, the unique such morphism, it follows that $\beh=\inj\cdot \llbracket-\rrbracket$, whence the algebra $(X,a)$ is adequate.
\end{proof}

In particular, this theorem applies to the denotational $\rho$-model derived from a locally final coalgebra
(\Cref{th:denotational-model}), and so we get:
\begin{corollary}
	\label{corr:adequacy}
	Let $\rho$ be a relatively flat higher-order GSOS law, and let $z\colon Z\to B(Z,Z)$ be a locally final coalgebra. Then the $\Sigma$-algebra $(Z,a_\rho)$ is adequate.
\end{corollary}

Note that the corresponding result in first-order abstract GSOS~\cite{Turi97} is trivial: here $B$ is an endofunctor, so denotational and behavioural equivalence simply coincide ($Z=\nu B$ and $\den{-}=\beh$).

\begin{example}[\stsc]\label{ex:stsc-adequate}
	For the behaviour bifunctor $B$ for \stsc, behavioural equivalence corresponds to \emph{strong applicative bisimilarity}~\cite{Abramsky90}. The latter is the greatest equivalence relation $\sim \,= (\sim_\tau\,\seq \Tr_\tau\times \Tr_\tau)_{\tau\in \Ty}$ on the $\Ty$-sorted set $\Tr$ of program terms such that:
	\begin{itemize}
		\item if $p\sim_\utype q$ and $p\to p'$, then there exists $q'$ with $q\to q'$ and $p'\sim_\utype q'$;
		\item if $p\sim_\utype q$ and $p\xto{\checkmark}$, then $q\xto{\checkmark}$;
		\item if $p\sim_{\arty{\tau_1}{\tau_2}} q$ and $p\to p'$, then there exists $q'$ with $q\to q'$ and $p'\sim_{\arty{\tau_1}{\tau_2}} q'$;
		\item if $p\sim_{\arty{\tau_1}{\tau_2}} q$ and $p\xto{t} p'$, then there exists $q'$ with $q\xto{t} q'$ and $p'\sim_{\tau_2} q'$.
	\end{itemize}
	The general adequacy result of \Cref{corr:adequacy} thus implies that denotational equivalence w.r.t.~the $\rho$-model constructed in \Cref{ex:stsc} is adequate for strong applicative bisimilarity:
	\[ \forall \tau\in \Ty.\, \forall  p,q\in \Tr_\tau.\, \llbracket p \rrbracket_\tau = \llbracket q \rrbracket_\tau \implies p\sim_\tau q. \]
\end{example}

\begin{remark}
	For every endofunctor $F$ that preserves weak pullbacks, behavioural equivalence w.r.t.\ the final coalgebra $\nu F$ corresponds to strong coalgebraic bisimilarity~\cite{Rutten2000}. This holds in particular for the endofunctor $B(\mS,-)$, which preserves weak pullbacks.
\end{remark}

\section{Constructing Locally Final
  Coalgebras}\label{sec:locally-final-coalgebras}

Our bialgebraic approach to denotational semantics presented above is based on locally final coalgebras for the
behaviour bifunctor $B\colon \C^{\op} \times \C \to \C$, but so far leaves open
when such coalgebras actually exist, and how to construct them. In this
section, we provide a general construction method.

\subsection{M-Categories}
We construct locally final coalgebras using a variant of \emph{M-categories}, a type of categories
enriched over ultrametric spaces introduced by~\citet{Birkedal10}.
$M$-categories provide a simple and natural setting for fixed point
constructions of functors.

\paragraph{Ultrametric Spaces}\label{sec:ultrametric-spaces}
We first recall some terminology from the theory of metric spaces~\cite{smyth90}. A metric space $(X, d\colon X^2 \to \mathbb{R}_{\geq 0})$ is an
\emph{ultrametric} space if it satisfies the strong triangle inequality
\begin{displaymath}
	\forall x, y, z \in X .\, d(x, z) \leq \max(d(x, y), d(y, z)).
\end{displaymath}
A metric space is \emph{1-bounded} if $d(x,y) \leq 1$ for all $x,y\in X$. The \emph{product space} of two metric spaces $(X,d_X)$ and $(Y,d_Y)$ is the space  \[(X,d_X)\times (Y,d_Y)=(X\times Y, d)\qquad\text{where}\qquad {d}((x,y),(x',y'))=\max \{ d_X(x,x'), d_Y(y,y') \}.\]
A map $f: X \to Y$ from a metric space $(X,d_X)$ to a metric space $(Y, d_Y)$
is \emph{non-expansive} when
\begin{displaymath}
	\forall x, y \in X .\, d_Y(f(x), f(y)) \leq d_X(x, y).
\end{displaymath}
A stronger condition is given by \emph{contractivity}, that is,
\begin{displaymath}
	\exists c \in [0,1) .\, \forall x, y \in X .\,
	d_Y(f(x), f(y)) \leq c \cdot d_X(x, y).
\end{displaymath}

A key construction in our theory is Banach's fixed-point theorem
\cite{Granas03}, which applies to \emph{complete} metric spaces. Recall that a
sequence $(x_n)_{n \in \N}$ in a metric space $(X,d)$ is a \emph{Cauchy
	sequence} if
\begin{displaymath}
	\forall \epsilon > 0.\, \exists N \in \N.\, \forall m, n \geq N .\,
	d(x_m, x_n) < \epsilon.
\end{displaymath}
A metric space is \emph{complete} if
every Cauchy sequence in it converges.

\begin{theorem}[Banach]
	Every contractive function on a complete, inhabited\,\footnote{
		Birkedal et al.\,\cite{Birkedal10} use phrases such as ``$X$ is non-empty'' ($\neg \forall x
			.\, x \notin X$) throughout their presentation. We choose ``$X$ is
		inhabited'' ($\exists x .\, x \in X$). This change makes Banach's
		fixed-point theorem, as well as the results from \cite{Birkedal10} we are
		going to use, constructive. To a classical mathematician, the two predicates
		are of course logically equivalent.}
	metric space has a unique fixed point.
\end{theorem}

\paragraph{M-Categories} We consider categories that assign to any two parallel morphisms a notion of distance given by an ultrametric, compatible with composition:

\begin{definition}[M-category]\label{def:m-category}
	An \emph{M-category} is a category $\C$ where
	\begin{itemize}
		\item each hom-set $\C(A, B)$ is equipped with the structure of a complete, 1-bounded
		      ultrametric space;
		\item each composition function $-\cdot- : \C(B, C) \times \C(A, B) \to \C(A, C)$ is
		      non-expansive.
	\end{itemize}
	We denote the ultrametric on $\C(A,B)$ by $d_{A,B}$ or simply $d$.
\end{definition}

\begin{remark}
	Our present definition of M-category is slightly more permissive than that of \citet{Birkedal10}, who additionally require all hom-sets of an M-category to be inhabited.
\end{remark}

\begin{example}[$\Set^\Ty$]\label{ex:m-cat-set-ty}
	The category $\Set^\Ty$ used for the categorical modelling of \stsc carries the structure of an M-category as follows. We define the \emph{complexity} $|\tau|$ of a type $\tau\in \Ty$ inductively by
	\[ |\utype{|}=1\qquad\text{and}\qquad |\arty{\tau_1}{\tau_2}| =  |\tau_1| + |\tau_2|.  \]
	Then the metric on $\Set^\Ty(A,B)$ is given by
	\begin{equation}\label{eq:set-ty-dist} d(f,g) = \sup\, \{ 2^{-(|\tau{|}-1)} \mid \text{$\tau\in \Ty$ and $f_\tau\neq g_\tau$} \}.  \end{equation}
	In other words, for all $n\in \Nat$ we have
	\[ d(f,g)\leq 2^{-n} \qquad\text{iff}\qquad \text{$f_\tau=g_\tau$ for all types $\tau$ with $|\tau|\leq n$}. \]
	One readily verifies that $(\Set^\Ty(A,B),d)$ forms a complete, $1$-bounded ultrametric
	space. To show that composition is non-expansive, let $f,g\colon A\to B$ and
	$h,k\in B\to C$. We need to show
	\begin{equation}\label{eq:pf-loc-nonexp} d(h\cdot f,k\cdot g)\leq \max \,\{ d(f,g), d(k,h) \}. \end{equation}
	Let $n\in \Nat$ and suppose that $d(f,g),d(h,k)\leq 2^{-n}$. This means that
	\[ f_\tau=g_\tau \quad\text{and}\quad h_\tau = k_\tau \qquad \text{for all $\tau\in \Ty$ with $|\tau|\leq n$}.  \]
	Therefore
	\[ (h\cdot f)_\tau = h_\tau \cdot f_\tau = k_\tau \cdot g_\tau = (k\cdot g)_\tau   \qquad \text{for all $\tau\in \Ty$ with $|\tau|\leq n$},  \]
	so $d(h\cdot f,k\cdot g)\leq 2^{-n}$. This proves \eqref{eq:pf-loc-nonexp}.
\end{example}


To construct fixed points of functors on M-categories, we need them to respect the enrichment:

\begin{definition}
	A functor $F\colon \B\to \C$ between M-categories is
	\begin{enumerate}
		\item \emph{locally non-expansive} if each of the maps \begin{displaymath}
			      F_{X,Y} : \B(X,Y) \to \C(FX,FY) \qquad (X,Y\in \C)
		      \end{displaymath}
		      between hom-sets is non-expansive, and
		\item	\emph{locally contractive} if there exists $c \in [0, 1)$ such that
		      each map $F_{X,Y}$ is contractive with factor $c$.
	\end{enumerate}
\end{definition}

\begin{remark}
	We stress that the contractivity factor in (2) is required to be global: all the maps $F_{X,Y}$ are contractive with the same factor $c$, rather than having distinct factors for each.
\end{remark}

In an M-category, the notion of a Cauchy sequence can be categorified:

\begin{definition}
	An \emph{increasing Cauchy tower} in an M-category is a diagram of the form
	\[\begin{tikzcd}
			{A_0} & {A_1} & \dots & {A_n} & \dots
			\arrow["{f_0}", shift left, from=1-1, to=1-2]
			\arrow["{g_0}", shift left, from=1-2, to=1-1]
			\arrow["{f_1}", shift left, from=1-2, to=1-3]
			\arrow["{g_1}", shift left, from=1-3, to=1-2]
			\arrow["{f_{n - 1}}", shift left, from=1-3, to=1-4]
			\arrow["{g_{n - 1}}", shift left, from=1-4, to=1-3]
			\arrow["{f_n}", shift left, from=1-4, to=1-5]
			\arrow["{g_n}", shift left, from=1-5, to=1-4]
		\end{tikzcd}\]
	such that $g_n \cdot f_n = 1_{A_n}$ for all $n$, and $\lim_{n \to \infty} d(f_n
		\cdot g_n, \id_{A_{n + 1}}) = 0$. An M-category has \emph{inverse
		limits of Cauchy towers} if, for every increasing Cauchy tower, the
	$\omega^\op$-chain given by the $g_n$ has a limit.
\end{definition}

The following result, a straightforward modification of
\cite[Thm.~3.4]{Birkedal10}, gives a sufficient criterion for the existence of
a fixed point of a contravariant functor on an M-category:
\begin{theorem}[Fixed point theorem]\label{th:m-category-fixpoint}
	Let $\C$ be an $M$-category with a terminal object and inverse limits of increasing Cauchy towers.
	Moreover, suppose that $F\colon \C^\op \to \C$ is a locally contractive
	functor and that $\C(1,F1)$ is inhabited.
	Then $F$ has a fixed point: there exists an object $X\in \C$ with an
	isomorphism
	\[ FX \cong X. \]
	Moreover, if all hom-sets of $\C$ are inhabited, the fixed point is unique
	up to isomorphism.
\end{theorem}
Of course, the condition that $\C$ has a terminal object and inverse limits of increasing Cauchy towers is satisfied whenever $\C$ is complete. This holds in all our applications.

\begin{remark}\label{rem:fixed-point-cons}
	\begin{enumerate}
		\item The proof of \cite[Thm.~3.4]{Birkedal10} provides an iterative
		      construction of a fixed point of~$F$: it is given by the inverse limit of the
		      increasing Cauchy tower
		      \begin{equation}\label{eq:fp-cons}
			      \begin{tikzcd}
				      {1} & {F1} & \dots & {F^n 1} & \dots
				      \arrow["{f_0}", shift left, from=1-1, to=1-2]
				      \arrow["{g_0}", shift left, from=1-2, to=1-1]
				      \arrow["{f_1}", shift left, from=1-2, to=1-3]
				      \arrow["{g_1}", shift left, from=1-3, to=1-2]
				      \arrow["{f_{n - 1}}", shift left, from=1-3, to=1-4]
				      \arrow["{g_{n - 1}}", shift left, from=1-4, to=1-3]
				      \arrow["{f_n}", shift left, from=1-4, to=1-5]
				      \arrow["{g_n}", shift left, from=1-5, to=1-4]
			      \end{tikzcd}\end{equation}
		      where $1$ is the terminal object,  $F^n 1=FF\cdots F1$ ($n$ applications of $F$), and the connecting morphisms are defined inductively as follows:
		      \begin{itemize}
			      \item $f_0$ is an arbitrary morphism (which exists by the assumption that $\C(1,F1)$ is inhabited), and $g_0$ is the unique morphism into the terminal object;
			      \item $f_{n+1}=Fg_n$ and $g_{n+1}=Ff_n$.
		      \end{itemize}
		\item \cite[Thm.~3.4]{Birkedal10} shows more generally that, for every
		      $M$-category $\C$ with a terminal object and inverse limits of increasing
		      Cauchy towers, every locally contractive \emph{bi}functor $F\colon
			      \C^\op\times \C\to \C$ has a fixed point $F(X,X)\cong X$.
		      \Cref{th:m-category-fixpoint} is the special case where $F$ ignores
		      the second argument. Let us  note that \emph{op.~cit.}\ works with a
		      stricter notion of M-category where all hom-sets of $\C$ are
		      required to be inhabited; however, the proof of existence of a fixed
		      point only uses that $\C(1,F(1,1))$ (or $\C(1,F1)$ in our case) is
		      inhabited, which allows one to seed the fixed point iteration with some $f_0$.
	\end{enumerate}
\end{remark}

\subsection{Locally Final Coalgebras in M-Categories}

The fixed point theorem for M-categories (\Cref{th:m-category-fixpoint}) is the
key to our construction of locally final coalgebras for behaviour bifunctors.
As a prerequisite, we recall how to iteratively construct final coalgebras for
\emph{endo}functors; see~\citet{Adamek_Milius_Moss_2025}
for a comprehensive account of the subject.

\begin{remark}[Final sequence]\label{rem:final-sequence}
	Given an ordinal number  $\alpha>0$ viewed as a poset, an \emph{$\alpha$-chain} is a diagram with scheme $\alpha$, and an \emph{$\alpha^\op$-chain} is a diagram with scheme $\alpha^\op$. For every category~$\C$ with a terminal object $1$ and limits of chains (i.e.\ $\alpha^\op$-chains for all $\alpha$), the \emph{final sequence} of an endofunctor $H\colon \C\to \C$ is the unique diagram $D\colon \mathbf{Ord}^\op\to \C$ (where  $\mathbf{Ord}$ is the large poset of all ordinals) with objects~$D_\alpha$ and connecting morphisms $D_{\alpha,\beta}\colon D_\alpha\to D_\beta$ for $\beta\leq \alpha$ defined as follows:
	\begin{itemize}
		\item $D_0=1$, $D_1=H1$, and $D_{1,0}\colon H1\to 1$ is the unique morphism.
		\item $D_{\alpha+1}=HD_\alpha$ for every ordinal $\alpha$, and $D_{\alpha+1,\beta+1} = HD_{\alpha,\beta}\colon D_{\alpha+1}\to D_{\beta+1}$ for all $\beta\leq \alpha$.
		\item If $\alpha$ is a limit ordinal, $D_\alpha$ is the limit of the
		      $\alpha^\op$-chain $(D_\beta)_{\beta<\alpha}$ with projections
		      $D_{\alpha,\beta}\colon D_\alpha\to D_\beta$.
	\end{itemize}
	The final sequence \emph{converges} (\emph{in $\alpha$ steps}) if $D_{\alpha+1,\alpha}$ is an isomorphism for some $\alpha$. In this case, \[D_{\alpha+1,\alpha}^{-1}\colon D_\alpha\xto{\cong} HD_\alpha\] is a final coalgebra for $H$~\cite[Thm.~6.5.4]{Adamek_Milius_Moss_2025}. Note
	that convergence in $\alpha$ steps implies convergence in~$\alpha'$ steps for
	every $\alpha'\geq \alpha$. Two useful sufficient criteria for convergence of the
	final sequence are:
	\begin{enumerate}
		\item The functor $H$ preserves limits of $\alpha^\op$-chains for some $\alpha$. In this case,
		      the final sequence converges in $\alpha$
		      steps~\cite[Cor.~6.5.5]{Adamek_Milius_Moss_2025}.
		\item The category $\C$ is locally presentable, and the functor $H$ preserves
		      monomorphisms and colimits of $\alpha$-chains of strong monomorphisms for some
		      regular cardinal $\alpha$~\cite[Thm.~11.4.15]{Adamek_Milius_Moss_2025}.
	\end{enumerate}
	For more information on locally presentable categories, see~\citet{AdamekRosicky1994}. For our purposes, it is enough to know that every presheaf category $\C=\Set^\B$ for a small category $\B$ is locally presentable~\cite[Thm.~1.46]{AdamekRosicky1994}. All categories arising in our applications are of this form.
\end{remark}

For our construction of locally final coalgebras, we need an extra
condition on M-categories:

\begin{definition}[Stable M-category]
	An M-category $\C$ is \emph{stable} if for every parallel pair $f,g\colon A\to
		B$ of morphisms and every jointly monic\footnote{An inhabited family $h_i\colon B\to C_i$ ($i\in I$) is \emph{jointly monic} if for  $f,g\colon A\to B$ with $h_i\cdot f = h_i\cdot g$ for all $i\in I$, one has $f=g$.} family $h_i\colon B\to C_i$ ($i\in I$), one has
	\[ d_{A,B}(f,g)=\sup_{i\in I} d_{A,C_i}(h_i\cdot f, h_i\cdot g). \]
\end{definition}
Note that the ``$\geq$'' relation holds in every M-category since composition is non-expansive.
\begin{example}[$\Set^\Ty$]\label{ex:stable-set-ty}
	The $M$-category $\Set^\Ty$ (\Cref{ex:m-cat-set-ty}) is stable. To see this, let $f,g\colon A\to B$ and let $h_i\colon B\to C_i$ ($i\in I$) be a jointly monic family in $\Set^\Ty$. We need to prove
	\begin{equation}\label{eq:pf-stable-set-ty} d(f,g)\leq \sup_{i\in I} d(h_i\cdot f, h_i\cdot g). \end{equation}
	Thus let $n\in \Nat$ and suppose that $d(h_i\cdot f,h_i\cdot g)\leq 2^{-n}$ for all $i\in I$, that is,
	\[ (h_i)_\tau \cdot f_\tau = (h_i)_\tau \cdot g_\tau \qquad \text{for all $i\in I$ and all $\tau\in \Ty$ with $|\tau|\leq n$}.  \]
	Since the family $(h_i)$ in $\Set^\Ty$ is jointly monic, so is the
	family $(h_i)_\tau$ in $\Set$ for each $\tau\in \Ty$, whence
	\[ f_\tau = g_\tau \qquad \text{for all $\tau\in \Ty$ with $|\tau|\leq n$}.  \]
	This proves $d(f,g)\leq 2^{-n}$ and therefore
	\eqref{eq:pf-stable-set-ty}.
\end{example}
With these preparations, we obtain our general existence result for locally final coalgebras:

\begin{theorem}[Existence of locally final coalgebras]\label{thm:lfc-existence}
	Let $\C$ be a stable $M$-category with a terminal object and limits of chains. Moreover, suppose that $B\colon \C^{\op} \times \C \to \C$ is a bifunctor such that
	\begin{enumerate}
		\item\label{thm:ass1} $B$ is locally contractive in the first argument, and locally non-expansive in the second argument;
		\item\label{thm:ass2} the hom-set $\C(1,B(1,1))$ is inhabited;
		\item\label{thm:ass3} for each object $X\in \C$ the final sequence of the endofunctor $B(X,-)$ converges.
	\end{enumerate}
	Then $B$ has a locally final coalgebra. If all hom-sets of $\C$ are
	inhabited, it is unique up to isomorphism.
\end{theorem}
\begin{remark}\label{rem:cons-lfc}
	A detailed proof is given in the appendix\ifarx{}{ of the extended arXiv
		version~\cite{Goncharov2026arXiv}}. The proof shows that the formation of the final $B(X,-)$-coalgebra for each object $X\in \C$  yields a contravariant functor
	\[ F\colon \C^\op\to \C,\qquad X \mapsto \nu B(X,-), \]
	and moreover that (1) fixed points of $F$ correspond precisely to locally final coalgebras for $B$, and (2) $F$ satisfies the assumptions of the fixed point theorem
	(\Cref{th:m-category-fixpoint}). These observations imply the existence and uniqueness statements for locally final coalgebras. Specifically, by \Cref{rem:fixed-point-cons}, a
	locally final coalgebra for $B$ can be constructed as the limit of the increasing
	Cauchy tower \eqref{eq:fp-cons}.
\end{remark}

\section{Case Studies}\label{sec:case-studies}
In this section we demonstrate how the theory developed in
\Cref{sec:denotational-models,sec:locally-final-coalgebras} instantiates to a
number of typed and untyped higher-order languages.

\subsection{Typed Combinatory Logic}
First, we revisit our running example, the language \stsc introduced in
\Cref{sec:ho-gsos}. We have constructed the locally final coalgebra underlying
its bialgebraic denotational model in \Cref{ex:stsc}. We now show that this ad hoc
construction is captured by the general approach of
\Cref{sec:locally-final-coalgebras} using M-categories. Recall that $\Set^\Ty$
is a stable M-category (\Cref{ex:m-cat-set-ty,ex:stable-set-ty}). Additionally,
we observe that the behaviour bifunctor $B\colon (\Set^\Ty)^{\op} \times \Set^\Ty \to \Set^\Ty$ for \stsc given by \eqref{eq:beh} satisfies all the conditions required for existence of a locally final coalgebra:

\begin{lemma}\label{lem:thm-cond-set-ty}
	\begin{enumerate}
		\item\label{thm:ass1-set-ty} $B$ is locally contractive in the first and locally non-expansive in the second argument.
		\item\label{thm:ass2-set-ty} The hom-set $\Set^\Ty(1,B(1,1))$ is inhabited.
		\item\label{thm:ass3-set-ty} For each $X\in \Set^\Ty$ the final sequence of the endofunctor $B(X,-)$ converges.
	\end{enumerate}
\end{lemma}

We thus obtain an adequate denotational semantics for \stsc from our general theory; recall that here behavioural equivalence is strong applicative bisimilarity (\Cref{ex:stsc-adequate}):
\begin{theorem}[Denotational $\rho$-model for \stsc]\label{thm:xtcl-lfc}
	The behaviour bifunctor $B$ for \stsc has a locally final coalgebra, unique up to isomorphism. It extends to a denotational $\rho$-model that is adequate for strong applicative bisimilarity.
\end{theorem}

\begin{proof}
	Existence follows from
	\Cref{thm:lfc-existence} together with \Cref{lem:thm-cond-set-ty}. Regarding uniqueness, we note that not all hom-sets
	of  $\Set^\Ty$ are inhabited, so the uniqueness statement of
	\Cref{thm:lfc-existence} does not apply directly. However, we can restrict $B$ to
	a bifunctor $B_{+}\colon (\Set_{+}^\Ty)^\op\times \Set_{+}^\Ty \to
		\Set_{+}^\Ty$, where $\Set_{+}\hookrightarrow \Set$ is the full subcategory of
	\emph{inhabited} sets. The category $\Set_{+}^\Ty$ has inhabited hom-sets, and $B_{+}$ satisfies all the conditions of \Cref{thm:lfc-existence} since $B$ does. It follows that the bifunctor~$B_{+}$ has a unique locally final coalgebra up to isomorphism, which implies that the same holds for $B$. The second part of the theorem follows from \Cref{ex:stsc-flat} and \Cref{corr:adequacy}.
\end{proof}

The explicit description of the denotation map $\sem{p}$ is given in \Cref{ex:stsc}.

It is instructive to see how the concrete description of the locally final coalgebra $z\colon Z\xto{\cong} B(Z,Z)$ given in \Cref{ex:stsc} arises from the general iterative construction given by \Cref{rem:cons-lfc}. An easy induction shows that the functor $FX=\nu B(X,-)$ used in the fixed point iteration satisfies:
\begin{lemma}\label{lem:fixed-point-iteration-xTCL}
	$(F^n1)_\tau = Z_\tau$ for all $\tau\in \Ty$ and all $n\geq |\tau|$. \end{lemma}
This immediately implies that $Z$ is the inverse limit of the increasing
Cauchy tower $(F^n1)$.

\subsection{Probabilistic Typed Combinatory Logic}\label{sec:xPTCL}
Thanks to the modularity of the abstract bialgebraic setting, it is not difficult to
move from the deterministic language \stsc to languages with additional
computational effects. This amounts to simply combining the behaviour bifunctor
of \stsc with a suitable monad. For illustration, we consider a probabilistic version \xPTCL of \stsc. We extend the
signature $\Sigma$ of \stsc to the signature $\Sigma^\prob$ that additionally
contains the constructor $\oplus_\tau\colon \tau\times \tau\to \tau$ for each $\tau\in \Ty$. Intuitively, the program $p\oplus_\tau q$ flips a fair coin and then executes $p$ or $q$ with probability $\frac{1}{2}$ each. Since this makes $\beta$-reductions probabilistic, we extend the behaviour bifunctor $B$ given by \eqref{eq:beh} to
\begin{equation*}
	B^\prob\c (\Set^{\Ty})^{\opp} \times \Set^{\Ty} \to \Set^{\Ty}, \quad B_\utype^\prob(X,Y)=\Df(Y_\utype) +1,\quad B^\prob_{\arty{\tau_1}{\tau_2}}(X,Y) = \Df(Y_{\arty{\tau_1}{\tau_2}}) + Y_{\tau_2}^{X_{\tau_1}}.
\end{equation*}
Here $\Df\colon \Set\to \Set$ is the finite distribution functor sending a set $X$ to the set $\Df X$ of all finite probability distributions on $X$ (i.e.\ functions $\phi\colon X\to [0,1]$ such that $\phi(x)=0$ for all but finitely many $x\in X$, and $\sum_{x\in X} \phi(x)=1$). It is often convenient to represent a distribution $\phi\in \Df X$ as a formal sum $\sum_{i\in I} r_i\cdot x_i$ where $I$ is a set, $x_i\in X$ and $r_i\in [0,1]$ for all $i\in I$, and $\sum_{x_i=x} r_i = \phi(x)$ for all $x\in X$. The operational semantics of \xPTCL is given by a higher-order coalgebra
\begin{equation}\label{eq:opmodel-xPTCL} \gamma^\prob\colon \mS^\prob \to B^\prob(\mS^\prob,\mS^\prob). \end{equation}
The transitions for combinators are the same labelled transitions as for \stsc (\Cref{fig:skirules}), and the transitions for $\oplus$ and application are given by \Cref{fig:probskirules}. The superscripts are transition probabilities $r\in [0,1]$,  and the rules are additive; thus, the rules for application overall express that
\[ \text{if$\quad$ $\gamma^\prob_{\arty{\tau_1}{\tau_2}}(p)= \sum_{i\in I} r_i \cdot p_i$ $\quad$then$\quad$ $\gamma_{\tau_2}^\prob(p\, q) = \sum_{i\in I} r_i\cdot (p_i\, q)$}.\]
The coalgebra \eqref{eq:opmodel-xPTCL} is the operational model of the higher-order GSOS law
\[ \rho^\prob_{X,Y}\colon \Sigma^\prob(X\times B^\prob(X,Y)) \to
	B^\prob(X,(\Sigma^\prob)^{\star}(X+Y)) \]
that, as usual, encodes the operational rules of \Cref{fig:probskirules} into a family of functions; for instance, on the constructors $\oplus$ and $\mathsf{app}$ the law is given as follows (an underscore ``$\_$'' means ``don't care''):
\begin{align*}
	\rho_{X,Y}^\prob((p,\_)\oplus_{\tau} (q,\_))                                          & = \frac{1}{2}\cdot p + \frac{1}{2}\cdot q                     \\
	\rho^\prob_{X,Y}(\mathsf{app}_{\tau_1,\tau_2}((p,\sum_{i\in I} r_i \cdot p_i),(q,\_)) & = \sum_{i\in I} r_i\cdot \mathsf{app}_{\tau_1,\tau_2}(p_i, q) \\
	\rho^\prob_{X,Y}(\mathsf{app}_{\tau_1,\tau_2}((p,f),(q,\_))                           & = f(q).
\end{align*}
As before, this law is relatively flat on $J = \{0 < 1\}$ if we have
$\Sigma^\prob_0$ contain application only, and $\Sigma^\prob_1$ all other
constructors. It is easy to see that the extended behaviour bifunctor $B^\prob$ still satisfies all conditions of \Cref{thm:lfc-existence}, and so we get an adequate denotational semantics for \xPTCL:
\begin{theorem}[Denotational model for \xPTCL]\label{thm:xptcl-lfc}
	The behaviour bifunctor $B^\prob$ for \xPTCL has a locally final coalgebra, unique up to isomorphism. It extends to a denotational model that is adequate for strong probabilistic applicative bisimilarity.
\end{theorem}
Here \emph{strong probabilistic applicative bisimilarity} is behavioural equivalence on the coalgebra \eqref{eq:opmodel-xPTCL} of programs. It is given by the greatest $\Ty$-indexed equivalence relation $\sim^\prob$ on the $\Ty$-sorted set $\mS^\prob$ of programs such that for every $\tau\in \Ty$ and  $p\sim^\prob_\tau q$, one of the following three conditions holds:
\begin{itemize}
	\item $\gamma^\prob(p),\gamma^\prob(q)\in \Df(\mS^\prob_\tau)$ and $\sum_{t\in E} \gamma^\prob(p)(t)=\sum_{t\in E} \gamma^\prob(q)(t)$ for every equivalence class $E$ of $\sim^\prob_\tau$;
	\item $\tau=\utype$ and $\gamma^\prob(p)=\gamma^\prob(q)=*$;
	\item $\tau=\arty{\tau_1}{\tau_2}$, $\gamma^\prob(p),\gamma^\prob(q)\colon \mS^\prob{\tau_1}\to \mS^\prob{\tau_2}$, and $\gamma^\prob(p)(t)\sim^\prob_{\tau_2}\gamma^\prob(q)(t)$ for every $t\in \mS_{\tau_1}^\prob$.
\end{itemize}
Thus, bisimilar programs $\beta$-reduce into each bisimilarity class with the same probability, corresponding to the standard notion of bisimilarity on discrete probabilistic transition systems~\cite{LarsenSkou1991,vr99}.

\begin{figure*}[t]
	\begin{gather*}
		\inference{}{p \oplus_\tau q  \to^{1/2} p}
		\quad
		\inference{}{p \oplus_\tau q  \to^{1/2} q}
		\quad
		\inference{p\to^r p'}{p\, q \to^{r} p'\, q}
		\quad
		\inference{p\xto{q} p'}{p\, q \to^{r} p'}
	\end{gather*}
	\caption{(Call-by-name) operational semantics of \xPTCL.}
	\label{fig:probskirules}
\end{figure*}

\subsection{Untyped Combinatory Logic}
\label{subsec:xcl}

We next show how to apply our framework to \emph{untyped} higher-order languages. For that purpose we consider the language \xcl, a.k.a.~
$\SKI_u$~\cite{Goncharov23}, an untyped combinatory logic  that is computationally equivalent to the untyped call-by-name $\lambda$-calculus. From a categorical perspective, the crucial difference between \stsc
and \xcl is that the latter normally lives in $\Set$ (and not $\Set^{\Ty}$), a
mathematical universe that lacks the inherent structural properties of types to
allow for the existence of suitable denotational domains. To solve this issue,
we model the operational semantics of \xcl in a \emph{guarded} fashion in the
\emph{topos of trees}\,\cite{Birkedal12}. The latter is the category
$\Set^{\omega^\op}$ of $\omega^\op$-chains in $\Set$, or equivalently (contravariant) presheaves on the poset $\omega=\{0,1,2,\ldots\}$, and is an important instance of a category that admits locally final
coalgebras for many bifunctors.

\paragraph{The Topos of Trees} We present an object $X\in \topos$ as an indexed family  $(X_{n})_{n \in \omega}$ of sets equipped with \emph{restriction maps} $X_{n,m} : X_{n} \to X_{m}$ for $n\geq m$. Morphisms of $\topos$ are natural transformations. Being a presheaf topos,
$\topos$ is complete and cartesian closed. The exponential~$Y^{X}$ is given at stage $n$ by the
set $(Y^X)_n$ of all $n+1$-tuples  $(f_{i}\colon X_i\to Y_i)_{0\leq
			i\leq n}$ of maps that commute with the restriction maps \cite[Sec.~2]{Birkedal12}.
For any functor $F$ with codomain $\topos$, we put $F_n X=(FX)_n$.

The key construction in $\topos$ that enables the existence of locally final
coalgebras comes in the form of the \emph{later modality}~\cite[Sec.~2.1]{Birkedal12}, defined by
\[
	\later :
	\topos \to \topos,\qquad \later_{0} X  = 1,\qquad \later_{n+1} X = X_{n}.
\]
The later modality is used to \emph{guard} inputs and
outputs, which is essentially a form of \emph{step-indexing}, thus enforcing
contractivity in the behaviour of our language.
Analogous to $\Set^\Ty$, the topos of trees forms a stable M-category with respect to the ultrametric on $\topos(X,Y)$ given as follows:
\begin{equation}
	d_{X,Y}(\mu,\nu) = \sup\, \{ 2^{-n}\mid n\in \omega\text{ and } \mu_n \neq \nu_n \}.
\end{equation}

\paragraph{The Language \xcl}

The syntax of \xcl is generated by the following grammar:
\begin{equation}
	\label{eq:syntaxxcl}
	p, q \Coloneqq S \mid K \mid I \mid S'(p) \mid K'(p) \mid S''(p,q) \mid p\, q.
\end{equation}

\vspace{0.2cm}
The operational rules of \xcl are simply an untyped version of the rules of \stsc:
\begin{gather*}
	\inference{}{S\xto{t}S'(t)}
	\qquad
	\inference{}{S'(p)\xto{t}S''(p,t)} \qquad
	\inference{}{S''(p,q)\xto{t}(p\app{\tau_{1}}{\arty{\tau_{2}}{\tau_{3}}} t)\app{\tau_{2}}{\tau_{3}} (q\app{\tau_{1}}{\tau_{2}} t)} \\[1ex]
	\inference{}{K\xto{t}K'(t)}
	\qquad
	\inference{}{K'(p)\xto{t}p}
	\qquad
	\inference{}{I\xto{t}t}
	\qquad
	\inference{p\to p'}{p \app{\tau_{1}}{\tau_{2}} q\to p' \app{\tau_{1}}{\tau_{2}} q}
	\qquad
	\inference{p\xto{q} p'}{p \app{\tau_{1}}{\tau_{2}} q\to p'}
\end{gather*}

\paragraph{Categorical Modelling of \xcl}

The grammar \eqref{eq:syntaxxcl} corresponds to the polynomial endofunctor
\begin{equation*}
	\Sigma^{\untyped,0}\colon
	\Set \to \Set, \qquad {\Sigma^{\untyped,0}(X)} = \underbrace{1}_{S}+\underbrace{1}_{K}+\underbrace{1}_{I} + \underbrace{X}_{S'}+\underbrace{X}_{K'} +\, \underbrace{X\times X}_{S''} \,+\,  \underbrace{X\times X}_{\mathsf{app}}.
\end{equation*}

The initial algebra $\mu\Sigma^{\untyped,0}$ is given by the set of
closed \xcl-terms. However, recall that we intend to model the semantics of
\xcl in $\topos$ and not in $\Set$. Therefore, we use the endofunctor
\begin{equation}
	\label{eq:sigma}
	\Sigma^\untyped : \topos \to \topos, \qquad \Sigma^\untyped_{n} X = \Sigma^{\untyped,0} X_{n}.
\end{equation}
Its initial algebra $\mS^\untyped$ is the constant
presheaf given by $\mS^\untyped_{n} = \mS^{\untyped,0}$ for every $n \in
	\omega$.

The standard modelling of \xcl as a higher-order GSOS law in
$\Set$\,\cite[Sec.~3.2]{Goncharov23} involves the behaviour bifunctor
$B^{\untyped, 0}(X,Y) = Y + Y^{X}$, with the two summands $Y$ and $Y^X$ representing a $\beta$-reduction and the behaviour of combinators, respectively. Unlike for $\Sigma^{\untyped,0}$, it is not possible to simply
extend $B^{\untyped, 0}$ to $\topos$ in a pointwise manner, as this would not yield a
bifunctor that is locally contractive in the contravariant argument, as required for the construction of a locally final coalgebra (\Cref{thm:xcl-lfc}). Instead, we make use of the later modality and work with the bifunctor
\begin{equation}
	\label{eq:bifunctorxcl}
	B^\untyped \c (\topos)^{\opp} \times \topos \to \topos,\qquad B^\untyped(X,Y) = \later Y + \later Y^{\later X}.
\end{equation}

\begin{remark}
	The idea is that the later modality guards negative occurrences/inputs ($\later
		X$) and positive occurrences/outputs ($\later Y$), making $B^\untyped(X,Y)$
	locally contractive in both arguments.
	Since \Cref{thm:lfc-existence} only requires local non-expansiveness in the second argument, guarding $Y$ is technically not required, but rather a sensible design choice:	guarding the second occurrence of $Y$ reflects that we want to model functions whose outputs are available no earlier than their
	inputs. Guarding the first occurrence as well makes our modelling align with previous approaches to step-indexing for guarded languages, where local contractivity in \emph{both} $X$ and $Y$ is needed~\cite[Sec.~2.6]{Birkedal12}.
\end{remark}
The operational semantics of \xcl is modelled by a higher-order GSOS law
\[ 	\rho_{X,Y}^\untyped \c \Sigma^\untyped(X \times B^\untyped(X,Y)) \to B^\untyped(X, (\Sigma^\untyped)^\star (X+Y)) \]
of $\Sigma^\untyped$ over $B^\untyped$ that encodes the above rules. It is given at stage $n\in \omega$ by the map
\[ \rho_{X,Y,n}^\untyped\colon \Sigma^{\untyped,0}(X_n\times (\later_n Y+(\later Y^{\later X})_n))\to (\Sigma^{\untyped,0})^\star (\later_n X +\later_n Y) + (((\Sigma^{\untyped})^\star (\later X+\later Y))^{\later X})_n \]
defined below (we consider application and $S''$, the other combinators are treated analogously):
\begin{enumerate}
	\item $\mathsf{app}((p,f),(q,\_))$ (where $p,q\in X_n$ and $f\in \later_n Y+(\later Y^{\later X})_n)$ is mapped to the term $t_n\in (\Sigma^{\untyped,0})^\star (\later_n X +\later_n Y)$ defined inductively by
	      \[
		      t_0=\begin{cases}
			      \inr(*)                        & \text{if $f\in \inr(*)$} \\
			      \mathsf{app}(\inr(*), \inl(*)) & \text{if $f\in \inl(*)$}
		      \end{cases} \quad\text{and}\quad t_{n+1}=\begin{cases}
			      f_{n+1}(X_{n+1,n}(q))        & \text{if $f\in  (\later Y ^{\later X})_{n+1}$} \\
			      \mathsf{app}(f,X_{n+1,n}(q)) & \text{if $f\in Y_n=\later_{n+1}Y$}.
		      \end{cases}
	      \]
	\item $S''((p,\_),(q,\_))$ ($p,q\in X_n$) is mapped to the element $f=(f_0,\ldots,f_n)\in (((\Sigma^{\untyped})^\star (\later X+\later Y))^{\later X})_n$, where $f_k\colon {\later_k}{X} \to (\Sigma^{\untyped,0})^\star (\later_k X + \later_k Y)$ is defined inductively by
	      \begin{align*}
		       & f_0\colon 1\to (\Sigma^{\untyped,0})^\star (1 + 1),            &  & * \mapsto (\inl(*)\, \inl(*))\, (\inl(*)\, \inl(*)) \\
		       & f_{k+1}\colon X_k \to (\Sigma^{\untyped,0})^\star (X_k + Y_k), &  &
		      t \mapsto (X_{n,k}(p)\, t)\, (X_{n,k}(q) \, t).
	      \end{align*}
\end{enumerate}
As before, this law is
relatively flat for $J = \{0 < 1\}$, with $\Sigma^\untyped_0$ given by application and $\Sigma^\untyped_1$ by the other constructors.
The operational model
\begin{equation}\label{eq:xcl-model}\gamma^\untyped : \mS^\untyped \to B^\untyped(\mS^\untyped,
	\mS^\untyped)\end{equation}
of the law $\rho^\untyped$ faithfully models the operational semantics of \xcl; for example, for all $n >
	1$, the map $\gamma_{n}^\untyped$ sends an application term $p\, q \in \mS^{\untyped,0}$ to the term it
$\beta$-reduces to.

\paragraph{Denotational Semantics of \xcl}
It is not difficult to verify that the behaviour functor $B^\untyped$ satisfies the conditions of \Cref{thm:lfc-existence}, analogous to \Cref{lem:thm-cond-set-ty}. We thus obtain an adequate denotational semantics for \xcl from our general theory:

\begin{theorem}[Denotational $\rho$-model for \xcl]\label{thm:xcl-lfc}
	The behaviour bifunctor $B^\untyped$ for \xcl has a locally final coalgebra. It extends to a denotational $\rho$-model, adequate for step-indexed strong applicative bisimilarity.
\end{theorem}

Here \emph{step-indexed strong applicative bisimilarity} is behavioural equivalence on the operational model \eqref{eq:xcl-model}
w.r.t.~the final coalgebra $B^\untyped(\mS^\untyped,-)$. Explicitly, it is given by the greatest $\omega$-indexed
family $(\sim_n^\untyped \,\seq \mS^{\untyped,0} \times \mS^{\untyped,0})_{n\in
			\omega}$ of equivalence relations on the set of \xcl-terms such that:
\begin{itemize}
	\item $p\sim_{n+1}^\untyped q$ implies $p \sim_n^\untyped q$;
	\item if $p\sim^\untyped_{n+1}  q$ and $p\to p'$, then there exists $q'$ with $q\to q'$ and
	      $p'\sim^\untyped_n  q'$;
	\item if $p\sim^\untyped_{n+1} q$ and $p\xto{t} p'$, then there exists $q'$ with $q\xto{t}
		      q'$ and $p'\sim^\untyped_n q'$.
\end{itemize}
The adequacy result of \Cref{thm:xcl-lfc} thus asserts that
\[ \forall  p,q\in (\mS^{\untyped,0})^\star.\, \forall n\in \omega.\;\; \llbracket p \rrbracket_n = \llbracket q
	\rrbracket_n \implies p\sim^\untyped_n q. \]
An explicit description of the locally final coalgebra $Z\cong B^\untyped(Z,Z)$ is given by the presheaf of
\emph{guarded interaction trees} (see e.g. \cite[Sec.~3]{Frumin24} for a
similar example). At stage $0$, we have $Z_{0} = 1 + 1 \cong 2$
(representing the ``blank'' behaviour of application and combinators without
any further available information). At stage $n+1$, an element $z \in Z_{n+1}$
is either a tree at stage~$n$ (representing the abstract behaviour of application), or a tuple of functions $(f_{0} : Z_{0} \to
	Z_{0}, \dots, f_{n} : Z_{n} \to Z_{n})$ that commute with the restriction
maps (representing the behaviour of combinators). The denotation map $\den{-}_n\colon (\mS^{\untyped,0})^\star \to Z_n$ at stage $n$ sends a program $p$ to the guarded interaction tree given by the depth-$n$ unravelling of the behaviour of $p$. One readily verifies that $Z\cong B^\untyped(Z,Z)$ emerges from the inverse limit construction from \Cref{rem:cons-lfc}, analogous to the typed case (\Cref{lem:fixed-point-iteration-xTCL}).

Note that unlike for the typed cases, \Cref{thm:xcl-lfc} does not yield uniqueness of the locally final coalgebra, although we conjecture that it is also unique up to isomorphism in the present setting.

\subsection{Non-Deterministic Concurrent Combinatory Logic}
Similar to \stsc, effectful versions of \xcl are captured in our
framework by using extended behaviour functors. For variety, let us consider a
non-deterministic, concurrent extension \xnccl of \xcl, adapted from \cite{clp98,Goncharov25}. We extend the signature $\Sigma^\untyped$ of \xcl
to the signature $\Sigma^\ndc$ that additionally contains a binary non-deterministic choice operator $- \oplus -$ and a parallel composition operator $-\paral-$. The latter reduces both
arguments simultaneously if possible to achieve fairness (the rules of a non-fair language could be modelled analogously). Formally, the operational semantics of \xnccl is given by the rules of \xCL and the following rules for  $- \oplus -$ and $-\paral-$:
\begin{gather*}
	\inference{}{p \oplus q\to p}
	\qquad
	\inference{}{p \oplus q\to q}
	\qquad
	\inference{p \to p' \quad q \to q'}{p \paral q \to p' \paral q'}
	\qquad
	\inference{p \xto{r} p' \quad q \to q'}{p \paral q \to p \paral q'}
	\\[1ex]
	\inference{p \to p' \quad q \xto{r} q'}{p \paral q \to p' \paral q}
	\qquad
	\inference{p \xto{r} p' \quad q \xto{s} q'}{p \paral q \xto{t} p \,t \paral q\, t}
\end{gather*}
\begin{remark}
	A term $p \paral q$ where $p\xto{r} p'$ and $q\xto{s}{q'}$ is, in spirit, a value. The language \xnccl and all other
	languages in this section can be given a call-by-(push-)value semantics rather than the present call-by-name semantics by moving from $\Set$ to $\Set^2$ as the base category, with the two sorts representing values and computations~\cite[Sec.~6.4]{Goncharov25}. We stick to call-by-name to achieve a simpler presentation, as the denotational models of the
	call-by-value cases are no more interesting.
\end{remark}

Since the above rules introduce finitary non-determinism (for $\beta$-reductions),
we modify the behaviour bifunctor $B^\untyped$ on $\topos$ given by \eqref{eq:bifunctorxcl} to
\begin{equation*}
	B^\ndc(X,Y)=\Powf(\later Y) + (\later Y)^{\later X}
\end{equation*}
where $\Powf$ is the (stage-wise) finite power set functor.
The operational rules of \xnccl determine a higher-order coalgebra
\[ \gamma^\ndc\colon \mS^\ndc\to B^\ndc(\mS^\ndc,\mS^\ndc) \]
sending each program to its finite set of successors according to the rules; e.g.
\begin{itemize}
	\item if $p,q \in \mS^{\ndc,0}$, then $\gamma (p \oplus q) \in
		      \PSet_\omega(\later \mS^\ndc) + (\later \mS^\ndc)^{\later \mS^\ndc}$ is
	      inductively defined as
	      \[\gamma_0 (p \oplus q) = \{*\} \qquad\text{and}\qquad
		      \gamma_{m + 1} (p \oplus q) = \{p,q\};\]
	\item if $n \in \omega$, then $\gamma_n (I) = (f_0,\dots,f_n) \in
		      ((\later \mS^\ndc)^{\later \mS^\ndc})_n$, where
	      \[f_0 (I)(*) = * \qquad\text{and}\qquad f_{k + 1} (I)(p) = p.\]
\end{itemize}
This coalgebra is the operational model of the higher-order GSOS law
\[ \rho^\ndc_{X,Y}\colon \Sigma^\ndc(X\times B^\ndc(X,Y)) \to
	B^\ndc(X,(\Sigma^\ndc)^{\star}(X+Y)) \]
that encodes the operational rules in the now familiar way. At
step $n \in \omega$, it is given by the map
\[ \rho^\ndc_{X,Y,n}\colon \Sigma^{\ndc,0}(X_n\times (\PSet_\omega(\later_n Y)
	+ ({\later Y}^{\later X})_n) \to \PSet_\omega((\Sigma^{\ndc,0})^\star(\later_n X
	+ \later_n Y)) + (((\Sigma^{\ndc,0})^\star (\later X + \later Y))^{\later X})_n.\]
We give cases for the two new constructors for illustration:
\begin{itemize}
	\item $(p, \_) \oplus (q, \_)$ (where $p,q \in X_n$) is mapped to the finite
	      set $t_n \in \PSet_\omega((\Sigma^{\ndc,0})^\star(\later_n X
		      + \later_n Y))$ defined inductively by
	      \[t_0 = \{ \inl(*) \}\qquad\text{and}\qquad t_{k+1} =
		      \{X_{k+1,k}(p),X_{k+1,k}(q)\};\]
	\item $(p, U) \paral (q, g)$ (where $p,q \in X_n$, $U \in \PSet_\omega(\later_n Y)$, $g \in
		      (\later Y^{\later X})_n$) is mapped to the finite set $t_n \in
		      \PSet_\omega((\Sigma^{\ndc,0})^\star(\later_n X + \later_n Y))$
	      defined inductively by
	      \[t_0 = \{\inr(*) \paral \inl(*)\}\qquad\text{and}\qquad t_{k+1} =
		      \{p' \paral X_{k+1,k}(q) \mathbin{|} p' \in U \}.\]
\end{itemize}
As before, this law is relatively flat on $J = \{0 < 1\}$, where
$\Sigma^\ndc_0$ contains application only, and~$\Sigma^\ndc_1$ contains all other constructors.
It is easy to verify that the behaviour functor $B^\ndc$ satisfies the conditions of \Cref{thm:lfc-existence}, and so we obtain an adequate denotational semantics for \xnccl:

\begin{theorem}[Denotational $\rho$-model for $\xnccl$]\label{thm:xntcl-lfc}
	The behaviour bifunctor $B^\ndc$ for $\xnccl$ has a locally final coalgebra. It extends to a denotational $\rho$-model that is adequate for step-indexed strong applicative bisimilarity.
\end{theorem}
Here \emph{step-indexed strong applicative bisimilarity} $\sim^\ndc$ is behavioural equivalence w.r.t.~the final coalgebra $B^\ndc(\mS^\ndc,-)$,  defined exactly like $\sim^\untyped$ in the previous section. Adequacy thus means that
\[ \forall  p,q\in (\mS^{\ndc,0})^\star.\, \forall n \in \omega.\,
	\llbracket p \rrbracket_n = \llbracket q
	\rrbracket_n \implies p\sim^\ndc_n q. \]
The locally final coalgebra $Z\cong B^\ndc(Z,Z)$ again consists of a form of guarded interaction trees.
At stage $0$, $Z_0 \coloneq \PSet_\omega(1) + 1\cong 3$. It encodes whether a term is stuck, can perform at least one $\beta$-reduction, or is
a combinator. At stage $n + 1$, an interaction tree is given by a
finite set of children. Each of these is either an interaction tree at stage
$n$, or a family $(f_i : Z_i \to Z_i)_{0 \leq i \leq n}$ whose members commute
with the restriction maps. See~\citet{Worrell05} for more details on final
coalgebras involving the power set functor. The denotational semantics is then given by the morphism $\llbracket -
	\rrbracket : \mS^{\ndc,0} \to Z$ that computes the guarded interaction tree of a
program inductively.

\subsection{The $\lambda$-calculus}
\label{subsec:lambda}

As our most intricate example, we consider guarded denotational models for the pure, untyped $\lambda$-calculus, whose operational semantics
is given by the two rules below:
\begin{equation}\label{eq:lambda}
	\inference{t \to t'}{t \app{}{} s \to t' \app{}{} s}
	\qquad
	\inference{}{(\lambda x.\,t) \app{}{} s \to t[s/x]}
\end{equation}
In comparison to \xcl, the $\lambda$-calculus involves \emph{variable binding} and \emph{substitution}. To apply our theory to this setting, we combine the topos-of-trees based modelling of \xcl (\Cref{subsec:xcl}) with the presheaf-based approach to
the operational semantics of languages with variable
binding~\cite{Goncharov23}, which builds on the seminal work on abstract syntax by Fiore, Plotkin and
Turi\,\cite{DBLP:conf/lics/FiorePT99}. We briefly recall the core constructions for the latter and  refer the reader to the two cited papers for details.

\paragraph{Presheaves for Variable Binding} We consider the
category  $\vcat$ of (covariant) presheaves on the
category $\fset$ of finite cardinals. Objects of $\fset$ are
sets $m=\{0,\dots,m-1\}\;(m \in \Nat)$, and morphisms $m\to k$ are functions. The category
$\vcat$ is a topos, hence admits finite limits and is cartesian
closed.

Intuitively, a presheaf $X \in \vcat$ assigns to each $m\in\fset$ (viewed as
context of $m$ untyped variables $x_0,\ldots,x_{m-1}$) a set $X(m)$ of ``terms'' with free variables from $m$, with functorial action $X(f \c m \to k)$ given by
renaming free variables according to $f$. For example, the formation of $\lambda$-terms yields a presheaf $\Lambda^0\in\vcat$ where $\Lambda^0(m)$ is the set of $\lambda$-terms modulo $\alpha$-equivalence with free variables from $m$. A more basic example is the presheaf
$V^{0}$ of \emph{variables}, given by
\begin{equation*}
	V^{0}(m) = m \qquad\text{and}\qquad V^{0}(f) = f.
\end{equation*}
That is, an element of $V^{0}(m)$ is a choice of a variable $i \in
	m$. Substitution is modelled abstractly by the bifunctor  $\langle -,- \rangle^{0}
	\c (\vcat)^{\opp} \times \vcat \to \vcat$ defined by the end construction
\begin{equation*}
	\langle X,Y \rangle^{0}(m) = \int_{l \in \fset} \Set(X(l)^{m}, Y(l)) = \vcat(X^{m} , Y),
\end{equation*}
where $X(l)^{m}$ is the $m$-fold product of the set $X(l)$.
Intuitively, an element of $\langle X,Y \rangle^{0}(m)$ is a family of maps
$X(l)^m\to Y(l)$, natural in $l$, which describes the simultaneous substitution
of $X$-terms in~$l$ variables for the $m$ variables of some fixed term,
resulting in a $Y$-term in $l$ variables. The functor $\langle -,- \rangle^{0}$
is the internal hom of a closed monoidal structure given by the \emph{substitution tensor}~\cite{DBLP:conf/lics/FiorePT99}; we omit its definition as it is not needed for our present purposes.

Finally, we have the \emph{context extension} endofunctor
\begin{equation*}
	\delta^{0} \c \vcat
	\to \vcat,\qquad \delta^{0} X (m) = X(m + 1), \qquad \delta^{0} X (h) = X(h + \id_{1}).
\end{equation*}
We think of $t\in \delta^0 X(m)$ as the term `$\lambda m.\, t$' with the last variable $m$ of $t$ being bound.

As elaborated in Fiore et al.\,\cite{DBLP:conf/lics/FiorePT99}, the presheaf $\Lambda^0$ of $\lambda$-terms forms the initial algebra for the endofunctor below whose summands correspond to the three constructors of the $\lambda$-calculus:
\begin{equation}\label{eq:lamsyn}
	\Sigma^{\lambda,0}: \vcat \to \vcat,\qquad \Sigma^{\lambda,0} X = \underbrace{V^{0}}_{\text{variables}}
	+ \underbrace{\delta^{0} X}_{\lambda x.(-)}
	+\, \underbrace{X \product X}_{\mathsf{app}}.
\end{equation}
The initial algebra $\Lambda=\mu\Sigma^{\lambda, 0}$ corresponds to the \emph{nameless}
representation of $\lambda$-terms in the style of De Bruijn indices~\cite{debruijn1972nameless}.

\paragraph{Categorical Modelling of the $\lambda$-Calculus} In previous work~\cite{Goncharov23}, the small-step operational semantics \eqref{eq:lambda} of the untyped $\lambda$-calculus has been modelled via a GSOS law over the presheaf category $\mathbf{Set}^\fset$. To apply our present theory, we devise an alternative modelling of the $\lambda$-calculus in guarded form. To this end, we combine the category $\fset$ with step-indexing and consider the
presheaf category \[\vartopos \cong \mathbf{Set}^{\fset \times \omega^{\opp}}.\]
This category can be thought of as a topos of trees for \emph{variable
	sets}, or simply the \emph{topos of {variable} trees}. Objects in $\vartopos$ are families
$(X_{n})_{n \in \omega}$ of presheaves equipped with restriction
\emph{natural transformations} $X_{n,m}\colon X_{n} \to X_{m}$ for all $n\geq m$. The
category $\vartopos$ is cartesian closed, and the exponential $Y^X$ is given at $n\in \omega$ and $m\in \fset$ by the set $Y^X(n)(m)$ of $n+1$-tuples $(f_i\in Y_i^{X_i}(m))_{0\leq i\leq n}$ where $Y_i^{X_i}$ denotes the exponential in $\vcat$ and the $f_i$ internally commute with restriction.

\begin{notation}
	Given $X \in \vartopos$, we write $X_{n}(m)$ for $X(n)(m)$ where $n\in \omega$ and $m\in \fset$.
\end{notation}

Next, we extend the constructions for abstract syntax in $\vcat$ to $\vartopos$. The presheaf $V^{0} \in \vcat$ of variables and the endofunctor $\delta^{0} \c \vcat \to \vcat$ are extended
pointwise to
\begin{align*}
	 & V \in \vartopos,                  &  & V_{n}(m) = V^{0}(m),
	 & \delta : \vartopos \to \vartopos, &  & \delta_{n}X =
	\delta^{0}(X_{n}).
\end{align*}
The substitution bifunctor $\langle -,- \rangle^{0} \c (\vcat)^{\opp} \times \vcat
	\to \vcat$ extends to
\begin{equation*}
	\langle -,- \rangle \c (\vartopos)^{\opp} \times \vartopos
	\to \vartopos,\qquad \langle X,Y \rangle_{n} = \int_{i \leq n} \langle X_{i} , Y_{i}\rangle^{0}.
\end{equation*}
More explicitly, $\langle X,Y
	\rangle_{n}(m)$ is the set of tuples of substitution structures $(g_i \in \langle X_{i},Y_{i}\rangle^0(m))_{0\leq i\leq n}$ that are compatible with restriction in the obvious way. The syntax functor $\Sigma^{\lambda,0}$ extends to
\begin{equation}
	\Sigma^\lambda\colon \vartopos \to \vartopos,\qquad
	\Sigma^\lambda_{n}X = \Sigma^{\lambda,0} X_{n} = V^0 + \delta^0 X_n +  X_n \times X_n.
\end{equation}
At any stage $n \in \omega$, the initial $\Sigma^\lambda$-algebra
$\mu\Sigma^\lambda_{n} \in \vartopos$ is the presheaf $\Lambda$ of $\lambda$-terms.

To model the guarded behaviour of the $\lambda$-calculus, we need one final
ingredient: a version of the later modality $\later$ for the topos of variable
trees. Its definition is virtually identical to the later modality used for \xcl, except that the objects $X$ below are now in $\vartopos$, not $\topos$:
\begin{equation}
	\later_{1} X = 1, \qquad \text{and} \qquad \later_{n+1} X = X_{n},
\end{equation}
where $1$ is the constant presheaf on the singleton. We can now define the behaviour bifunctor
\begin{equation}
	\label{eq:behlambda}
	B^\lambda \c (\vartopos)^{\opp} \times
	\vartopos \to \vartopos, \qquad B^\lambda(X,Y) = \langle \later X, \later Y \rangle\times (\later Y + \later Y^{\later X} + 1).
\end{equation}

The bifunctor $B^\lambda$ has two components: a (guarded) substitution component
$\langle \later X, \later Y \rangle$, which intends to model the simultaneous
substitution structure of $\lambda$-terms, and a computational component, expressing the fact that
a term may either
$\beta$-reduce ($\later Y$), terminate to a
$\lambda$-abstraction and thus behaves as a (guarded) function on terms $(\later
	Y^{\later X})$, or be stuck $(1)$. Apart from the presence of the $\later$-modality
that makes $B^\lambda$ locally contractive in both arguments, $B^\lambda$ is analogous to the
behaviour bifunctor of the $\lambda$-calculus previously considered in higher-order abstract GSOS~\cite[Sec.~5.4]{Goncharov23}.

The definition of the higher-order GSOS law for the $\lambda$-calculus involves a minor technicality:
\begin{remark}
	Formally, the operational semantics of the $\lambda$-calculus does not form a
	higher-order GSOS law in $\vartopos$, but instead a \emph{$V$-pointed} higher-order GSOS law,
	meaning that the presheaves~$X$ appearing in $\rho_{X,Y}$ are assumed to be
	equipped with a natural transformation $V\to X$ that includes variables into $X$-terms.
	This artifact,
	originating back to \citet{DBLP:conf/lics/FiorePT99}, does not
	affect our theory -- all definitions, theorems, and proofs of previous sections carry over with trivial notational changes to the $V$-pointed case. Therefore, points are omitted from our following presentation, and we tacitly assume that every variable is an $X$-term.
\end{remark}
With this preparation, the operational semantics of the $\lambda$-calculus is captured by a ($V$-pointed) higher-order GSOS law of $\Sigma^\lambda$ over $B^\lambda$, that is, a (di)natural transformation of type
\begin{equation}
	\label{eq:rholambda}
	\begin{aligned}
		\rho_{X,Y}^\lambda = \langle \nu_{X,Y}, \kappa_{X,Y} \rangle \c
		 & \Sigma^\lambda(X \times \langle \later X, \later Y \rangle\times (\later Y + \later Y^{\later X} + 1)) \to
		\\
		 & \langle \later X ,
		\later (\Sigma^\lambda)^{\star}(X + Y) \rangle \times (\later (\Sigma^\lambda)^{\star}(X + Y)
		+ (\later (\Sigma^{\lambda})^{\star}(X + Y))^{\later X} + 1).
	\end{aligned}
\end{equation}
Observe that, by the definition of $B^\lambda$
\eqref{eq:behlambda}, the law $\rho^\lambda$ consists of two components, the
\emph{simultaneous substitution} component $\nu$ and the \emph{computational}
component $\kappa$. The map $\nu_{X,Y,n \in \omega,m \in \fset}$ acts as follows:
\begin{enumerate}
	\item $\mathsf{app}((p,g,\_),(q,h,\_))$ (where $p,q\in X_n(m)$ and
	      $g,h\in \langle \later X , \later Y\rangle_{n}(m)$) is mapped to an element of $\langle \later X, \later (\Sigma^\lambda)^{\star}(X+Y)\rangle_n(m)$, a tuple of $\vcat$-morphisms $(\later_i X)^m \to (\Sigma^{\lambda,0})^\star(\later_i X+\later_i X)$ for $i\leq n$:
	      \[
		      \lambda i \leq n.\,\lambda \vec{u}\colon {\later_{i}}X(l)^{m}.\,\mathsf{app}(g_{i}(\vec{u}),h_{i}(\vec{u})).
	      \]
	\item $\mathsf{lam}(p,g,\_)$ (where $p\in X_{n}(m+1)$ and $g \in \langle \later X , \later Y\rangle_{n}(m+1)$)
	      is mapped to the tuple
	      \[
		      \lambda i \leq n.\,\lambda \vec{u}\colon
		      {\later_{i}} X(l)^{m}.\,\mathsf{lam}(g_{i}(\mathrm{wk}(\vec{u}),m)),
	      \]
	      where $\mathsf{lam}$ denotes the $\lambda$-abstraction constructor, $\mathrm{wk}$ weakens the substitution $\vec{u} : \later_{i} X(l)^{m}$ to one of type
	      $\later_{i} X(l+1)^{m}$, and $m$ is the
	      most ``fresh'' variable.
	\item a variable $x_{j}$, $j < m$, is mapped to the tuple
	      \[
		      \lambda i \leq n.\, \lambda \vec{u}\colon {\later_{i}} X(l)^{m}.\,\vec{u}(j).
	      \]
\end{enumerate}
Next, we move on to the component $\kappa$. For simplicity, we identify elements $f\in B^A(m)$ of an exponential in $\vcat$ with the corresponding function $A(m)\to B(m)$, $a\mapsto \ev_m(f,a)$, where $\ev\colon B^A\times A\to B$ is the evaluation map. For $\kappa_{X,Y,n \in
		\omega,m \in \fset}$, we have (omitting the trivial variable case):
\begin{enumerate}
	\item $\mathsf{app}((p,\_,f),(q,\_,\_))$ (where $p,q\in X_n(m)$ and $f\in
		      \later_n Y(m)+(\later Y^{\later X})_n(m))$ is mapped to the term $t_n\in
		      (\Sigma^\lambda)^\star_{n} (\later X +\later Y)(m)$ defined inductively by
	      \[
		      t_0=\begin{cases}
			      \inr(*)                        & \text{if $f\in \inr(*)$} \\
			      \mathsf{app}(\inr(*), \inl(*)) & \text{if $f\in \inl(*)$}
		      \end{cases}\qquad t_{n+1}=\begin{cases}
			      f_{n+1}(X_{n,m}(q))        & \text{if $f\in  (\later Y ^{\later X})_{n+1}(m)$} \\
			      \mathsf{app}(f,X_{n,m}(q)) & \text{if $f\in Y_n(m)=\later_{n+1}Y(m)$}.
		      \end{cases}
	      \]
	\item $\mathsf{lam}(p,g,\_)$ (where $p \in X_{n}(m+1)$ and $g \in \langle \later X, \later Y \rangle_{n}(m+1)$)
	      is mapped to an element of $(\later (\Sigma^{\lambda})^{\star}(X + Y))^{\later X}(m)$, i.e.\ a tuple of elements of $(\later_i (\Sigma^{\lambda})^{\star}(X + Y))^{\later_i X}(m)$ ($i\leq n$), given by
	      \[
		      \lambda i \leq n.\,\lambda e \in \later_{i} X(m).\,g_{i}({0},\dots,{m-1},e).
	      \]
\end{enumerate}
The law \eqref{eq:rholambda} is trivially relatively flat for $J = 1$. Its operational $\rho$-model \[\langle \gamma^{1}, \gamma^{2}\rangle \c {\mS^\lambda} \to
	\langle \later {\mS^\lambda}, \later {\mS^\lambda} \rangle \times (\later {\mS^\lambda} +  {(\later \mS^\lambda)}^{\later {(\mS^\lambda)}} + 1)\]
models the (guarded) operational semantics of  the $\lambda$-calculus:
for all $n >0$, $m \in \fset$ and $t \in \Lambda(m)$, we have (i) $(\gamma^{1}_{n,m}(t))_{n}(\vec{u}) =
	t[\vec{u}]$ for any suitable substitution $\vec{u}$, and (ii) $\gamma^{2}_{n,m}(t) = t'$ iff $t \to t'$.

\paragraph{Denotational Semantics of the $\lambda$-Calculus} The category $\vartopos$ is an M-category analogous to $\topos$, and the bifunctor $B^\lambda$ satisfies the conditions of \Cref{thm:lfc-existence} analogous to the bifunctor~$B^\untyped$ for \xcl. Thus, our general theory yields an adequate denotational semantics for the $\lambda$-calculus:

\begin{theorem}[Denotational $\rho$-model for the $\lambda$-calculus]\label{thm:lambda-lfc}
	The behaviour bifunctor $B^\lambda$ has a locally final coalgebra. It extends to a denotational $\rho$-model, adequate for step-indexed strong applicative bisimilarity.
\end{theorem}

The notion of bisimilarity is similar to \xcl, except that it is now required to respect substitution: $p\sim^\lambda q$ implies $p[\vec{u}] \sim^\lambda q[\vec{u}]$ for every applicable substitution $\vec{u}$. See \cite[Sec.~5]{Goncharov23} for more details.

The locally final coalgebra
$Z \cong B(Z,Z)$ once again emerges from the inverse limit construction of
\Cref{rem:cons-lfc}. Explicitly, the object $Z$ consists of \emph{guarded
	interaction trees} with renaming, where each tree $ z \in Z_{n}(m)$ has an
additional observable in the form of a substitution action. At stages $n + 1 \in \omega$ and $m \in \fset$, an element $z
	\in Z_{n+1}(m)$ consists  of (i) a family of substitutions
$(\langle Z_{i},Z_{i}\rangle^{0})_{i \leq n}$ compatible with restriction and
(ii) either another tree $z' \in Z_{n}(m)$ or a tuple of elements of exponentials $(f_{0} \c
	Z_{0}^{Z_{0}}(m),\dots,f_{n} \c Z_{n}^{Z_{n}}(m))$ compatible with restriction.

\section{Conclusions and Future Work}
We have presented a general bialgebraic approach to denotational semantics in
the higher-order abstract GSOS framework, extending the principles of Turi and
Plotkin's work on first-order abstract GSOS~\cite{Turi97} to a
higher-order setting. Our denotational models are based on the novel notion of
\emph{locally final} coalgebra for a behaviour bifunctor. We have shown that all
such models are adequate w.r.t.\ the operational semantics and given general criteria for
their existence and uniqueness.

One important message of our work is that the denotational semantics for
higher-order languages need not be conceptually isolated from the operational
semantics; rather, they may \emph{follow} from the operational semantics. This
is a fresh perspective, and we are thus compelled to further apply our theory to
more complex examples, such as $\pi$-calculi or, following recent advances in substructural abstract syntax and substitution~\cite{fr25, lenke2026unifiedtreatmentsubstitutionpresheaves},
substructural (e.g.~linear) varieties of the $\lambda$-calculus.

While this paper focuses exclusively on computational adequacy of the
constructed denotational $\rho$-models, we have not explored under which
conditions these are also sound, which is the missing step to achieve full
abstraction. In the first-order setting, if the behaviour functor preserves weak
pulbacks, full abstraction is guaranteed~\cite{Turi97}. The corresponding result
for higher-order languages is left as future work.

Higher-order abstract GSOS is tailored to small-step operational semantics, and our adequacy result applies to notions of strong applicative bisimilarity. A future direction is to study denotational models of \emph{abstract higher-order specifications}, an extension of higher-order abstract GSOS capturing big-step semantics~\cite{Urbat26}. In this way, we expect that more coarse-grained forms of semantics that regard $\beta$-reductions as unobservable and focus on functional behaviour -- with \emph{weak} applicative bisimilarity~\cite{Abramsky90} as behavioural equivalence -- emerge in a general bialgebraic setting as well.

Our present construction method for locally final coalgebras builds on
M-categories and yields denotational models in the spirit of synthetic guarded
domain theory~\cite{Birkedal10,Birkedal12,Frumin24}. An alternative approach to
fixed point constructions of functors is given by \emph{CPO-enriched}
categories, introduced by \citet{SmythPlotkin1982} (see
also~\cite[Sec.~5]{Adamek_Milius_Moss_2025}) as a categorical foundation for
constructing the more traditional, order-theoretic denotational domains for
higher-order languages~\cite{scott72}. We aim to develop a version of our
\Cref{thm:lfc-existence} for CPO-enriched categories, which should provide the
means to study such models in the context of our abstract bialgebraic setting.
Similarly, {categories of sheaves over a fixed, complete Heyting algebra
		with a well-founded basis} are an alternative to M-categories for
fixed point constructions~\cite{Birkedal12}. Adapting \Cref{thm:lfc-existence} to
such categories could bring about denotational domains obtained via
\emph{transfinite} guarded recursion.

Our treatment of $\xnccl$
(\Cref{thm:xntcl-lfc}) is conceptually related to the work of
\citet{MogelbergVezzosi21}, who construct two guarded recursive powerdomains in
Clocked Cubical Type Theory yielding adequate denotational models for the untyped
$\lambda$-calculus with finite non-determinism. Whereas their models target applicative simulation in may- and must-convergence
variants, our locally final coalgebra for $B^\ndc$ targets strong applicative
bisimilarity. A further
investigation of the connections between these denotational models is left for future work.




\begin{acks}
	Sergey Goncharov, Stelios Tsampas, and Stefano Volpe acknowledge funding by
	the Deutsche Forschungsgemeinschaft (DFG, German Research Foundation) -
	project number 527481841. Henning Urbat is supported by Deutsche
	Forschungsgemeinschaft (DFG, German Research Foundation) – project number
	569130867.
\end{acks}

%
\appendix

\section{Omitted Proofs and Details}

\subsection*{Proof of \Cref{th:denotational-model}}
We construct the algebra structure $a_\rho$ as follows.
First, we define a family of morphisms
\begin{displaymath}
	(a^j_{\rho} : Z + \Sigma_j Z \to Z)_{j \in J}
\end{displaymath}
by well-founded recursion on $j$. Thus let $j \in J$, and assume that
$a^i_{\rho}$ is already defined for $i < j$. Put
\[
	a^{< j}_{\rho}\colon \Sigma_{< j} Z \to Z,\qquad
	a^{< j}_{\rho} \coloneq [a^i_{\rho} \cdot \inr ]_{i < j}.
\]
The object $Z+\Sigma_j Z$ can be equipped with the $B(Z,-)$-coalgebra structure
shown in \cref{fig:sigma-algebra}. We then define $a^j_{\rho}$ to be the unique $B(Z,-)$-co\-al\-ge\-bra morphism from that coalgebra to the
final $B(Z,-)$-co\-al\-ge\-bra $(Z,z)$; in other words, $a^j_\rho$ is the unique morphism making the diagram in \cref{fig:sigma-algebra} commute.
\begin{figure*}
	\centering
	\begin{tikzcd}[scale cd = .8, column sep=4em]
		{Z + \Sigma_j Z} & Z & {B(Z,Z)} \\
		{B(Z,Z)+\Sigma_j (Z \times B(Z,Z))} && {B(Z, Z + \Sigma_j Z)} \\
		{B(Z,Z)+B(Z, \Sigma^{\star}_{<j}Z + \Sigma_j \Sigma^{\star}_{<j}Z)} & {B(Z,Z)+B(Z, Z + \Sigma_j Z)} & {B(Z,Z + \Sigma_j Z)+B(Z, Z + \Sigma_j Z)}
		\arrow["{a^j_{\rho}}", dashed, from=1-1, to=1-2]
		\arrow["{z+\Sigma_j\langle \id,\, {z} \rangle}"', from=1-1, to=2-1]
		\arrow["{z \quad \cong}", from=1-2, to=1-3]
		\arrow["{\id + \rho^j_{Z}}"', from=2-1, to=3-1]
		\arrow["{B(\id,\, a^j_{\rho})}"', from=2-3, to=1-3]
		\arrow["{\id + B(\id, \, \widehat{a^{<j}_{\rho}} + \Sigma_j \widehat{a^{<j}_{\rho}})}", bend left=10, from=3-1, to=3-2]
		\arrow["{B(\id,\,\inl) + \id}", from=3-2, to=3-3]
		\arrow["{\nabla}"', from=3-3, to=2-3]
	\end{tikzcd}
	\Description{The coalgebra on $Z + \Sigma_j Z$ is:
		$Z + \Sigma_j Z \xrightarrow{[z,\Sigma_j \hat{z}]}
			B(Z,Z) + \Sigma_j B^{\infty}(Z, Z) \xrightarrow{B(Z,Z) + \rho^j_Z}
			B(Z,Z) + B(Z,\Sigma^{\star}_{<j}Z + \Sigma_j \Sigma^{\star}_{<j}Z)
			\xrightarrow{(B(Z,Z) + B(Z,(-) + \Sigma_j (-)))\widehat{a^{<j}_{\rho}}}$
		B(Z,Z) + B(Z,Z + \Sigma_j Z) \xrightarrow{B(Z,\inl)+B(Z,Z+\Sigma_j Z)}
		B(Z,Z+\Sigma_j Z) + B(Z,Z+\Sigma_j Z) \xrightarrow{\nabla_{B(Z,Z+\Sigma_j
				Z)}} B(Z,Z+\Sigma_j Z)}
	\caption{A $B(Z,-)$-coalgebra on $Z + \Sigma_j Z$. There is
		a unique coalgebra morphism $a^j_{\rho}$ from it to the final coalgebra.}
	\label{fig:sigma-algebra}
\end{figure*}

Using the morphisms $a_\rho^j$ defined above we now obtain the following
$\Sigma$-algebra structure $a_\rho$ on $Z$:
\[
	a_{\rho}\colon \Sigma Z \to Z,\qquad
	a_{\rho} \coloneq [a^j_{\rho} \cdot \inr]_{j \in J}.
\]
It remains to verify that $(Z,a_\rho,z)$ forms a $\rho$-bialgebra, that is, the diagram below commutes:
\[\begin{tikzcd}
		{\Sigma Z} & Z & {B(Z, Z)} \\
		{\Sigma (Z\times B(Z, Z))} && {B(Z, \Sigma^{\star}Z)}
		\arrow["{a_{\rho}}", from=1-1, to=1-2]
		\arrow["{\Sigma \langle\id,z\rangle}"', from=1-1, to=2-1]
		\arrow["z \quad \cong", from=1-2, to=1-3]
		\arrow["{\rho_Z}", from=2-1, to=2-3]
		\arrow["{B(\id,\, \widehat{a_{\rho}})}"', from=2-3, to=1-3]
	\end{tikzcd}\]
By definition of $a_{\rho}$ and since $\rho$ is relatively flat, this is
equivalent to showing that the outside of the diagram in
\cref{fig:denotational-model-proof} commutes for all $j\in J$.
\begin{figure*}
	\centering
	\begin{tikzcd}[column sep=-2.5em, scale cd=.8]
		{\Sigma_jZ} && {Z + \Sigma_j Z} && Z && {B(Z,Z)} \\
		&& {B(Z,Z)+\Sigma_j (Z\times B(Z,Z))} & \phantom{----}{B(Z,Z)+B(Z,Z + \Sigma_j Z)} &&& \\
		&& {B(Z,Z)+B(Z, \Sigma^{\star}_{<j}Z + \Sigma_j \Sigma^{\star}_{<j}Z)} & & \!\!\!\!\!\!{B(Z,Z + \Sigma_j Z) + B(Z,Z + \Sigma_j Z)} \\
		{\Sigma Z} & \phantom{-}{\Sigma_j (Z\times B(Z,Z))} &&&& \!\!\!\!\!\!\!\!\!\!\!{B(Z,Z + \Sigma_j Z)} \\
		&& {B(Z, \Sigma^{\star}_{<j}Z + \Sigma_j \Sigma^{\star}_{<j}Z)} && {B(Z,Z + \Sigma_j Z)} \\
		\\
		{\Sigma (Z\times B(Z,Z))} &&& {\coprod_{i \in J} B(Z,\Sigma^{\star}_{<i}Z + \Sigma_i \Sigma^{\star}_{<i}Z)} &&& {B(Z,\Sigma^{\star}Z)}
		\arrow["\inr", from=1-1, to=1-3]
		\arrow["{\inj_j}"', from=1-1, to=4-1]
		\arrow["{\Sigma_j \langle \id,z\rangle}", from=1-1, to=4-2]
		\arrow["{a^j_{\rho}}", from=1-3, to=1-5]
		\arrow["{z + \Sigma_j  \langle\id,z\rangle}"', from=1-3, to=2-3]
		\arrow["{z \quad \cong}", from=1-5, to=1-7]
		\arrow["{\id + \rho^j_Z}"', from=2-3, to=3-3]
		\arrow["{\id + B(\id, \widehat{a^{<j}_{\rho}} + \Sigma_j \widehat{a^{<j}_{\rho}})}" description, from=3-3, to=2-4]
		\arrow["{B(\id,\inl) + \id}", from=2-4, to=3-5]
		\arrow["{\nabla}", from=3-5, to=4-6]
		\arrow["{\Sigma \langle \id,z\rangle}"', from=4-1, to=7-1]
		\arrow["~\inr"' yshift=5, curve={height=-30pt}, from=4-2, to=2-3]
		\arrow["{\rho^j_Z}"', from=4-2, to=5-3]
		\arrow["{\inj_j}", from=4-2, to=7-1]
		\arrow["{B(\id,a^j_{\rho})}", from=4-6, to=1-7, xshift=5]
		\arrow["\inr"', from=5-3, to=3-3]
		\arrow["{B(\id, \widehat{a^{<j}_{\rho}} + \Sigma_j \widehat{a^{<j}_{\rho}})}", from=5-3, to=5-5]
		\arrow["{\inj_j}"', from=5-3, to=7-4]
		\arrow["{B(\id,e_{j,Z})}", from=5-3, to=7-7]
		\arrow["\inr", from=5-5, to=3-5]
		\arrow["{\id}"', from=5-5, to=4-6]
		\arrow["{[\inj_i \cdot \rho^i_Z]_{i \in J}}", from=7-1, to=7-4]
		\arrow["{[B(\id,e_{i,Z})]_{i \in J}}", from=7-4, to=7-7]
		\arrow["{B(\id,\widehat{a_{\rho}})}"', from=7-7, to=1-7, xshift=12]
	\end{tikzcd}
	\Description{}
	\caption{Main diagram for the proof of \cref{th:denotational-model}.}
	\label{fig:denotational-model-proof}
\end{figure*}
The only non-trivial cells are the upper right cell, which commutes by definition of
$a^j_{\rho}$ in \cref{fig:sigma-algebra}, and the rightmost
pentagon. By functoriality, for the latter it is enough to check commutativity of
\begin{equation}\label{eq:proof-diag}\begin{tikzcd}[column sep=8em]
		{\Sigma^{\star}_{<j}Z+\Sigma_j\Sigma^{\star}_{<j}Z} & {Z+\Sigma_jZ} \\
		{\Sigma^{\star}Z} & Z.
		\arrow["{\widehat{a^{<j}_{\rho}}+\Sigma_j \widehat{a^{<j}_{\rho}}}", from=1-1, to=1-2]
		\arrow["{e_{j,Z}}"', from=1-1, to=2-1]
		\arrow["{a^j_{\rho}}", from=1-2, to=2-2]
		\arrow["{\widehat{a_{\rho}}}", from=2-1, to=2-2]
	\end{tikzcd}\end{equation}
We consider the restriction of the diagram to the two coproduct components $\Sigma^{\star}_{<j}Z$ and $\Sigma_j\Sigma^{\star}_{<j}Z$ separately. The first component of \eqref{eq:proof-diag} is given by the outside of the following diagram:
\begin{equation}\label{eq:proof-diag-2}\begin{tikzcd}[row sep=huge]
		{\Sigma^{\star}_{<j}Z} & Z & {Z+\Sigma_j Z} \\
		{\Sigma^{\star}Z} && Z
		\arrow["{\widehat{a^{<j}_{\rho}}}", from=1-1, to=1-2]
		\arrow["{\inj^{\star}_{<j}}"', from=1-1, to=2-1]
		\arrow["\inl", from=1-2, to=1-3]
		\arrow["\id", from=1-2, to=2-3]
		\arrow["{a^j_{\rho}}", from=1-3, to=2-3]
		\arrow["{\widehat{a_{\rho}}}"', from=2-1, to=2-3]
	\end{tikzcd}\end{equation}
The right cell commutes by definition of $a_\rho^j$, as can be seen by
considering the first coproduct component of the diagram in
\cref{fig:sigma-algebra}. The left cell commutes due to the universal property
of the free algebra $\Sigmas_{<j}Z$ since all morphisms appearing in the cell
are $\Sigma_{<j}$-algebra morphisms whose restriction to $Z$ is given by the identity.


Similarly, the second component of \eqref{eq:proof-diag} is given by the
outside of the following diagram:

\[\begin{tikzcd}
		{\Sigma_j \Sigma^{\star}_{<j}Z} && {\Sigma_j Z} \\
		{\Sigma_j \Sigma^{\star}Z} & {\Sigma Z} & {Z+\Sigma_j Z} \\
		{\Sigma \Sigma^{\star}Z} & {\Sigma^{\star}Z} & Z
		\arrow["{\Sigma_j \widehat{a^{<j}_{\rho}}}", from=1-1, to=1-3]
		\arrow["{\Sigma_j \inj^{\star}_{<j}}"', from=1-1, to=2-1]
		\arrow["{\inj_j}", from=1-3, to=2-2]
		\arrow["\inr", from=1-3, to=2-3]
		\arrow["{\Sigma_j \widehat{a_{\rho}}}", from=2-1, to=1-3]
		\arrow["{\inj_j}"', from=2-1, to=3-1]
		\arrow["{a_{\rho}}", from=2-2, to=3-3]
		\arrow["{a^j_{\rho}}", from=2-3, to=3-3]
		\arrow["{\Sigma \widehat{a_{\rho}}}", from=3-1, to=2-2]
		\arrow["{\ini_Z}"', from=3-1, to=3-2]
		\arrow["{\widehat{a_{\rho}}}"', from=3-2, to=3-3]
	\end{tikzcd}\]
Commutativity of the upper triangle follows from commutativity of the left cell
of \eqref{eq:proof-diag-2}. The lower cell commutes by definition of
$\widehat{(-)}$, the central cell by naturality of $\inj_j$, and the right cell
by definition of $a_\rho$.
\qed

\subsection*{Proof of \Cref{thm:lfc-existence}}
Before proving the theorem, we need to establish a general property of stable M-categories. Stability ensures that limits of chains interact with distances of morphisms in an expected way:
\begin{lemma}\label{lem:stable-m-cat-limits}
	Let $\C$ be a stable $M$-category with limits of chains, and let $f,g\colon D\to D'$ be two natural transformations between $\alpha^\op$-chains, with induced morphisms $\lim f,\lim g\colon \lim D\to \lim D'$. Then
	\[ d(\lim f, \lim g) \leq \sup_{\beta<\alpha} d(f_\beta,g_\beta).  \]
\end{lemma}

\begin{proof}
	Let  $\pi_\beta\colon \lim D\to D_\beta$ and $\pi_\beta'\colon \lim D'\to D'_\beta$ ($\beta<\alpha$) denote the limit projections. Since the morphisms $\pi'_\beta$  are jointly monic, we have
	\begin{align*}
		d(\lim f,\lim g) & = \sup_{\beta<\alpha} d(\pi_\beta'\cdot \lim f,\, \pi_\beta'\cdot \lim g) & \text{$\C$ stable}                         \\
		                 & = \sup_{\beta<\alpha} d(f_\beta\cdot \pi_\beta,\, g_\beta\cdot \pi_\beta) & \text{def.~$\lim f$, $\lim g$}             \\
		                 & \leq  \sup_{\beta<\alpha} d(f_\beta,\, g_\beta)                           & \text{$\C$ an M-category.} \tag*{\qedhere}
	\end{align*}
\end{proof}
Now we turn to the actual proof of \Cref{thm:lfc-existence}. The strategy is to apply \Cref{th:m-category-fixpoint} to a contravariant functor $F\colon \C^\op\to \C$ whose fixed points correspond precisely to the locally final coalgebras for $B$. We proceed in several steps.
\begin{enumerate}[label=(\alph*)]
	\item\label{pf1} We first construct the desired functor $F$. Condition \ref{thm:ass3} implies that for each object $X\in \C$ the endofunctor $B(X,-)$ has a final coalgebra, which we denote by
	      \[ \tau_X\colon FX \xto{\cong} B(X,FX). \]
	      The object assignment $X\mapsto FX$ extends to a contravariant functor \[F\colon \C^\op\to \C\] sending $f\colon X\to Y$ to the unique $B(X,-)$ coalgebra morphism
	      from $(FY,B(f,\id)\cdot \tau_Y)$ to $(FX, \tau_X)$:
	      \begin{equation}\label{eq:Ff}
		      \begin{tikzcd}
			      FY \ar{r}{\tau_Y} \ar[dashed]{d}[swap]{Ff} & B(Y,FY) \ar{r}{B(f,\id)} & B(X,FY) \ar{d}{B(\id,Ff)} \\
			      FX \ar{rr}{\tau_X} && B(X,FX)
		      \end{tikzcd}
	      \end{equation}
	      Functoriality of $F$ is clear by the uniqueness of $Ff$.
	\item\label{pf2} We prove that the functor $F$ is locally contractive. For every object $X\in \C$, let $D^X$ denote the final sequence of the endofunctor $B(X,-)$. Every morphism $f\colon X\to Y$ induces a natural transformation $\bar f\colon D^Y\to D^X$ whose components are defined by transfinite recursion as follows:
	      \begin{enumerate}[label=(\roman*)]
		      \item $\bar{f}_0$ is given by $D^Y_0 = 1 \xto{\id} 1 = D^X_0$.
		      \item $\bar{f}_{\alpha+1}$ is given by $D^Y_{\alpha+1} = B(Y,D^Y_\alpha) \xrightarrow{B(f,\bar{f}_\alpha)} B(X,D^X_\alpha)=D^X_{\alpha+1}$.
		      \item If $\alpha$ is a limit ordinal, $\bar{f}_\alpha$ is the unique morphism such
		            that the diagrams below commute:
		            \[
			            \begin{tikzcd}
				            D^Y_\beta \ar{d}[swap]{\bar{f}_\beta} & D^Y_\alpha \ar{l}[swap]{D^Y_{\alpha,\beta}} \ar[dashed]{d}{\bar{f}_\alpha}\\
				            D^X_\beta & D^X_\alpha \ar{l}[swap]{D^X_{\alpha,\beta}}
			            \end{tikzcd} \qquad (\beta<\alpha)
		            \]
	      \end{enumerate}
	      Let $c\in [0,1)$ be the contraction factor witnessing that $B$ is locally contractive in the first argument. We claim that, for all $f,g\colon X\to Y$ and all ordinals $\alpha$,
	      \begin{equation}\label{eq:d-prop}
		      d(\bar{f}_\alpha, \bar{g}_\alpha) \leq c\cdot d(f,g).
	      \end{equation}
	      The proof is by transfinite induction on $\alpha$:
	      \begin{enumerate}[label=(\roman*)]
		      \item For $\alpha=0$ we have $\bar{f}_\alpha = \bar{g}_\alpha = \id_1$, so
		            $d(\bar{f}_\alpha,\bar{g}_\alpha)=0$ and \eqref{eq:d-prop} holds trivially.
		      \item For a successor ordinal $\alpha+1$, we have
		            \begin{align*}
			            d(\bar{f}_{\alpha+1},\bar{g}_{\alpha+1}) & = d(B(f,\bar{f}_\alpha), B(g,\bar{g}_\alpha))                                           & \text{def.\ $\bar{f}_{\alpha+1}$, $\bar{g}_{\alpha+1}$} \\
			                                                     & = d(B(\id,\bar{f}_\alpha)\cdot B(f,\id), B(\id,\bar{g}_\alpha)\cdot B(g,\id))           & \text{functoriality of $B$}                             \\
			                                                     & \leq \max \{ d(B(\id,\bar{f}_\alpha), B(\id,\bar{g}_\alpha)),  d(B(f,\id), B(g,\id)) \} & \text{$\C$ an M-category}                               \\
			                                                     & \leq \max \{ d(\bar{f}_\alpha, \bar{g}_\alpha), d(B(f,\id), B(g,\id)) \}                & \text{$B$ non-exp.\ in 2nd arg.}                        \\
			                                                     & \leq c\cdot d(f,g).                                                                     &
		            \end{align*}
		            In the final step we use that $d(\bar{f}_\alpha, \bar{g}_\alpha)\leq c\cdot d(f,g)$ by induction, and $d(B(f,\id), B(g,\id)\leq c\cdot d(f,g)$ because $B$ is locally contractive with factor $c$ in the first argument.
		      \item For a limit ordinal $\alpha$, we have
		            \begin{align*}
			            d(\bar{f}_{\alpha},\bar{g}_{\alpha}) & \leq \sup_{\beta<\alpha} d(\bar{f}_\beta,\bar{g}_\beta) & \text{\Cref{lem:stable-m-cat-limits}} \\
			                                                 & \leq c\cdot d(f,g)                                      & \text{by induction}.
		            \end{align*}
	      \end{enumerate}

	      By condition \ref{thm:ass3} there exists an ordinal $\alpha$ such that the
	      final sequences $D^X$ and $D^Y$ converge in $\alpha$ steps. We claim that
	      $Ff=\bar{f}_\alpha$. To see this, consider the diagram below:
	      \begin{equation}
		      \begin{tikzcd}[column sep=4em]
			      FY \ar[equals]{r} \ar{d}[swap]{\bar{f}_\alpha}\ar[shiftarr={yshift=20}]{rr}{\tau_Y} & D^Y_\alpha \ar{r}{(D^X_{\alpha+1,\alpha})^{-1}}  & B(Y,D^Y_\alpha) \ar{dr}{\bar{f}_{\alpha+1}} \ar{r}{B(f,\id)} & B(X,D^Y_\alpha) \ar{d}{B(\id,F\bar{f}_\alpha)} \\
			      FX \ar[equals]{r} \ar[shiftarr={yshift=-20}]{rrr}{\tau_Y}  & D^X_\alpha \ar{rr}{(D^X_{\alpha+1,\alpha})^{-1}} && B(X,D^X_\alpha)
		      \end{tikzcd}
	      \end{equation}
	      The upper cell commutes because the final sequence $D^Y$ converges in
	      $\alpha$ steps, which means that $(D^Y_{\alpha+1,\alpha})^{-1}\colon
		      D_\alpha^Y\to B(Y,D_\alpha^Y)$ is the final coalgebra $(FY,\tau_Y)$ for
	      $B(Y,-)$. Analogously for the lower cell. The central cell commutes
	      because $\bar{f}$ is a natural transformation, and the right-hand cell by
	      definition of $\bar{f}_{\alpha+1}$. Thus the outside diagram commutes, and
	      since by definition $Ff$ is the unique morphism making the outside commute, we conclude that $Ff=\bar{f}_\alpha$. Analogously, $Fg=\bar{g}_\alpha$. Thus
	      \begin{align*}
		      d(Ff,Fg) = & \; d(\bar{f}_\alpha,\bar{g}_\alpha) & \text{as shown above}        \\
		      \leq       & \; c\cdot d(f,g)                    & \text{by \eqref{eq:d-prop}},
	      \end{align*}
	      which proves that $F$ is locally contractive.
	\item\label{pf3} The hom-set $\C(1,F1)$ is inhabited: there exists a coalgebra $1\to B(1,1)$ by condition \ref{thm:ass2}, and the unique $B(1,-)$-coalgebra morphism to the final $B(1,-)$-coalgebra is a morphism $1\to F1$ in $\C$.
	\item\label{pf4} By parts \ref{pf2} and \ref{pf3} the conditions of the fixed point theorem (\Cref{th:m-category-fixpoint}) are satisfied. Thus $F$ has a fixed point
	      \[ i\colon FZ\cong Z. \]
	      It follows that the higher-order coalgebra
	      \begin{equation}\label{eq:ho-coalg} \begin{tikzcd}  Z \ar{r}{i^{-1}}[swap]{\cong} & FZ \ar{r}{\tau_Z}[swap]{\cong} & B(Z,FZ)  \ar{r}{B(\id,i)}[swap]{\cong} &  B(Z,Z) \end{tikzcd} \end{equation}
	      is a final $B(Z,-)$-coalgebra (since the coalgebra $(FZ,\tau_Z)$ is by definition, and the coalgebra \eqref{eq:ho-coalg} is isomorphic to it by construction). We thus conclude that \eqref{eq:ho-coalg} is a locally final coalgebra.
	\item\label{pf5} It remains to prove the uniqueness of locally final
	      coalgebras, assuming that hom-sets in $\C$ are inhabited. This is immediate from the uniqueness statement of \Cref{th:m-category-fixpoint} once we show that every locally final coalgebra
	      \[z\colon Z\xto{\cong} B(Z,Z) \]
	      is a fixed point of $F$. But this is clear: since both $(Z,z)$ and $(FZ,\tau_Z)$
	      are final $B(Z,-)$-coalgebras, and final coalgebras are unique up to
	      isomorphism, we have $FZ\cong Z$. \qedhere
\end{enumerate}



\subsection*{Proof of \Cref{lem:thm-cond-set-ty}}
\begin{enumerate}
	\item We first prove that $B$ is locally non-expansive in the second argument: for
	      all $f,g\colon Y\to Y'$,
	      \[ d(B(\id,f),B(\id,g))\leq d(f,g),\]
	      or equivalently, for all $n\in \Nat$,
	      \begin{equation}\label{eq:pf-goal-set-ty-1} d(f,g)\leq 2^{-n} \implies d(B(\id,f),B(\id,g)) \leq 2^{-n}. \end{equation}
	      Thus let $d(f,g)\leq 2^{-n}$, that is,
	      \begin{equation}\label{eq:pf-taug} f_\tau=g_\tau \qquad \text{for all $\tau\in \Ty$ with $|\tau|\leq n$}, \end{equation}
	      and let $\tau$ be a type of complexity $|\tau|\leq n$. If $\tau=\utype$, we have
	      \[ B_\utype(\id,f)=f_\tau+\id \qquad\text{and}\qquad B_{\utype}(\id, g)=g_\tau+\id \]
	      by definition of $B_\utype$. If $\tau=\arty{\tau_1}{\tau_2}$, we have
	      \[B_{\arty{\tau_1}{\tau_2}}(\id,f) = f_{\arty{\tau_1}{\tau_2}}  + f_{\tau_2}\cdot (-) \qquad\text{and}\qquad B_{\arty{\tau_1}{\tau_2}}(\id,g) = g_{\arty{\tau_1}{\tau_2}}  + g_{\tau_2}\cdot (-) \]
	      Thus $B_\tau(\id,f)= B_\tau(\id,g)$ by \eqref{eq:pf-taug} because $|\tau_2|\leq
		      n$. This proves $d(B(\id,f),B(\id,g)) \leq 2^{-n}$, so
	      \eqref{eq:pf-goal-set-ty-1} holds.

	      Similarly, we prove that $B$ is locally contractive in the first argument, with
	      contraction factor $\frac{1}{2}$: for all $f,g\colon X\to X'$,
	      \[ d(B(f,\id),B(g,\id))\leq \frac{1}{2} \cdot d(f,g),\]
	      or equivalently, for all $n\in \Nat$,
	      \begin{equation}\label{eq:pf-goal-set-ty} d(f,g)\leq 2^{-n} \implies d(B(f,\id),B(g,\id)) \leq 2^{-(n+1)}. \end{equation}
	      Thus let $d(f,g)\leq 2^{-n}$, that is,
	      \begin{equation}\label{eq:pf-ftau} f_\tau=g_\tau \qquad \text{for all $\tau\in \Ty$ with $|\tau|\leq n$}. \end{equation}
	      and let $\tau$ be a type of complexity $|\tau|\leq n+1$. If $\tau=\utype$, we have
	      \[ B_\utype(f,\id)=B_{\utype}(g,\id)=\id \]
	      by definition of $B_\utype$. If $\tau=\arty{\tau_1}{\tau_2}$, we have
	      $f_{\tau_1}=g_{\tau_1}$ by \eqref{eq:pf-ftau} because $|\tau_1|\leq n$, and
	      therefore
	      \[B_{\arty{\tau_1}{\tau_2}}(f,\id) = \id +  (-)\cdot f_{\tau_1} = \id + (-)\cdot g_{\tau_1}= B_{\arty{\tau_1}{\tau_2}}(g,\id)\]
	      by definition of $B_{\arty{\tau_1}{\tau_2}}$. We have thus shown that
	      $d(B(f,\id),B(g,\id)) \leq 2^{-(n+1)}$, proving \eqref{eq:pf-goal-set-ty}.
	\item This statement is clear because $B_\tau(1,1)\neq \emptyset$ for all $\tau\in
		      \Ty$.
	\item The category $\Set^\Ty$ is locally ($\omega$-)presentable, being a presheaf
	      category. Moreover, the functor $B(X,-)$ clearly preserves monomorphisms
	      (sortwise injective maps), and it preserves colimits of $\alpha$-chains for
	      every regular cardinal~$\alpha$ that is greater than any $|X_\tau|$ for
	      $\tau\in \Ty$. The latter follows from the definition \eqref{eq:beh}
	      of $B$ together with the fact that colimits of $\alpha$-chains commute with
	      coproducts in any category, and with products of arity less than $\alpha$ in
	      any locally $\alpha$-presentable category~\cite[Prop.~1.59]{AdamekRosicky1994}.
	      Therefore, by \Cref{rem:final-sequence}, the final sequence of $B(X,-)$
	      converges. \qedhere
\end{enumerate}

\subsection*{Proof of \Cref{lem:fixed-point-iteration-xTCL}}
The proof is by induction on $n$. For $n=0$ the statement is vacuous because there is no type of complexity $0$. For the induction step, let $n\in \Nat$; we are to prove that
\begin{equation}\label{eq:pf-nub}
	(F^{n+1}1)_\tau = Z_\tau \qquad \text{for all $\tau\in \Ty$ with $n+1\geq |\tau|$}.
\end{equation}
We prove this by an inner induction on the complexity of $\tau$. Note that $FX=\nu B(X,-)$ is given by
\[ (FX)_\utype = \nu \bar{B}_1 \qquad\text{and}\qquad (FX)_{\arty{\tau_1}{\tau_2}} = \nu\bar{B}_{{(FX)_{\tau_2}}^{X_{\tau_1}}},   \]
where $\bar B_A Y = Y+A$ on $\Set$. Therefore we get
\[ (F^{n+1}1)_\utype = (FF^n1)_\utype = \nu \bar{B}_1 = Z_\utype, \]
and
\[ (F^{n+1}1)_{\arty{\tau_1}{\tau_2}} = (FF^n1)_{\arty{\tau_1}{\tau_2}} = \nu\bar{B}_{{(F^{n+1}1)_{\tau_2}}^{(F^n1)_{\tau_1}}} = \nu\bar{B}_{Z_{\tau_2}^{Z_{\tau_1}}} = Z_{\arty{\tau_1}{\tau_2}}.  \]
In the penultimate step, we use that $(F^n1)_{\tau_1}=Z_{\tau_1}$ by the outer
induction on $n$  (note that $n\geq |\tau_1|$ because $n+1\geq |\tau| =
	|\arty{\tau_1}{\tau_2}|$), and $(F^{n+1}1)_{\tau_2} = Z_{\tau_2}$ by the inner induction
on $\tau$. \qed

\subsection*{Proof of \Cref{thm:xptcl-lfc}}
We check that the functor $B^\prob$ satisfies the conditions of \Cref{th:m-category-fixpoint}:
\begin{enumerate}
	\item\label{thm:ass1-set-ty-prob} $B^\prob$ is locally contractive in the first and locally non-expansive in the second argument;
	\item\label{thm:ass2-set-ty-prob} The hom-set $\Set^\Ty(1,B^\prob(1,1))$ is inhabited.
	\item\label{thm:ass3-set-ty-prob} For each $X\in \Set^\Ty$ the final sequence of the endofunctor $B^\prob(X,-)$ converges.
\end{enumerate}
These three statements are immediate from the corresponding properties of $B$ (\Cref{lem:thm-cond-set-ty}) together with the observation that
\begin{enumerate}\item  the functor $\Df^\Ty$ is locally non-expansive (as is every functor of the form $F^\Ty$);
	\item $\Df X\neq \emptyset$ whenever $X\neq \emptyset$;
	\item  the functor $\Df$ preserves monomorphisms (injections) and colimits of $\omega$-chains.
\end{enumerate}
The uniqueness statement is shown like in proof of \Cref{thm:xcl-lfc} since $B^\prob$ again restricts to $\Set_{+}^\Ty$. \qed

\subsection*{Proof of \Cref{thm:xcl-lfc}}
We only need to show that the functor $B^\untyped$ satisfies the three conditions of \Cref{thm:lfc-existence}:
\begin{enumerate}
	\item\label{thm:ass1-set-omega} $B^\untyped$ is locally contractive. This can be shown like for the typed behaviour functor \eqref{eq:beh} for \stsc. Alternatively, one can observe that $B^\untyped$ is \emph{locally contractive in the internal sense of
		      $\topos$}~\cite[Definition 2.12]{Birkedal12}. Our metric is the same
	      as the one given by \cite[Proposition 5.1]{Birkedal12}.
	      Therefore, by \cite[Proposition 5.2]{Birkedal12}, $B^\untyped$ is
	      locally contractive.
	\item\label{thm:ass2-set-omega} The hom-set $\topos(1,B^\untyped(1,1))$ is inhabited. This is witnessed by the morphism \[1 \xrightarrow{\inl} 1 + (\later 1)^{\later 1} =
		      \later 1 + (\later 1)^{\later 1} = B^\untyped(1,1).\]
	\item\label{thm:ass3-set-omega} For each $X\in \topos$ the final sequence of the
	      endofunctor $B^\untyped(X,-)$ converges. This follows from the second sufficient criterion mentioned in
	      \cref{rem:final-sequence}. As $\topos$ is a presheaf category, it is locally
	      ($\omega$-)presentable. Additionally, $\later$ preserves both
	      monomorphisms and colimits. This, together with the observations made
	      in the proof of \cref{lem:thm-cond-set-ty}, proves that for fixed $X \in
		      \topos$, the endofunctor $B^\untyped(X,-)$ preserves monomorphisms and colimits of
	      $\alpha$-chains for every regular cardinal $\alpha$ that is greater than any
	      $|X_n|$ for $n < \omega$. (For the latter, note also that colimits in presheaf categories are formed at the level of underlying sets.) \qedhere
\end{enumerate}

\cref{thm:lfc-existence} hence applies (except for the uniqueness statement,
since not all hom-sets of $\topos$ are inhabited).

\subsection*{Proof of \Cref{thm:xntcl-lfc}}
We only need to show that the functor $B^\ndc$ satisfies the three conditions of \Cref{thm:lfc-existence}:
\begin{enumerate}
	\item\label{thm:ass1-set-nondet} $B^\ndc$ is locally contractive in the first and locally non-expansive in the second argument.
	\item\label{thm:ass2-set-nondet} The hom-set $\topos(1,B^\ndc(1,1))$ is inhabited.
	\item\label{thm:ass3-set-nondet} For each $X\in \topos$ the final sequence of the endofunctor $B^\ndc(X,-)$ converges.
\end{enumerate}
These three statements are immediate from the corresponding properties of $B^\ndc$ together with the observation that
\begin{enumerate}
	\item  the functor $\Powf$ is locally non-expansive;
	\item $\Powf X\neq \emptyset$ whenever $X\neq \emptyset$;
	\item the functor $\Powf$ preserves monomorphisms (injections) and colimits of $\omega$-chains.\qedhere
\end{enumerate}

\bibliographystyle{ACM-Reference-Format}
\bibliography{works}

@preamble{{\providecommand{\noopsort}[1]{}}}

@string{acm = "ACM"}

@string{elsevier = "Elsevier"}

@string{springer = "Springer"}

@article{gmstu26_jfp,
  title = {Higher-Order Bialgebraic Semantics},
  author = {Sergey Goncharov and Stefan Milius and Lutz Schröder and Stelios
            Tsampas and Henning Urbat},
  url = {https://jfp.episciences.org/17738},
  doi = {10.46298/jfp.17738},
  journal = {Journal of Functional Programming},
  issn = {1469-7653},
  volume = {Volume 36},
  eid = 1,
  year = {2026},
  month = {Jun},
}

@article{debruijn1972nameless,
  author = {N. G. de Bruijn},
  title = {Lambda Calculus Notation with Nameless Dummies: A Tool for Automatic
           Formula Manipulation, with Application to the Church–Rosser Theorem},
  journal = {Indagationes Mathematicae (Proceedings)},
  volume = {75},
  number = {5},
  pages = {381--392},
  year = {1972},
  publisher = {North-Holland},
  doi = {10.1016/1385-7258(72)90034-0},
}

@article{vr99,
  author = {Erik P. de Vink and Jan J. M. M. Rutten},
  title = {Bisimulation for Probabilistic Transition Systems: {A} Coalgebraic
           Approach},
  journal = {Theor. Comput. Sci.},
  volume = {221},
  number = {1-2},
  pages = {271--293},
  year = {1999},
  doi = {10.1016/S0304-3975(99)00035-3},
}

@inproceedings{fr25,
  author = {Marcelo Fiore and Sanjiv Ranchod},
  title = {Substructural Abstract Syntax with Variable Binding and
           Single-Variable Substitution},
  booktitle = {40th Annual {ACM/IEEE} Symposium on Logic in Computer Science, {
               LICS} 2025, Singapore, June 23-26, 2025},
  pages = {196--208},
  publisher = {{IEEE}},
  year = {2025},
  url = {https://doi.org/10.1109/LICS65433.2025.00022},
  doi = {10.1109/LICS65433.2025.00022},
  bibsource = {dblp computer science bibliography, https://dblp.org},
}

@article{clp98,
  author = {Dezani-Ciancaglini, Mariangiola and de'Liguoro, Ugo and Piperno,
            Adolfo},
  title = {A Filter Model for Concurrent Lambda-Calculus},
  journal = {SIAM Journal on Computing},
  volume = {27},
  number = {5},
  pages = {1376-1419},
  year = {1998},
  doi = {10.1137/S0097539794275860},
}

@inproceedings{scott72,
  author = "Scott, Dana",
  editor = "Lawvere, F. W.",
  title = "Continuous lattices",
  booktitle = "Toposes, Algebraic Geometry and Logic",
  year = "1972",
  publisher = "Springer Berlin Heidelberg",
  address = "Berlin, Heidelberg",
  pages = "97--136",
  isbn = "978-3-540-37609-5",
}

@incollection{SmythPlotkin1982,
  author = {M. B. Smyth and G. D. Plotkin},
  title = {The Category-Theoretic Solution of Recursive Domain Equations},
  booktitle = {SIAM Journal on Computing},
  year = {1982},
  volume = {11},
  number = {4},
  pages = {761--783},
  doi = {10.1137/0211062},
  publisher = {Society for Industrial and Applied Mathematics},
}

@inproceedings{MogelbergVezzosi21,
  author = {M{\o}gelberg, Rasmus Ejlers and Vezzosi, Andrea},
  editor = {Ana Sokolova},
  title = {Two Guarded Recursive Powerdomains for Applicative Simulation},
  booktitle = {Proceedings 37th Conference on Mathematical Foundations of
               Programming Semantics, {MFPS} 2021, Hybrid: Salzburg, Austria and
               Online, 30th August - 2nd September, 2021},
  series = {{EPTCS}},
  pages = {200--217},
  year = {2021},
  doi = {10.4204/EPTCS.351.13},
  url = {https://doi.org/10.4204/EPTCS.351.13},
}

@inproceedings{mp16,
  author = {M\o{}gelberg, Rasmus Ejlers and Paviotti, Marco},
  title = {Denotational semantics of recursive types in synthetic guarded domain
           theory},
  year = {2016},
  isbn = {9781450343916},
  publisher = {Association for Computing Machinery},
  address = {New York, NY, USA},
  doi = {10.1145/2933575.2934516},
  pages = {317–326},
  numpages = {10},
  location = {New York, NY, USA},
  series = {LICS '16},
}

@inproceedings{bbm14,
  author = "Bizjak, Ale{\v{s}} and Birkedal, Lars and Miculan, Marino",
  editor = "Dowek, Gilles",
  title = "A Model of Countable Nondeterminism in Guarded Type Theory",
  booktitle = "Rewriting and Typed Lambda Calculi",
  year = "2014",
  publisher = "Springer International Publishing",
  address = "Cham",
  pages = "108--123",
  isbn = "978-3-319-08918-8",
  doi = "10.1007/978-3-319-08918-8_8",
}

@inproceedings{nakano2000,
  author = {Nakano, H.},
  booktitle = {Proceedings Fifteenth Annual IEEE Symposium on Logic in Computer
               Science (Cat. No.99CB36332)},
  title = {A modality for recursion},
  year = {2000},
  volume = {},
  number = {},
  pages = {255-266},
  doi = {10.1109/LICS.2000.855774},
}

@techreport{Plotkin1981SOS,
  author = {Gordon D. Plotkin},
  title = {A Structural Approach to Operational Semantics},
  institution = {Aarhus University, Computer Science Department},
  number = {DAIMI FN-19},
  year = {1981},
  address = {Aarhus, Denmark},
  note = {Reprinted in Journal of Logic and Algebraic Programming, 2004},
  doi = {10.1016/j.jlap.2004.05.001},
}

@book{winskel1993formal,
  title = {The formal semantics of programming languages: an introduction},
  author = {Winskel, Glynn},
  year = {1993},
  publisher = {MIT press},
}

@article{LarsenSkou1991,
  author = {Kim G. Larsen and Arne Skou},
  title = {Bisimulation through Probabilistic Testing},
  journal = {Information and Computation},
  volume = {94},
  number = {1},
  pages = {1--28},
  year = {1991},
  doi = {10.1016/0890-5401(91)90030-6},
}

@inbook{Abramsky90,
  author = {Abramsky, Samson},
  title = {The lazy lambda calculus},
  year = {1990},
  isbn = {0201172364},
  publisher = {Addison-Wesley Longman Publishing Co., Inc.},
  address = {USA},
  booktitle = {Research Topics in Functional Programming},
  pages = {65–116},
  numpages = {52},
}

@book{Adamek_Milius_Moss_2025,
  place = {Cambridge},
  series = {Cambridge Tracts in Theoretical Computer Science},
  title = {Initial Algebras and Terminal Coalgebras: The Theory of Fixed Points
           of Functors},
  publisher = {Cambridge University Press},
  author = {Adámek, Jiří and Milius, Stefan and Moss, Lawrence S.},
  year = {2025},
  collection = {Cambridge Tracts in Theoretical Computer Science},
}

@book{AdamekRosicky1994,
  author = {Ji{\v r}{\'\i} Ad{\'a}mek and Ji{\v r}{\'\i} Rosick{\'y}},
  title = {Locally Presentable and Accessible Categories},
  series = {Cambridge Tracts in Mathematics},
  volume = {189},
  publisher = {Cambridge University Press},
  address = {Cambridge},
  year = {1994},
  isbn = {9780511600579},
}

@book{Awodey10,
  address = {USA},
  author = {Awodey, Steve},
  edition = {2nd},
  isbn = {0199237182},
  publisher = {Oxford University Press, Inc.},
  title = {Category Theory},
  year = {2010},
}

@article{Birkedal10,
  author = {Birkedal, Lars and St{\o}vring, Kristian and Thamsborg, Jacob},
  doi = {10.1016/j.tcs.2010.07.010},
  issue = {47},
  journal = {Theoretical Computer Science},
  language = {en},
  month = {10},
  pages = {4102--4122},
  publisher = {Elsevier BV},
  title = {The category-theoretic solution of recursive metric-space equations},
  url = {https://doi.org/10.1016/j.tcs.2010.07.010},
  volume = {411},
  year = {2010},
}

@article{Birkedal12,
  author = {Birkedal, Lars and M{\o}gelberg, Rasmus Ejlers and Schwinghammer,
            Jan and St{\o}vring, Kristian},
  date = {2012-10-03},
  doi = {10.2168/LMCS-8(4:1)2012},
  issn = {1860-5974},
  journal = {Logical Methods in Computer Science},
  language = {},
  publisher = {Episciences.org},
  title = {First steps in synthetic guarded domain theory: step-indexing in the
           topos of trees},
  url = {https://lmcs.episciences.org/1041},
  volume = {Volume 8, Issue 4},
}

@incollection{smyth90,
  author = {Smyth, M. B.},
  title = {Topology},
  year = {1993},
  publisher = {Oxford University Press, Inc.},
  booktitle = {Handbook of Logic in Computer Science (Vol. 1): Background:
               Mathematical Structures},
  pages = {641–761},
  numpages = {121},
}

@article{Frumin24,
  author = {Frumin, Dan and Timany, Amin and Birkedal, Lars},
  doi = {10.1145/3632854},
  issue = {POPL},
  journal = {Proceedings of the ACM on Programming Languages},
  language = {en},
  month = {1},
  pages = {332--361},
  publisher = {Association for Computing Machinery (ACM)},
  title = {Modular Denotational Semantics for Effects with Guarded Interaction
           Trees},
  volume = {8},
  year = {2024},
}

@article{Goncharov23,
  author = {Sergey Goncharov and Stefan Milius and Lutz Schr{\"{o}}der and
            Stelios Tsampas and Henning Urbat},
  title = {Towards a Higher-Order Mathematical Operational Semantics},
  journal = {Proceedings of the ACM on Programming Languages},
  volume = {7, Issue POPL},
  pages = {632--658},
  year = {2023},
  doi = {10.1145/3571215},
}

@inproceedings{Goncharov24,
  address = {Berlin, Heidelberg},
  author = {Goncharov, Sergey and Santamaria, Alessio and Schr{\"o}der, Lutz and
            Tsampas, Stelios and Urbat, Henning},
  booktitle = {Foundations of Software Science and Computation Structures: 27th
               International Conference, FoSSaCS 2024, Held as Part of the
               European Joint Conferences on Theory and Practice of Software,
               ETAPS 2024, Luxembourg City, Luxembourg, April 6\textendash 11,
               2024, Proceedings, Part II},
  doi = {10.1007/978-3-031-57231-9_3},
  isbn = {9783031572319},
  journal = {Lecture Notes in Computer Science},
  language = {en},
  month = {4},
  pages = {47--69},
  publisher = {Springer Nature Switzerland},
  title = {Logical Predicates in Higher-Order Mathematical Operational Semantics
           },
  url = {https://doi.org/10.1007/978-3-031-57231-9_3},
  year = {2024},
  editor = {Kobayashi, Naoki and Worrell, James},
  volume = {14575},
}

@article{Goncharov25,
  author = {Goncharov, Sergey and Partow, Pouya and Tsampas, Stelios},
  title = {Big Steps in Higher-Order Mathematical Operational Semantics},
  year = {2025},
  issue_date = {August 2025},
  publisher = {Association for Computing Machinery},
  address = {New York, NY, USA},
  volume = {9},
  number = {ICFP},
  url = {https://doi.org/10.1145/3747540},
  doi = {10.1145/3747540},
  journal = {Proc. ACM Program. Lang.},
  month = aug,
  articleno = {271},
  numpages = {29},
}

@book{Granas03,
  author = {Granas, Andrzej and Dugundji, James},
  doi = {10.1007/978-0-387-21593-8},
  isbn = {9780387215938},
  journal = {Springer Monographs in Mathematics},
  publisher = {Springer New York},
  address = {New York},
  title = {Fixed Point Theory},
  url = {https://doi.org/10.1007/978-0-387-21593-8},
  year = {2003},
}

@book{Hindley08,
  title = {Lambda-{{Calculus}} and {{Combinators}}: {{An Introduction}}},
  shorttitle = {Lambda-{{Calculus}} and {{Combinators}}},
  author = {Hindley, J. Roger and Seldin, Jonathan P.},
  year = {2008},
  edition = {2},
  publisher = {{Cambridge University Press}},
  OPTaddress = {{Cambridge}},
  doi = {10.1017/CBO9780511809835},
  isbn = {978-0-521-89885-0},
}

@book{Jacobs16,
  author = {Bart Jacobs},
  title = {{Introduction to Coalgebra: Towards Mathematics of States and
           Observation}},
  series = {Cambridge Tracts in Theoretical Computer Science},
  volume = {59},
  publisher = {Cambridge University Press},
  year = {2016},
  doi = {10.1017/CBO9781316823187},
  isbn = {9781316823187},
  bibsource = {dblp computer science bibliography, https://dblp.org},
}

@article{Kelly80,
  author = {Kelly, G.M.},
  doi = {10.1017/S0004972700006353},
  journal = {Bulletin of the Australian Mathematical Society},
  number = {1},
  pages = {1\textendash 83},
  title = {A unified treatment of transfinite constructions for free algebras,
           free monoids, colimits, associated sheaves, and so on},
  volume = {22},
  year = {1980},
}

@incollection{Milner75,
  title = {Processes: A Mathematical Model of Computing Agents},
  editor = {H.E. Rose and J.C. Shepherdson},
  series = {Studies in Logic and the Foundations of Mathematics},
  publisher = {Elsevier},
  volume = {80},
  pages = {157-173},
  year = {1975},
  booktitle = {Logic Colloquium '73},
  issn = {0049-237X},
  doi = {https://doi.org/10.1016/S0049-237X(08)71948-7},
  url = {https://www.sciencedirect.com/science/article/pii/S0049237X08719487},
  author = {Robin Milner},
  address = {North-Holland},
}

@inproceedings{Turi97,
  author = {Turi, Daniele and Plotkin, Gordon},
  booktitle = {Twelfth Annual IEEE Symposium on Logic in Computer Science},
  doi = {10.1109/lics.1997.614955},
  journal = {Proceedings of Twelfth Annual IEEE Symposium on Logic in Computer
             Science},
  pages = {280--291},
  publisher = {IEEE Comput. Soc},
  title = {Towards a mathematical operational semantics},
  url = {https://doi.org/10.1109/lics.1997.614955},
  venue = {Warsaw, Poland},
  year = {1997},
  volume = {1},
}

@article{Urbat26,
  author = {Urbat, Henning},
  journal = {Proc. ACM Program. Lang.},
  doi = {10.1145/3776697},
  publisher = {Association for Computing Machinery},
  volume = {10},
  number = {POPL},
  title = {Higher-Order Behavioural Conformances via Fibrations},
  year = {2026},
}

@article{Rutten2000,
  author = {Rutten, Jan J. M. M.},
  title = {Universal Coalgebra: A Theory of Systems},
  journal = {Theoretical Computer Science},
  volume = {249},
  number = {1},
  pages = {3--80},
  year = {2000},
  doi = {10.1016/S0304-3975(00)00056-6},
}

@inproceedings{GianantonioRPO,
  author = "Di Gianantonio, Pietro and Honsell, Furio and Lenisa, Marina",
  editor = "Amadio, Roberto",
  title = "RPO, Second-Order Contexts, and $\lambda$-Calculus",
  booktitle = "Foundations of Software Science and Computational Structures",
  year = "2008",
  publisher = "Springer Berlin Heidelberg",
  address = "Berlin, Heidelberg",
  pages = "334--349",
  doi = {10.1007/978-3-540-78499-9_24},
  isbn = "978-3-540-78499-9",
}

@inproceedings{UrbatTsampasEtAl23,
  author = {Henning Urbat and Stelios Tsampas and Sergey Goncharov and Stefan
            Milius and Lutz Schr{\"o}der},
  title = {Weak Similarity in Higher-Order Mathematical Operational Semantics},
  year = "2023",
  publisher = {IEEE Computer Society Press},
  booktitle = {38th Annual ACM/IEEE Symposium on Logic in Computer Science (LICS
               2023)},
  OPTeditor = {Igor Walukiewicz},
  doi = {10.1109/LICS56636.2023.10175706},
}

@article{Worrell05,
  author = {James Worrell},
  title = {On the final sequence of a finitary set functor},
  journal = {Theoretical Computer Science},
  volume = {338},
  number = {1--3},
  pages = {184--199},
  year = {2005},
  doi = {10.1016/j.tcs.2004.12.009},
}

@article{Goncharov25b,
  author = {Goncharov, Sergey and Milius, Stefan and Schr\"{o}der, Lutz and
            Tsampas, Stelios and Urbat, Henning},
  title = {Bialgebraic Reasoning on Stateful Languages},
  year = {2025},
  issue_date = {August 2025},
  publisher = {Association for Computing Machinery},
  address = {New York, NY, USA},
  volume = {9},
  number = {ICFP},
  doi = {10.1145/3747513},
  articleno = {244},
  numpages = {30},
}

@inproceedings{gmtu24lics,
  author = {Sergey Goncharov and Stefan Milius and Stelios Tsampas and Henning
            Urbat},
  title = {Bialgebraic Reasoning on Higher-Order Program Equivalence},
  year = "2024",
  publisher = {IEEE Computer Society Press},
  pages = {39:1-39:15},
  booktitle = {39th Annual ACM/IEEE Symposium on Logic in Computer Science (LICS
               2024)},
  doi = {10.1145/3661814.3662099},
  OPTnote = {Preprint: \url{https://arxiv.org/abs/2402.00625}},
}

@article{10.1145/3704871,
  author = {Goncharov, Sergey and Tsampas, Stelios and Urbat, Henning},
  title = {Abstract Operational Methods for Call-by-Push-Value},
  year = {2025},
  issue_date = {January 2025},
  publisher = {Association for Computing Machinery},
  address = {New York, NY, USA},
  volume = {9},
  number = {POPL},
  url = {https://doi.org/10.1145/3704871},
  doi = {10.1145/3704871},
  journal = {Proc. ACM Program. Lang.},
  month = jan,
  articleno = {35},
  numpages = {27},
}

@inproceedings{DBLP:conf/lics/FiorePT99,
  author = {Marcelo P. Fiore and Gordon D. Plotkin and Daniele Turi},
  title = {Abstract Syntax and Variable Binding},
  booktitle = {14th Annual {IEEE} Symposium on Logic in Computer Science (LICS
               1999)},
  pages = {193--202},
  publisher = {{IEEE} Computer Society},
  year = {1999},
  doi = {10.1109/LICS.1999.782615},
  bibsource = {dblp computer science bibliography, https://dblp.org},
}

@misc{lenke2026unifiedtreatmentsubstitutionpresheaves,
  title = {A Unified Treatment of Substitution for Presheaves, Nominal Sets,
           Renaming Sets, and so on},
  author = {Fabian Lenke and Stefan Milius and Henning Urbat},
  year = {2026},
  eprint = {2602.11907},
  archivePrefix = {arXiv},
  primaryClass = {cs.LO},
  url = {https://arxiv.org/abs/2602.11907},
  howpublished = {Accepted at the 41th Annual ACM/IEEE Symposium on Logic in
                  Computer Science (LICS 2026)},
}

@misc{Goncharov2026arXiv,
  title = {Towards a Higher-Order Bialgebraic Denotational Semantics},
  author = {Sergey Goncharov and Marco Peressotti and Stelios Tsampas and
            Henning Urbat and Stefano Volpe},
  year = {2026},
  eprint = {2602.18295},
  archivePrefix = {arXiv},
  primaryClass = {cs.PL},
  url = {https://arxiv.org/abs/2602.18295},
}

\end{document}
\endinput